\documentclass[a4paper,reqno,10pt]{amsart}

\usepackage[
    top    = 2.0cm,
    bottom = 2.0cm,
    left   = 2.0cm,
    right  = 2.0cm]{geometry}
\usepackage[english]{babel}
\usepackage{amsmath}
\usepackage{graphicx}
\usepackage{amssymb}
\usepackage{bm} 

\usepackage{amstext}
\usepackage{array}
\usepackage{setspace,cite,enumitem}
\usepackage[capitalise,noabbrev]{cleveref}
\usepackage[table]{xcolor}

\usepackage{tikz}
\usetikzlibrary{shapes}
\usetikzlibrary{arrows,positioning}

\tikzstyle{main node}=[circle,fill=black, inner sep = 2pt, minimum size=5pt]
\tikzstyle{csv}=[circle,fill=black, inner sep = 2pt, minimum size=5pt]
\tikzstyle{sv}=[circle,fill=black, inner sep = 2pt, minimum size=5pt]
\tikzstyle{subsv} =[circle,fill=black, double, inner sep = 2pt, minimum size=5pt]
\tikzstyle{cross}=[cross out,-,draw=black, inner sep = 0pt, minimum size=5pt]
\usepackage{caption}

\numberwithin{equation}{section}

\newcolumntype{C}{>{$}c<{$}} 
\newcolumntype{L}[1]{>{\raggedright}m{#1}}
\newcolumntype{D}[1]{>{\centering\arraybackslash\vspace{2mm}}m{#1}<{\vspace{2mm}}}
\newcolumntype{R}[1]{>{\raggedleft}m{#1}}

\newenvironment{amatrix}[1]{\left( \begin{array}{@{}#1@{}}}{\end{array} \right)}

\let\originalleft\left     
\let\originalright\right
\renewcommand{\left}{\mathopen{}\mathclose\bgroup\originalleft}
\renewcommand{\right}{\aftergroup\egroup\originalright}

\allowdisplaybreaks

\setitemize{leftmargin=*}   
\setenumerate{leftmargin=*} 

\newcommand{\BKL}{\cellcolor[gray]{0.8}} 
\newcommand{\IKL}{\cellcolor[gray]{0.6}} 

\newcommand{\alg}[1]{\mathfrak{#1}} 

\newcommand{\func}[2]{#1 \left( #2 \right)} 
\newcommand{\tfunc}[2]{#1 \bigl( #2 \bigr)} 

\newcommand{\brac}[1]{\left( #1 \right)}
\newcommand{\tbrac}[1]{\bigl( #1 \bigr)}
\newcommand{\sqbrac}[1]{\left[ #1 \right]}
\newcommand{\tsqbrac}[1]{\bigl[ #1 \bigr]}
\newcommand{\set}[1]{\left\{ #1 \right\}}

\newcommand{\st}{\mspace{5mu} : \mspace{5mu}} 

\newcommand{\abs}[1]{\left\lvert #1 \right\rvert}

\newcommand{\ZZ}{\mathbb{Z}}

\newcommand{\QQ}{\mathbb{Q}}
\newcommand{\RR}{\mathbb{R}}

\newcommand{\CC}{\mathbb{C}}

\newcommand{\pd}{\partial}         

\newcommand{\dd}{\mathrm{d}}   
\newcommand{\ii}{\mathfrak{i}} 
\newcommand{\ee}{\mathsf{e}}   

\newcommand{\wun}{\mathbf{1}}  

\newcommand{\inner}[2]{\left\langle #1 , #2 \right\rangle} 

\newcommand{\normord}[1]{\mbox{${} : #1 : {}$}} 

\newcommand{\comm}[2]{\bigl[ #1 , #2 \bigr]}
\newcommand{\acomm}[2]{\bigl\{ #1 , #2 \bigr\}}

\newcommand{\ideal}[1]{\left\langle #1 \right\rangle}

\newcommand{\ChebyPolyT}[1]{T_{#1}}
\newcommand{\ChebyPolyU}[1]{U_{#1}}
\newcommand{\ChebyT}[2]{\tfunc{\ChebyPolyT{#1}}{#2}} 
\newcommand{\ChebyU}[2]{\tfunc{\ChebyPolyU{#1}}{#2}} 

\newcommand{\ra}{\rightarrow}
\newcommand{\Ra}{\Rightarrow}
\newcommand{\lra}{\longrightarrow}

\newcommand{\affine}[1]{\widehat{#1}}
\newcommand{\SLA}[2]{\alg{#1} \bigl( #2 \bigr)}                             
\newcommand{\AKMA}[2]{\affine{\alg{#1}} \left( #2 \right)}                  
\newcommand{\AKMSA}[3]{\affine{\alg{#1}} \left( #2 \middle\vert #3 \right)} 

\newcommand{\Ver}[1]{\mathcal{V}_{#1}}       
\newcommand{\Irr}[1]{\mathcal{L}_{#1}}       
\newcommand{\Kac}[1]{\mathcal{K}_{#1}}       
\newcommand{\Fock}[1]{\mathcal{F}_{#1}}      
\newcommand{\Stag}[2]{\mathcal{S}_{#1}^{#2}} 

\newcommand{\logcoup}[2]{\beta_{#1}^{#2}}    

\newcommand{\spsub}[1]{#1^{\text{ss}}}       

\newcommand{\chmap}{\mathrm{ch}}

\newcommand{\Gr}[1]{\bigl[ #1 \bigr]}            

\newcommand{\ch}[1]{\chmap \Gr{#1}}              
\newcommand{\fch}[2]{\ch{#1} \bigl( #2 \bigr)}   

\newcommand{\jth}[1]{\vartheta_{#1}}             
\newcommand{\fjth}[2]{\jth{#1} \bigl( #2 \bigr)} 

\newcommand{\modS}{\mathsf{S}} 

\newcommand{\Smat}[2]{\modS \bigl[ #1 \ra #2 \bigr]}  

\newcommand{\fuse}{\mathbin{\times}}                                            
\newcommand{\Grfuse}{\mathbin{\boxtimes}}                                       
\newcommand{\fuscoeff}[3]{\mathsf{N}_{#1 \, #2}^{\hphantom{#1 \, #2} #3}}       

\newcommand{\coproductsymb}{\Delta}                                                
\newcommand{\coproduct}[1]{\coproductsymb \bigl( #1 \bigr)}                        
\newcommand{\Ncoproductsymb}[1]{\coproductsymb^{(#1)}}                             
\newcommand{\parNcoproductsymb}[2]{\Ncoproductsymb{#1}_{#2}}                       
\newcommand{\parcoproduct}[2]{\coproductsymb_{#1} \bigl( #2 \bigr)}                
\newcommand{\parNcoproduct}[3]{\Ncoproductsymb{#1}_{#2} \bigl( #3 \bigr)}          

\newcommand{\ses}[3]{0 \ra #1 \ra #2 \ra #3 \ra 0}                                  
\newcommand{\dses}[5]{0 \lra #1 \overset{#2}{\lra} #3 \overset{#4}{\lra} #5 \lra 0} 

\newcommand{\cft}{conformal field theory}
\newcommand{\cfts}{conformal field theories}
\newcommand{\uea}{universal enveloping algebra}
\newcommand{\lcft}{logarithmic conformal field theory}
\newcommand{\lcfts}{logarithmic conformal field theories}

\newcommand{\ope}{operator product expansion}
\newcommand{\opes}{operator product expansions}
\newcommand{\hw}{highest-weight}
\newcommand{\Hw}{Highest-weight}
\newcommand{\hws}{\hw{} vector}
\newcommand{\hwss}{\hw{} vectors}
\newcommand{\sv}{singular vector}
\newcommand{\svs}{singular vectors}
\newcommand{\ssv}{subsingular vector}
\newcommand{\ssvs}{subsingular vectors}
\newcommand{\hwm}{\hw{} module}
\newcommand{\hwms}{\hw{} modules}

\newcommand{\voa}{vertex operator algebra}

\newcommand{\NGK}{Nahm-Gaberdiel-Kausch}
\newcommand{\lhs}{left-hand side}
\newcommand{\rhs}{right-hand side}
\newcommand{\ns}{Neveu-Schwarz}

\newcommand{\eps}{\varepsilon}

\renewcommand{\Im}{\operatorname{Im}}

\newcommand{\qplus}{\overset{\text{?}}{\oplus}}

\DeclareMathOperator{\id}{id}
\DeclareMathOperator{\vspn}{span}

\renewcommand{\ge}{\geqslant}
\renewcommand{\le}{\leqslant}

\makeatletter 
\def\@endtheorem{\endtrivlist}
\makeatother 

\theoremstyle{plain}
\newtheorem{thm}{Theorem}[section]
\newtheorem{prop}[thm]{Proposition}

\begin{document}

\title[Fusion Rules for the $N=1$ Neveu-Schwarz Algebra]{Fusion rules for the logarithmic $\bm{N=1}$ superconformal \\ minimal models I:  the Neveu-Schwarz sector}

\author[M Canagasabey]{Michael Canagasabey}

\address[Michael Canagasabey]{
Mathematical Sciences Institute \\
Australian National University \\
Acton, ACT 2601 \\
Australia
}

\email{nishan.canagasabey@anu.edu.au}

\author[J Rasmussen]{J\o{}rgen Rasmussen}

\address[Jorgen Rasmussen]{
School of Mathematics and Physics \\
University of Queensland \\
St Lucia, Queensland 4072 \\
Australia
}

\email{j.rasmussen\;\!@\;\!uq.edu.au}

\author[D Ridout]{David Ridout}

\address[David Ridout]{
Department of Theoretical Physics \\
Research School of Physics and Engineering;
and
Mathematical Sciences Institute;
Australian National University \\
Acton, ACT 2601 \\
Australia
}

\email{david.ridout@anu.edu.au}

\thanks{\today}

\begin{abstract}
It is now well known that non-local observables in critical statistical lattice models, polymers and percolation for example, may be modelled in the continuum scaling limit by logarithmic conformal field theories.  Fusion rules for such theories, sometimes referred to as logarithmic minimal models, have been intensively studied over the last ten years in order to explore the representation-theoretic structures relevant to non-local observables.  Motivated by recent lattice conjectures, this work studies the fusion rules of the $N=1$ supersymmetric analogues of these logarithmic minimal models in the Neveu-Schwarz sector.  Fusion rules involving Ramond representations will be addressed in a sequel.
\end{abstract}

\maketitle

\onehalfspacing

\section{Introduction} \label{sec:Intro}

Superconformal algebras have a long history in mathematical physics, being intertwined with the development of string theory through their role as infinitesimal symmetries of superstrings (see \cite{FriCon86}, for example), as well as appearing as extended symmetries of the scaling limits of certain lattice models \cite{FriSup85}.  These infinite-dimensional Lie superalgebras each contain a Virasoro subalgebra that quantises the infinitesimal conformal symmetries of the plane.  Field-theoretically, the superconformal algebras are customarily parametrised by the number $N$ of fermionic partners of the energy-momentum tensor.  The simplest examples, after the Virasoro algebra itself (corresponding to $N=0$) are the $N=1$ superconformal algebras:  the \emph{\ns{} algebra} \cite{NevFac71} and the \emph{Ramond algebra} \cite{RamDua71}.

After this debut in superstring theory and statistical mechanics, the structure theory for Verma modules over the \ns{} algebra was quickly settled.  A determinant formula for their invariant bilinear forms, originally conjectured in \cite{KacCon79}, was proven in \cite{MeuHig86} and the possible submodule structures were elucidated in \cite{AstStr97} with the result essentially repeating that of the Virasoro algebra.  In particular, all non-trivial submodules are generated by \svs{}, the dimension of the space of \svs{} in any weight space is at most one, and non-trivial homomorphisms between Verma modules are necessarily injective.  

By way of contrast, the structures of the Ramond Verma modules can be more intricate \cite{IohRepI03}:  submodules not generated by \svs{} can exist, there can be up to two linearly independent \svs{} of each parity in a given weight space, and there can exist non-trivial non-injective homomorphisms.  These features of the Ramond case are now relatively well understood, but serve as a useful toy model for the more pronounced difficulties that one encounters when investigating the Verma modules of the $N>1$ superconformal algebras.  Despite the fact that the latter algebras have fundamental applications to mirror symmetry \cite{GreDua90}, the AdS/CFT correspondence \cite{BeiRev12} and Mathieu moonshine \cite{EguNot11}, it is fair to say that their representation theories remain poorly understood at best.

Our purpose with this article is to explore some of the representation-theoretic aspects of the \ns{} algebra that pertain to \lcft{}.  The more challenging exploration of the representations of the Ramond algebra will be addressed in a sequel \cite{CanFusII15}.  Here, the qualifier ``logarithmic'' means that the underlying vertex operator (super)algebra admits modules upon which the Virasoro zero mode $L_0$ acts non-semisimply.  A collection of reviews on this topic may be found in \cite{RidLog13}.  The results presented here are motivated by a recent lattice-theoretic study, reported in \cite{PeaLog14}, although \lcfts{} with supersymmetry have been discussed in the past, see \cite{KhoLog98,MavNev03,RasLog04,NagLog05,AdaMil09} for example.  This study proposes a conjecture, based on numerical evidence, that the scaling limits of a certain collection of integrable lattice models are described by logarithmic analogues of the $N=1$ superconformal minimal models.  Here, these logarithmic superconformal models are explored directly in the continuum, as \cfts{}.

In particular, we study the fusion rules of a certain collection of \ns{} modules that we will refer to as (\ns{}) \emph{Kac modules}.  This usage follows the nomenclature established in \cite{RasCla11,BusKaz12,MorKac15} for Kac modules over the Virasoro algebra.  The physical relevance of these modules is that they are believed \cite{PeaLog06,PeaLog14} to identify the scaling limits corresponding to an accessible class of integrable boundary conditions for the underlying lattice models.  Because of this, the continuum fusion rules for Virasoro Kac modules have been well studied both via lattice approximations \cite{PR07,ReaAss07,RasFus07,RasFus07b} and direct calculations \cite{GabInd96,EbeVir06,RidPer07,RidLog07,RidPer08,GabFus09}.  In contrast, the fusion rules of the $N=1$ Kac modules only seem to have been studied from the lattice point of view \cite{PeaLog14}, although those of the simple (irreducible) modules appearing in the $N=1$ minimal models have received continuum treatments, see \cite{EicMin85,SotSta86,GabFus97,IohFus01} for instance.

We begin, in \cref{sec:Back}, by recalling the $N=1$ \ns{} algebra and reviewing the necessary structure theory of its \hwms{} and Fock spaces.  The latter are the superconformal analogues of the well known Feigin-Fuchs modules that appear in the Coulomb gas free field realisation of the Virasoro minimal models.  Their structure theory is detailed in \cite{IohRepII03}.  This theory is important for the present investigation because we define the $N=1$ \ns{} Kac modules as certain submodules of Fock spaces, following \cite{RasCla11,BusKaz12,MorKac15} in the Virasoro case.

When a \cft{} is rational, meaning that the relevant modules of the underlying \voa{} are semisimple and that only finitely many simple modules appear, the fusion rules may be efficiently computed using the Verlinde formula \cite{VerFus88}.  While rigorous proofs seem to require rationality, see \cite{HuaVer05}, a lesson learned from more physical proofs \cite{MooPol88} is that this formula follows from the deeper consistency requirements of \cft{} and so should be valid, in some form, more generally.  In \cref{sec:Verlinde}, we utilise a continuous version of the Verlinde formula to compute the characters of the fusion products of the Kac modules.  This approach was pioneered for logarithmic theories in \cite{CreRel11} and is now referred to as the standard module formalism.  It has since been shown to yield correct (or, at least, sensible) results rather generally, see \cite{CreMod12,BabTak12,CreMod13,RidMod14,RidBos14,MorKac15}.  We refer to \cite{CreLog13,RidVer14} in which general features of the standard module formalism are discussed.

Armed with this character information, we turn to the detailed structure of the fusion product of two Kac modules.  The additional structural data may be obtained, at least in certain examples, by explicitly constructing (a truncation of) the fusion product using the \NGK{} fusion algorithm \cite{NahQua94,GabInd96}.  In \cref{sec:TheExample}, we detail an explicit example in order to illustrate how one typically employs this algorithm, noting that \emph{a priori} knowledge of the fusion product's character can lead to significant simplifications.  We also use this example to emphasise that completely identifying the result, up to isomorphism, may require computing additional indecomposability parameters \cite{GabInd96,GurCTh99}, in particular the $N=1$ analogues of the logarithmic couplings of \cite{RidPer07}.  Moreover, we demonstrate that structural theorems for more general classes of \ns{} modules may also be used to significantly simplify the identification of the fusion product.

\cref{sec:Results} then summarises the fusion rules that we have obtained for \ns{} Kac modules by combining the \NGK{} fusion algorithm with the Verlinde formula and \ns{} structural theorems.  An infinite family of fusion rules are subsequently conjectured and encoded in polynomial ring structures.  We conclude with a brief discussion that looks toward generalising these results to the Ramond sector and outlines how the formalism being developed will contribute to our long term research goals.  This is followed by three appendices.  \cref{app:Fusion} provides a detailed derivation of the coproduct formulae that underlie the \NGK{} fusion algorithm and a description of the usage of the algorithm itself.  While this material roughly follows \cite{GabFus94b,GabInd96}, it makes the explicit computations in \cref{sec:TheExample} self-contained as well as prepares the reader for the rather more involved generalisation, detailed in \cite{CanFusII15}, needed to discuss the Ramond sector.  \cref{app:StagMod} then introduces the \ns{} \emph{staggered modules}, following \cite{RohRed96,RidSta09,CreLog13} and \cite{PeaLog14}, that appear in certain fusion products.  We prove a few basic structural results before reviewing the definition of logarithmic couplings and explaining how these parameters may be computed.  \cref{app:Results} contains a selection of fusion products computed by combining the information provided by the Verlinde formula, the Nahm-Gaberdiel-Kausch fusion algorithm and the theory of staggered modules.

\section{Background and conventions} \label{sec:Back}

The $N=1$ superconformal algebras may be defined as the Lie superalgebras spanned by the even (bosonic) modes $L_n$ and $C$, and the odd (fermionic) modes $G_k$, subject to the commutation relations
\begin{equation} \label{eq:CommN=1}
\begin{aligned}
\comm{L_m}{L_n} &= \brac{m-n} L_{m+n} + \frac{m^3-m}{12} \delta_{m+n=0} \: C, & \comm{L_m}{G_k} &= \brac{\frac{1}{2} m - k} G_{m+k}, \\
\acomm{G_j}{G_k} &= 2 L_{j+k} + \frac{4j^2-1}{12} \delta_{j+k=0} \: C, & \comm{L_m}{C} &= \comm{G_j}{C} = 0.
\end{aligned}
\end{equation}
In \eqref{eq:CommN=1}, we will take $m,n \in \ZZ$ and $j,k \in \ZZ + \frac{1}{2}$, stipulating that we are studying the \emph{\ns{} algebra}.  Taking $j,k \in \ZZ$ results in the \emph{Ramond algebra} instead.  The central element $C$ will be taken to act in all representations as a fixed multiple $c$ of the identity operator, called the central charge.  Formally, we thus consider the quotient of the universal enveloping algebra of the \ns{} algebra by the ideal generated by 
$C - c \, \id$; we will also refer to this quotient as the \ns{} algebra.

In field-theoretic terms, the \ns{} superalgebra extends the Virasoro algebra by the modes of a fermionic primary field of conformal weight $\frac{3}{2}$.  More precisely, the fields generated by the bosonic and fermionic modes are, respectively, the energy-momentum tensor $\func{T}{z}$ and its superpartner $\func{G}{z}$:
\begin{equation}
\func{T}{z} = \sum_{n \in \ZZ} L_n z^{-n-2}, \qquad 
\func{G}{z} = \sum_{j \in \ZZ + 1/2} G_j z^{-j-3/2}.
\end{equation}
The \opes{} equivalent to \eqref{eq:CommN=1} then take the form
\begin{equation} \label{OPE:TTTGGG}
\begin{gathered}
\func{T}{z} \func{T}{w} \sim \frac{c/2}{\brac{z-w}^4} + \frac{2 \: \func{T}{w}}{\brac{z-w}^2} + \frac{\pd \func{T}{w}}{z-w}, \\
\func{T}{z} \func{G}{w} \sim \frac{3 \: \func{G}{w} / 2}{\brac{z-w}^2} + \frac{\pd \func{G}{w}}{z-w}, \qquad 
\func{G}{z} \func{G}{w} \sim \frac{2c/3}{\brac{z-w}^3} + \frac{2 \: \func{T}{w}}{z-w},
\end{gathered}
\end{equation}
supplemented by the locality condition $\func{T}{z} \func{G}{w} = \func{G}{w} \func{T}{z}$.

\subsection{\Hw{} modules} \label{sec:hwms}

\ns{} \hw{} theory works as one would expect.  The triangular decomposition splits the superalgebra into the span of the positive modes $L_n$ and $G_j$, with $n,j>0$, the negative modes $L_n$ and $G_j$, with $n,j<0$, and the zero modes $L_0$ and 
$C = c \, \id$.  A \emph{\hws{}} $v_h$ is therefore characterised by its conformal weight $h$ (we regard the central charge of the 
module as implicitly fixed) and satisfies
\begin{equation}
L_n v_h = G_j v_h = 0 \quad \text{for \(n,j>0\);} \qquad L_0 v_h = h v_h.
\end{equation}
The \emph{Verma module} $\Ver{h}$ is then constructed from $v_h$ as an induced module and it has a unique simple quotient that we shall denote by $\Irr{h}$.

These \ns{} modules are naturally $\ZZ_2$-graded by choosing the parity of the \hws{} to be even or odd.  Structurally, this choice makes no difference, but it is sometimes useful to keep the parity explicit.  When this is the case, we affix a superscript sign $\pm$ to the module, with the sign matching the parity chosen for $v_h$.  We will generalise this convention to all the indecomposable \ns{} modules considered here by matching the superscript sign to the (common) parity of the vectors of minimal conformal dimension.

The standard parametrisation suggested by the \ns{} analogue of the Kac determinant formula is
\begin{equation}
c = \frac{15}{2} - 3 \brac{t+t^{-1}}, \qquad 
h_{r,s} = \frac{r^2-1}{8} t^{-1} - \frac{rs-1}{4} + \frac{s^2-1}{8} t,
\end{equation}
where $t \in \CC \setminus \set{0}$.  The Verma module $\Ver{r,s} \equiv \Ver{h_{r,s}}$ is then reducible when $r$ and $s$ are positive integers satisfying $r=s \bmod{2}$.\footnote{In contrast, Verma modules for the Ramond algebra turn out to be reducible for positive integers $r$ and $s$ satisfying $r \neq s \bmod{2}$.  However, the parametrisation for $h_{r,s}$ is also corrected in the Ramond sector by adding $\frac{1}{16}$.}  If $t$ is rational, then this parametrisation may be written in the form
\begin{equation} \label{eq:ParByt}
t = \frac{p}{p'}, \qquad 
c = \frac{3}{2} \brac{1 - \frac{2 \brac{p'-p}^2}{pp'}}, \qquad 
h_{r,s} = \frac{\brac{p'r-ps}^2 - \brac{p'-p}^2}{8pp'},
\end{equation}
where one customarily takes $p=p' \bmod{2}$ and $\gcd \set{p, \frac{1}{2} \brac{p'-p}} = 1$.  The $N=1$ superconformal minimal models \cite{EicMin85,BerSup85,FriSup85} are built from the \ns{} simple \hwms{} $\Irr{r,s} \equiv \Irr{h_{r,s}}$ with $1 \le r \le p-1$, $1 \le s \le p'-1$ and $r=s \bmod{2}$, as well as their Ramond counterparts.  However, we are not studying these minimal models, so we do not insist, for example, that $p,p' \ge 2$.

In most respects, the \hw{} theory for the \ns{} algebra parallels that of the Virasoro algebra.  In particular, the submodules of a Verma module are generated by \svs{} and the maximal dimension of the space of \svs{} of any given conformal weight is $1$.  The submodule structure of a Verma module $\Ver{h}$ then reduces to knowing its \svs{} and here the possibilities exactly mirror those of the Virasoro Verma modules \cite{AstStr97,IohRepI03}.  We reproduce the structures diagrammatically in \cref{fig:VermaStructures}.  As mentioned above, if $h \neq h_{r,s}$ for any positive integers $r$ and $s$ with $r=s \bmod{2}$, then $\Ver{h}$ is simple.  On the other hand, if $t$ is irrational, then the maximal proper submodule of the $\Ver{h}$ with $h=h_{r,s}$, $r=s \bmod{2}$, is always simple.  It is generated by a \sv{} of \emph{depth} $\frac{1}{2} rs$, meaning that its conformal weight is $h_{r,s} + \frac{1}{2} rs$.

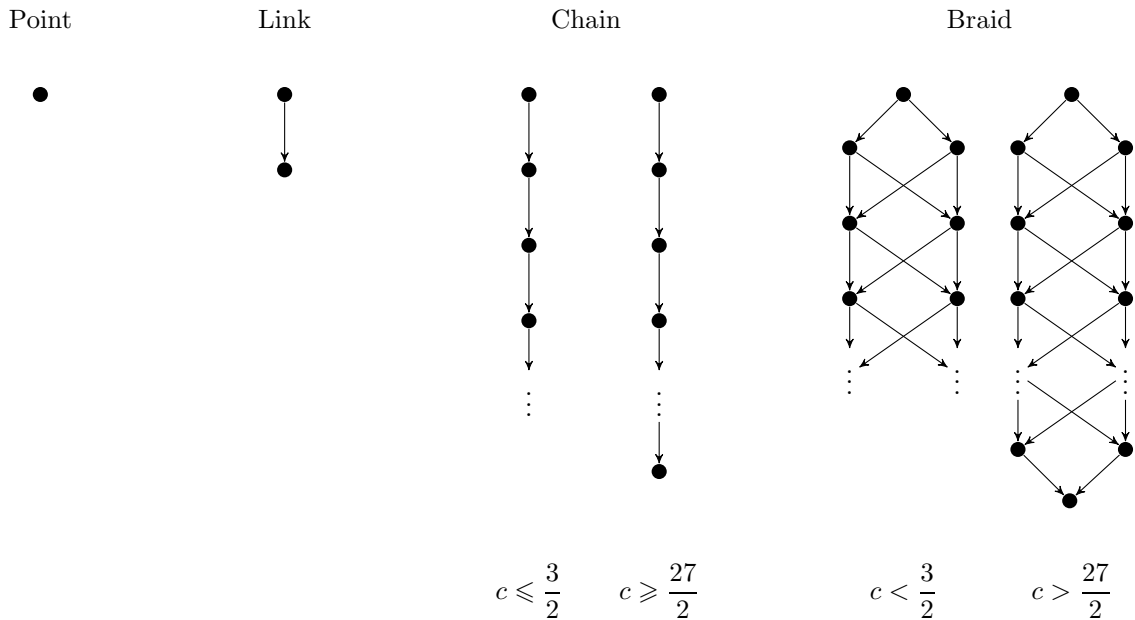
\begin{figure}
\begin{tikzpicture}[->,>=stealth', node distance=1cm]
  \node[main node] (1) [] {};
  \node[] (point) [above of =1] {Point};
  \node[main node] (2) [right = 3cm of 1] {};
  \node[] (link) [above of =2] {Link};
  \node[main node] (2a) [below of = 2] {};
  \path[]
  (2) edge node {} (2a);
  \node[main node] (3) [right = 3cm of 2] {};
  \node[] (chain) [right = 0.75cm of 3,above of =3] {Chain};
  \node[main node] (3a) [below of =3] {};
  \node[main node] (3b) [below of =3a] {};
  \node[main node] (3c) [below of =3b] {};
  \node[inner sep = 2pt] (3d) [below of =3c] {$\vdots$};
  \node[] (3bot) [below =6cm of 3] {$c\le\dfrac{3}{2}$};
  \path[]
  (3) edge node {} (3a)
  (3a) edge node {} (3b)
  (3b) edge node {} (3c)
  (3c) edge node {} (3d);
  \node[main node] (4) [right = 1.5cm of 3] {};
  \node[main node] (4a) [below of =4] {};
  \node[main node] (4b) [below of =4a] {};
  \node[main node] (4c) [below of =4b] {};
  \node[inner sep = 2pt] (4d) [below of =4c] {$\vdots$};
  \node[main node] (4e) [below of =4d] {};
  \node[] (4bot) [below =6cm of 4] {$c\ge\dfrac{27}{2}$};
  \path[]
  (4) edge node {} (4a)
  (4a) edge node {} (4b)
  (4b) edge node {} (4c)
  (4c) edge node {} (4d)
  (4d) edge node {} (4e);	
  \node[main node] (5) [right = 3cm of 4] {};
  \node[] (braid) [right = 1cm of 5,above of =5] {Braid};
  \node[main node] (5a) [below left of =5] {};
  \node[main node] (5b) [below of =5a] {};
  \node[main node] (5c) [below of =5b] {};
  \node[inner sep = 2pt] (5d) [below of =5c] {$\vdots$};
  \node[main node] (5e) [below right of =5] {};
  \node[main node] (5f) [below of =5e] {};
  \node[main node] (5g) [below of =5f] {};
  \node[inner sep = 2pt] (5h) [below of =5g] {$\vdots$};
  \node[] (5bot) [below =6cm of 5] {$c < \dfrac{3}{2}$};
  \path[]
  (5) edge node {} (5a)
  (5a) edge node {} (5b)
  (5a) edge node {} (5f)
  (5b) edge node {} (5c)
  (5b) edge node {} (5g)
  (5c) edge node {} (5d)
  (5c) edge node {} (5h)
  (5) edge node {} (5e)
  (5e) edge node {} (5f)
  (5e) edge node {} (5b)
  (5f) edge node {} (5g)
  (5f) edge node {} (5c)
  (5g) edge node {} (5h)
  (5g) edge node {} (5d);	
  \node[main node] (6) [right = 2cm of 5] {};
  \node[main node] (6a) [below left of =6] {};
  \node[main node] (6b) [below of =6a] {};
  \node[main node] (6c) [below of =6b] {};
  \node[inner sep = 2pt] (6d) [below of =6c] {$\vdots$};
  \node[main node] (6d1) [below of =6d] {};
  \node[main node] (6d2) [below right = 0.75cm of 6d1] {};
  \node[main node] (6e) [below right of =6] {};
  \node[main node] (6f) [below of =6e] {};
  \node[main node] (6g) [below of =6f] {};
  \node[inner sep = 2pt] (6h) [below of =6g] {$\vdots$};
  \node[main node] (6h1) [below of =6h] {};
  \node[] (6bot) [below =6cm of 6] {$c > \dfrac{27}{2}$};
  \path[]
  (6) edge node {} (6a)
  (6a) edge node {} (6b)
  (6a) edge node {} (6f)
  (6b) edge node {} (6c)
  (6b) edge node {} (6g)
  (6c) edge node {} (6d)
  (6c) edge node {} (6h)
  (6d) edge node {} (6d1)
  (6d) edge node {} (6h1)
  (6d1) edge node {} (6d2)
  (6) edge node {} (6e)
  (6e) edge node {} (6f)
  (6e) edge node {} (6b)
  (6f) edge node {} (6g)
  (6f) edge node {} (6c)
  (6g) edge node {} (6h)
  (6g) edge node {} (6d)
  (6h) edge node {} (6h1)
  (6h) edge node {} (6d1)
  (6h1) edge node {} (6d2);	
\end{tikzpicture}
\caption{The \sv{} structure, marked by black circles, of \ns{} Verma modules.  Arrows from one \sv{} to another indicate that the latter may be obtained from the former by acting with a suitable polynomial in the $L_n$ and $G_j$.  Note that $t > 0$ corresponds to $c \le \tfrac{3}{2}$ and $t < 0$ corresponds to $c \ge \tfrac{27}{2}$.} \label{fig:VermaStructures}
\end{figure}

We will exclusively focus on the case in which $t=p/p'$ is rational and positive, with $h = h_{r,s}$ and $r=s \bmod{2}$.  If $r$ is a multiple of $p$, or $s$ is a multiple of $p'$, then the \sv{} structure of $\Ver{h}$ is represented by the infinite chain diagram in \cref{fig:VermaStructures}.  Otherwise, the structure corresponds to the infinite braid diagram; this latter case is the one relevant to the study of minimal models.  In both cases, a \sv{} is always present at depth $\frac{1}{2} rs$, though there may be other \svs{} at other depths.

All of this information may be conveniently summarised in the \ns{} analogue of the extended Kac table, see \cref{fig:KacTables}.  This simply tabulates the values of $h_{r,s}$ as $r$ and $s$ run through the positive integers, subject to $r=s \bmod{2}$.\footnote{The ``gaps'' in the table, when $r \neq s \bmod{2}$, correspond to conformal weights of representations of the Ramond algebra.}  To make contact with the above structural results, we partition the extended Kac table into three subsets as follows:
\begin{itemize}
\item If $p$ divides $r$ and $p'$ divides $s$, then we say that $(r,s)$ is of \emph{corner} type in the extended Kac table.
\item If $p$ divides $r$ or $p'$ divides $s$, but not both, then $(r,s)$ is said to be of \emph{boundary} type.
\item If $p$ does not divide $r$ and $p'$ does not divide $s$, then $(r,s)$ is said to be of \emph{interior} type.
\end{itemize}
Summarising, corner and boundary type Verma modules have singular vectors arranged in chains whereas interior type Verma modules have a braided pattern of singular vectors.  We remark that when $p=1$ or $p'=1$, there are no interior entries in the extended Kac table, and if $p=p'=1$, then there will be no boundary entries either.  We illustrate this with three pertinent examples of extended \ns{} Kac tables in \cref{fig:KacTables}.

{
\renewcommand{\arraystretch}{1.1}
\begin{figure}[p]
\begin{center}
\begin{tikzpicture}
\node (Kac1) at (0,0) {
\setlength{\extrarowheight}{4pt}
\begin{tabular}{|C|C|C|C|C|C|C|C|C|C|C|C|C}
\hline
0 &  & \frac{1}{2} &  & 2 &  & \frac{9}{2} &  & 8 &  & \frac{25}{2} &  & \cdots \\[1mm]
\hline
 & 0 &  & \frac{1}{2} &  & 2 &  & \frac{9}{2} &  & 8 &  & \frac{25}{2} & \cdots \\[1mm]
\hline
\frac{1}{2} &  & 0 &  & \frac{1}{2} &  & 2 &  & \frac{9}{2} &  & 8 &  & \cdots \\[1mm]
\hline
 & \frac{1}{2} &  & 0 &  & \frac{1}{2} &  & 2 &  & \frac{9}{2} &  & 8 & \cdots \\[1mm]
\hline
2 &  & \frac{1}{2} &  & 0 &  & \frac{1}{2} &  & 2 &  & \frac{9}{2} &  & \cdots \\[1mm]
\hline
 & 2 &  & \frac{1}{2} &  & 0 &  & \frac{1}{2} &  & 2 &  & \frac{9}{2} & \cdots \\[1mm]
\hline
\vdots & \vdots & \vdots & \vdots & \vdots & \vdots & \vdots & \vdots & \vdots & \vdots & \vdots & \vdots & \ddots
\end{tabular}
};
\node [below=3mm of Kac1] {
$t = 1, \qquad (p,p')=(1,1),\qquad c = \dfrac{3}{2}.$
};
\node (Kac2) [below=20mm of Kac1] {
\setlength{\extrarowheight}{4pt}
\begin{tabular}{|CC|C|CC|C|CC|C|CC|C|C}
\hline
\BKL 0 & \BKL & -\frac{1}{6} & \BKL & \BKL 0 &  & \BKL \frac{1}{2} & \BKL & \frac{4}{3} & \BKL & \BKL \frac{5}{2} &  & \BKL \cdots \\[1mm]
\hline
\BKL & \BKL \frac{1}{2} &  & \BKL 0 & \BKL & -\frac{1}{6} & \BKL & \BKL 0 &  & \BKL \frac{1}{2} & \BKL & \frac{4}{3} & \BKL \cdots \\[1mm]
\hline
\BKL \frac{5}{2} & \BKL & \frac{4}{3} & \BKL & \BKL \frac{1}{2} &  & \BKL 0 & \BKL & -\frac{1}{6} & \BKL & \BKL 0 &  & \BKL \cdots \\[1mm]
\hline
\BKL & \BKL 4 &  & \BKL \frac{5}{2} & \BKL & \frac{4}{3} & \BKL & \BKL \frac{1}{2} &  & \BKL 0 & \BKL & -\frac{1}{6} & \BKL \cdots \\[1mm]
\hline
\BKL 8 & \BKL & \frac{35}{6} & \BKL & \BKL 4 &  & \BKL \frac{5}{2} & \BKL & \frac{4}{3} & \BKL & \BKL \frac{1}{2} &  & \BKL \cdots \\[1mm]
\hline
\BKL & \BKL \frac{21}{2} &  & \BKL 8 & \BKL & \frac{35}{6} & \BKL & \BKL 4 &  & \BKL \frac{5}{2} & \BKL & \frac{4}{3} & \BKL \cdots \\[1mm]
\hline
\BKL \vdots & \BKL \vdots & \vdots & \BKL \vdots & \BKL \vdots & \vdots & \BKL \vdots & \BKL \vdots & \vdots & \BKL \vdots & \BKL \vdots & \vdots & \BKL \ddots
\end{tabular}
};
\node [below=3mm of Kac2] {
$t = \dfrac{1}{3}, \qquad (p,p')=(1,3),\qquad c = -\dfrac{5}{2}.$
};
\node (Kac3) [below=20mm of Kac2] {
\setlength{\extrarowheight}{4pt}
\begin{tabular}{|CCC|C|CCC|C|CCC|C|C}
\hline
\IKL 0 & \IKL & \IKL 0 & \BKL & \IKL \frac{1}{2} & \IKL & \IKL \frac{3}{2} & \BKL & \IKL 3 & \IKL & \IKL 5 & \BKL & \IKL \cdots \\[1mm]
\hline
\BKL & \BKL \frac{3}{16} & \BKL & -\frac{1}{16} & \BKL & \BKL \frac{3}{16} & \BKL & \frac{15}{16} & \BKL & \BKL \frac{35}{16} & \BKL & \frac{63}{16} & \BKL \cdots \\[1mm]
\hline
\IKL \frac{3}{2} & \IKL & \IKL \frac{1}{2} & \BKL & \IKL 0 & \IKL & \IKL 0 & \BKL & \IKL \frac{1}{2} & \IKL & \IKL \frac{3}{2} & \BKL & \IKL \cdots \\[1mm]
\hline
\BKL & \BKL \frac{35}{16} & \BKL & \frac{15}{16} & \BKL & \BKL \frac{3}{16} & \BKL & -\frac{1}{16} & \BKL & \BKL \frac{3}{16} & \BKL & \frac{15}{16} & \BKL \cdots \\[1mm]
\hline
\IKL 5 & \IKL & \IKL 3 & \BKL & \IKL \frac{3}{2} & \IKL & \IKL \frac{1}{2} & \BKL & \IKL 0 & \IKL & \IKL 0 & \BKL & \IKL \cdots \\[1mm]
\hline
\BKL & \BKL \frac{99}{16} & \BKL & \frac{63}{16} & \BKL & \BKL \frac{35}{16} & \BKL & \frac{15}{16} & \BKL & \BKL \frac{3}{16} & \BKL & -\frac{1}{16} & \BKL \cdots \\[1mm]
\hline
\IKL \vdots & \IKL \vdots & \IKL \vdots & \BKL \vdots & \IKL \vdots & \IKL \vdots & \IKL \vdots & \BKL \vdots & \IKL \vdots & \IKL \vdots & \IKL \vdots & \BKL \vdots & \IKL \ddots
\end{tabular}
};
\node [below=3mm of Kac3] {
$t = \dfrac{1}{2}, \qquad (p,p')=(2,4),\qquad c = 0.$
};
\end{tikzpicture}
\caption{Parts of three of the extended \ns{} Kac tables for $c=\tfrac{3}{2}$, $c=-\tfrac{5}{2}$ and $c=0$.  The rows of the tables are labelled by $r = 1, 2, 3, \ldots$ and the columns by $s = 1, 2, 3, \ldots$\,.  Interior points are shaded dark grey, boundary points are shaded light grey, while corner points are white.} \label{fig:KacTables}
\end{center}
\end{figure}
}

\subsection{\ns{} Fock spaces} \label{sec:Fock}

The $N=1$ superconformal algebras have a free field realisation in terms of a free boson and a free fermion, the latter taken in the free fermion \ns{} or Ramond sector to obtain the corresponding superconformal sectors.  In particular, the \ns{} algebra acts on the tensor product of any bosonic Fock space with the vacuum fermionic Fock space.  We shall refer to such tensor products as \emph{\ns{} Fock spaces}.

At the level of \opes{}, one starts with a free boson field $\func{a}{z} = \sum_{n \in \ZZ} a_n z^{-n-1}$ and a free 
fermion field $\func{b}{z} = \sum_{j \in \ZZ-1/2} b_j z^{-j-1/2}$ satisfying
\begin{equation}
\func{a}{z} \func{a}{w} \sim \frac{1}{\brac{z-w}^2}, \qquad \func{b}{z} \func{b}{w} \sim \frac{1}{z-w}.
\end{equation}
These are then used to construct the energy-momentum tensor and its superpartner:
\begin{equation}
\func{T}{z} = \frac{1}{2} \normord{\func{a}{z} \func{a}{z}} + \frac{Q}{2} \func{\pd a}{z} + \frac{1}{2} \normord{\func{\pd b}{z} \func{b}{z}}, \qquad 
\func{G}{z} = \func{a}{z} \func{b}{z} + Q \func{\pd b}{z}.
\end{equation}
Here, $\normord{\cdots}$ denotes normal ordering.  It is straightforward to check that these fields satisfy the \opes{} \eqref{OPE:TTTGGG} with $c = \frac{3}{2} - 3Q^2$.  To match the central charge of \eqref{eq:ParByt}, we set
\begin{equation}
Q = \sqrt{\frac{p'}{p}} - \sqrt{\frac{p \vphantom{p'}}{p'}} = \frac{p'-p}{\sqrt{pp'}}.
\end{equation}

The \ns{} Fock space $\Fock{\lambda}$, being the tensor product of a free boson Verma module with the free fermion vacuum module, is generated by a \hws{} $v_{\lambda}$ satisfying
\begin{equation}
a_n v_{\lambda} = b_j v_{\lambda} = 0 \quad \text{for \(n,j>0\);} \qquad a_0 v_{\lambda} = \lambda v_{\lambda}.
\end{equation}
This generator $v_{\lambda}$ then has conformal weight
\begin{equation}
h_{\lambda} = \frac{1}{2} \lambda \brac{\lambda - Q} = \frac{4pp' \brac{\lambda - Q/2}^2 - \brac{p'-p}^2}{8pp'}
\end{equation}
which will coincide with a weight $h_{r,s}$ of the extended Kac table, given in \eqref{eq:ParByt}, when
\begin{equation} \label{eq:DefLambda}
\lambda = \lambda_{r,s} \equiv -\alpha' \brac{r-1} + \alpha \brac{s-1}.
\end{equation}
Here, we have introduced the quantities
\begin{equation}
\alpha = \sqrt{\frac{p \vphantom{p'}}{4p'}}, \qquad \alpha' = \sqrt{\frac{p'}{4p}}
\end{equation}
and note the symmetries
\begin{equation} \label{eq:FFSymm}
\lambda_{r+p,s} = \lambda_{r,s} - \frac{1}{2} \sqrt{pp'}, \quad 
\lambda_{r,s+p'} = \lambda_{r,s} + \frac{1}{2} \sqrt{pp'} \qquad \Ra \qquad 
\lambda_{r+p,s+p'} = \lambda_{r,s},
\end{equation}
for later use.

Fock spaces are always simple as modules over the product of the free boson and fermion (universal enveloping) algebras.  However, this need not remain true upon restricting to the \ns{} algebra.  Specifically, the Fock space $\Fock{\lambda}$ will only be simple as a \ns{} module when $\lambda \neq \lambda_{r,s}$ for any $r,s \in \ZZ$ with $r=s \bmod{2}$.  For $t \in \QQ_+$, we describe the submodule structure of $\Fock{\lambda}$ as islands, a chain, or a braid, illustrating the possibilities in \cref{fig:FockStructures}.  These structures mirror those of the Feigin-Fuchs modules of the Virasoro algebra.  More precisely, when $(r,s)$ is a corner/boundary/interior type entry of the extended Kac table, then the submodule structure of $\Fock{r,s} \equiv \Fock{\lambda_{r,s}}$ is of islands/chain/braid type.  There are two possible structures for chain and braid type Fock spaces $\Fock{r,s}$, corresponding to the fact that these \ns{} modules are not isomorphic to their contragredient duals $\Fock{Q-\lambda_{r,s}} = \Fock{-r,-s}$.  We shall describe the detailed structure of the $\Fock{r,s}$ in the rest of this section.

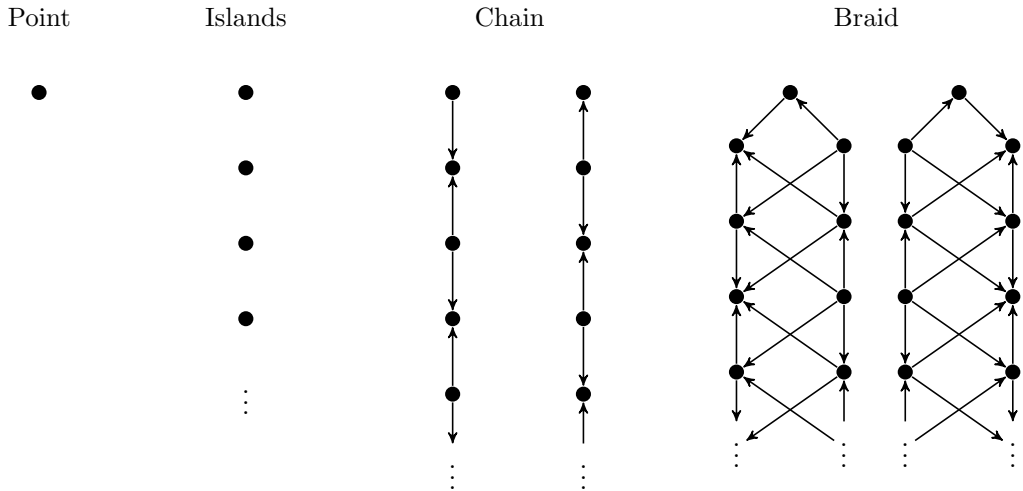
\begin{figure}
\begin{tikzpicture}
  [->,node distance=1cm,>=stealth',semithick,
   asoc/.style={circle,draw=black,fill=black,inner sep = 2pt,minimum size=5pt},
   bsoc/.style={circle,draw=black,fill=gray,inner sep = 2pt,minimum size=5pt},
   csoc/.style={circle,draw=black,fill=white,inner sep = 2pt,minimum size=5pt}
  ]
  \node[sv] (1) [] {};
  \node[] (point) [above of =1] {Point};
  \node[sv] (2) [right = 2.5cm of 1] {};
  \node[] (island) [above of =2] {Islands};
  \node[sv] (2a) [below of =2] {};
  \node[sv] (2b) [below of =2a] {};
  \node[sv] (2c) [below of =2b] {};
  \node[inner sep = 2pt] (2d) [below of =2c] {$\vdots$};
  \node[subsv] (3) [right = 2.5cm of 2] {};
  \node[] (chain) [right = 0.75cm of 3,above of =3] {Chain};
  \node[sv] (3a) [below of =3] {};
  \node[subsv] (3b) [below of =3a] {};
  \node[sv] (3c) [below of =3b] {};
  \node[subsv] (3d) [below of =3c] {};
  \node[inner sep = 2pt] (3e) [below of =3d] {$\vdots$};
  \path[] (3) edge node {} (3a)
          (3b) edge node {} (3a)
          (3b) edge node {} (3c)
          (3d) edge node {} (3c)
          (3d) edge node {} (3e);
  \node[sv] (4) [right = 1.5cm of 3] {};
  \node[subsv] (4a) [below of =4] {};
  \node[sv] (4b) [below of =4a] {};
  \node[subsv] (4c) [below of =4b] {};
  \node[sv] (4d) [below of =4c] {};
  \node[inner sep = 2pt] (4e) [below of =4d] {$\vdots$};
  \path[] (4a) edge node {} (4)
          (4a) edge node {} (4b)
          (4c) edge node {} (4b)
          (4c) edge node {} (4d)
          (4e) edge node {} (4d);
  \node[subsv] (5) [right = 2.5cm of 4] {};
  \node[] (braid) [right = 1cm of 5,above of =5] {Braid};
  \node[sv] (5a) [below left of =5] {};
  \node[subsv] (5b) [below of =5a] {};
  \node[sv] (5c) [below of =5b] {};
  \node[subsv] (5d) [below of =5c] {};
  \node[inner sep = 2pt] (5e) [below of =5d] {$\vdots$};
  \node[csv] (5j) [below right of =5] {};
  \node[subsv] (5k) [below of =5j] {};
  \node[csv] (5l) [below of =5k] {};
  \node[subsv] (5m) [below of =5l] {};
  \node[inner sep = 2pt] (5n) [below of =5m] {$\vdots$};
  \path[] (5) edge node {} (5a)
          (5b) edge node {} (5a)
          (5k) edge node {} (5a)
          (5b) edge node {} (5c)
          (5j) edge node {} (5b)
          (5l) edge node {} (5b)
          (5d) edge node {} (5c)
          (5k) edge node {} (5c)
          (5m) edge node {} (5c)
          (5d) edge node {} (5e)
          (5l) edge node {} (5d)
          (5n) edge node {} (5d)
          (5m) edge node {} (5e)
          (5j) edge node {} (5)
          (5j) edge node {} (5k)
          (5l) edge node {} (5k)
          (5l) edge node {} (5m)
          (5n) edge node {} (5m);	
  \node[subsv] (6) [right = 2cm of 5] {};
  \node[csv] (6a) [below left of =6] {};
  \node[subsv] (6b) [below of =6a] {};
  \node[csv] (6c) [below of =6b] {};
  \node[subsv] (6d) [below of =6c] {};
  \node[inner sep = 2pt] (6e) [below of =6d] {$\vdots$};
  \node[sv] (6j) [below right of =6] {};
  \node[subsv] (6k) [below of =6j] {};
  \node[sv] (6l) [below of =6k] {};
  \node[subsv] (6m) [below of =6l] {};
  \node[inner sep = 2pt] (6n) [below of =6m] {$\vdots$};
  \path[] (6) edge node {} (6j)
          (6k) edge node {} (6j)
          (6b) edge node {} (6j)
          (6k) edge node {} (6l)
          (6a) edge node {} (6k)
          (6c) edge node {} (6k)
          (6m) edge node {} (6l)
          (6b) edge node {} (6l)
          (6d) edge node {} (6l)
          (6m) edge node {} (6n)
          (6c) edge node {} (6m)
          (6e) edge node {} (6m)
          (6d) edge node {} (6n)
          (6a) edge node {} (6)
          (6a) edge node {} (6b)
          (6c) edge node {} (6b)
          (6c) edge node {} (6d)
          (6e) edge node {} (6d);	
\end{tikzpicture}
\caption{The structure of Fock spaces over the \ns{} algebra, for $t \in \QQ_+$ ($c \le \tfrac{3}{2}$).  Each black circle represents a \ssv{} and the arrows represent the action of the algebra as in \cref{fig:VermaStructures}.  The two braid diagrams are actually reflections of one another, but the repetition reminds us that braid type Fock spaces are not self-contragredient.} \label{fig:FockStructures}
\end{figure}

It is important to note that, unlike the Verma modules of the previous section, the submodules of the \ns{} Fock spaces are not necessarily generated by \svs{}.  Instead, submodules may be associated with \emph{\ssvs{}}, these being vectors which become singular in an appropriate quotient of the parent module.  Unlike \svs{}, the Fock space \ssvs{} are not usually unique up to normalisation because subsingularity is preserved by adding any element of the submodule by which one quotients to obtain a \sv{}.  In \cref{fig:FockStructures}, we may regard each black circle as representing a \ssv{}, modulo the aforementioned non-uniqueness, generating submodules of the Fock space.

It is natural to associate \ssvs{} with the simple quotient of the submodule that they generate.  The circles in \cref{fig:FockStructures} thus also represent the (simple) composition factors of the Fock space.  For point and islands type Fock spaces, there are no arrows in the associated diagrams, indicating that these Fock spaces are semisimple.  For chain and braid type Fock spaces, the composition factors may be partitioned into three classes according as to whether the arrows incident on the corresponding circle are all pointing towards it, all pointing away from it, or there are some pointing towards and some pointing away.  The latter class is empty for chain-type modules.  When all arrows point towards the circle, it represents a \sv{} generating a simple submodule.  The direct sum of the composition factors in this class therefore gives the maximal semisimple submodule, also known as the \emph{socle} of the Fock space.  Similarly, the maximal semisimple quotient, also called the \emph{head} of the Fock space, is the direct sum of the composition factors in the class corresponding to all arrows pointing away.

The conformal weights of the \ssvs{} of the Fock spaces $\Fock{r,s}$, with $r,s \in \ZZ_+$ and $r=s \bmod{2}$, coincide with those of the \svs{} in the \ns{} module $\Ver{r,s}$.  It therefore follows that $\Fock{r,s}$ has a \ssv{} of depth $\frac{1}{2} rs$ and that its contragredient dual $\Fock{-r,-s}$ does too.  One can verify that this \ssv{} is always associated to either the socle or the head of the Fock space (its circle in \cref{fig:FockStructures} has either all arrows pointing towards it or all arrows pointing away from it, respectively).  With the parametrisation \eqref{eq:DefLambda} that we chose above, it turns out that the depth $\frac{1}{2} rs$ \ssv{} is always associated with the head for $r,s \in \ZZ_+$ and with the socle for $r,s \in \ZZ_-$.  This structural identification may be extended to all $r,s \in \ZZ$ by using the symmetries \eqref{eq:FFSymm}.

This realisation fixes the structure of a chain type Fock space $\Fock{r,s}$ uniquely.  One only has to determine whether the depth $\frac{1}{2} rs$ \ssv{} belongs to the socle or the head and find the number of \ssvs{} of lesser depth; this is sufficient to distinguish between the two possibilities in \cref{fig:FockStructures}.  When $\Fock{r,s}$ is of braid type, this information is not quite sufficient.  At every other horizontal level in the braid type pictures in \cref{fig:FockStructures}, there is one singular and one (non-singular) \ssv{} (at the other horizontal levels, excepting the highest, no vector is singular).  To identify which is which, given their depths, the following fact is germane:  if the depth of the \sv{} at a given horizontal level is greater than that of the \ssv{} at the same level, then it will also be greater at the other horizontal levels (and vice versa).  One may then check which has greater depth in a given module because one knows the nature of the depth $\frac{1}{2} rs$ \ssv{}.

\subsection{Kac modules} \label{sec:Kac}

The reducible Fock spaces $\Fock{r,s}$, with $r,s \in \ZZ$ and $r=s \bmod{2}$, are not themselves the modules of central interest here.  Rather, it is a certain related class of modules that we shall call (\ns{}) \emph{Kac modules}, following \cite{RasCla11,BusKaz12,MorKac15}, that take centre stage.  These are indexed by positive integers $r$ and $s$.  Over the Virasoro algebra, Kac modules were introduced non-constructively in \cite{PeaLog06,PRannecy,RasFus07,RasFus07b} in order to describe the boundary sectors of the scaling limits of certain integrable lattice models.  Their characters were determined in many examples and the results were compatible with Kac modules being identified with certain quotients of Verma modules.  However, more recent consistency checks have led to a different proposal \cite{RasCla11} for the identity of Virasoro Kac modules, at least for some models, as submodules of Fock spaces rather than as quotients of Verma modules.  This proposal has been generalised in \cite{MorKac15} where evidence verifying this identification in many non-trivial examples is presented.

\ns{} Kac modules were recently considered from a lattice point of view in \cite{PeaLog14}.  Although this analysis only studied the action of $L_0$ on certain examples and so obtained only a bare minimum of structural information, it is reasonable to expect that these \ns{} Kac modules may likewise be identified with submodules of \ns{} Fock spaces.  Extrapolating the results of \cite{MorKac15} leads us to define the \emph{\ns{} Kac module} $\Kac{r,s}$, with $r,s \in \ZZ_+$ and $r=s \bmod{2}$, to be the submodule of the \ns{} Fock space $\Fock{r,s}$ that is generated by the subsingular vectors of depth strictly less than $\frac{1}{2} rs$.  Note that this does not exclude the possibility that $\Kac{r,s}$ may have a singular vector of depth greater than $\frac{1}{2} rs$.

We illustrate this definition with examples of Kac modules for $(p,p') = (2,4)$ ($c=0$).  First, we determine the structure of the corresponding Fock spaces in \cref{fig:Fockc=0}.  Remembering that $\Fock{r,s} = \Fock{r+p,s+p'}$, by the symmetries \eqref{eq:FFSymm}, we arrive at the Kac module structures depicted in \cref{fig:Kacc=0}.  For general $p$ and $p'$, with $\gcd\set{p,\frac{1}{2}(p'-p)}=1$,
the structure of the Kac module $\Kac{r,s}$ is indicated, for small $r$ and $s$, in \cref{fig:KacStructures}.

\begin{figure}[p]
\begin{tikzpicture}[->,-stealth',scale=0.8, transform shape, node distance=1.1cm]
  \node[main node, label = above:$0$] (1) [] {};
  \node[] (n) [above of =1] {$\Fock{1,1} = \Fock{3,5} = \cdots$};	
  \node[csv, label = left:$\frac{1}{2}$] (2) [below left of=1] {};
  \node[main node, label = right:$\frac{3}{2}$] (3) [below right of=1] {};
  \node[subsv, label = left:$3$] (4) [below of=2] {};
  \node[subsv, label = right:$5$] (5) [below of=3] {};
  \node[subsv, label = left:$\frac{15}{2}$] (7) [below of=4] {};
  \node[subsv, label = right:$\frac{21}{2}$] (6) [below of=5] {};
  \node[subsv, label = left:$14$, label = below:$\vdots$] (8) [below of=7] {};
  \node[subsv, label = right:$18$, label = below:$\vdots$] (9) [below of=6] {};
 \path[every node/.style={font=\sffamily\small}]
	(2) edge (1)
	(1) edge (3)
	(2) edge (4)
	(2) edge (5)
	(4) edge (3)
	(5) edge (3)
	(5) edge (6)
	(4) edge (6)
	(7) edge (4)
	(7) edge (5)
	(7) edge (8)
	(8) edge (6)
	(9) edge (6)
	(7) edge (9);
  \node[] (tmp) [right = 2cm of 1] {};
  \node[main node, label = left:$\frac{3}{16}$] (1d) [below = 7cm of tmp] {};
  \node[] (nd) [above of =1d] {$\Fock{2,2} = \Fock{4,6} = \cdots$};
  \node[subsv, label = left:$\frac{35}{16}$] (2d) [below of=1d] {};
  \node[subsv, label = left:$\frac{99}{16}$] (3d) [below of=2d] {};
  \node[subsv, label = left:$\frac{195}{16}$, label = below:$\vdots$] (4d) [below of=3d] {};
 \path[every node/.style={font=\sffamily\small}]
	(2d) edge (1d)
	(2d) edge (3d)
	(4d) edge (3d);
  \node[main node, label = above:$0$] (1a) [right = 4cm of 1] {};
  \node[] (na) [above of =1a] {$\Fock{1,3} = \Fock{3,7} = \cdots$};	
  \node[main node, label = left:$\frac{1}{2}$] (2a) [below left of=1a] {};
  \node[csv, label = right:$\frac{3}{2}$] (3a) [below right of=1a] {};
  \node[subsv, label = left:$3$] (4a) [below of=2a] {};
  \node[subsv, label = right:$5$] (5a) [below of=3a] {};
  \node[main node, draw, label = left:$\frac{15}{2}$] (6a) [below of=4a] {};  
  \node[subsv, label = right:$\frac{21}{2}$] (7a) [below of=5a] {};
  \node[subsv, label = left:$14$, label = below:$\vdots$] (8a) [below of=6a] {};
  \node[subsv, label = right:$18$, label = below:$\vdots$] (9a) [below of=7a] {};
 \path[every node/.style={font=\sffamily\small}]
	(1a) edge (2a)
	(3a) edge (1a)
	(4a) edge (2a)
	(5a) edge (2a)
	(3a) edge (4a)
	(3a) edge (5a)
	(4a) edge (6a)
	(5a) edge (6a)
	(7a) edge (4a)
	(7a) edge (5a)
	(8a) edge (6a)
 	(7a) edge (8a)
	(7a) edge (9a)
 	(9a) edge (6a);
  \node[main node, label = left:$-\frac{1}{16}$] (1e) [right = 4cm of 1d] {};
  \node[] (ne) [above of =1e] {$\Fock{2,4} = \Fock{4,8} = \cdots$};
  \node[subsv, label = left:$\frac{63}{16}$] (2e) [below of=1e] {};
  \node[subsv, label = left:$\frac{255}{16}$] (3e) [below of=2e] {};
  \node[subsv, label = left:$\frac{575}{16}$, label = below:$\vdots$] (4e) [below of=3e] {};
  \node[main node, label = above:$\frac{1}{2}$] (1b) [right = 4cm of 1a] {};
  \node[] (nb) [above of =1b] {$\Fock{1,5} = \Fock{3,9} = \cdots$};	
  \node[csv, label = left:$3$] (2b) [below left of=1b] {};
  \node[main node, label = right:$5$] (3b) [below right of=1b] {};
  \node[subsv, label = left:$\frac{15}{2}$] (4b) [below of=2b] {};
  \node[subsv, label = right:$\frac{21}{2}$] (5b) [below of=3b] {};
  \node[subsv, label = left:$14$] (7b) [below of=4b] {};
  \node[subsv, label = right:$18$] (6b) [below of=5b] {};
  \node[subsv, label = left:$\frac{45}{2}$, label = below:$\vdots$] (8b) [below of=7b] {};
  \node[subsv, label = right:$\frac{55}{2}$, label = below:$\vdots$] (9b) [below of=6b] {};
 \path[every node/.style={font=\sffamily\small}]
	(2b) edge (1b)
	(1b) edge (3b)
	(2b) edge (4b)
	(2b) edge (5b)
	(4b) edge (3b)
	(5b) edge (3b)
	(5b) edge (6b)
	(4b) edge (6b)
	(7b) edge (4b)
	(7b) edge (5b)
	(7b) edge (8b)
	(8b) edge (6b)
	(9b) edge (6b)
	(7b) edge (9b);
  \node[main node, label = left:$\frac{3}{16}$] (1f) [right = 4cm of 1e] {};
  \node[] (nf) [above of =1f] {$\Fock{2,6} = \Fock{4,10} = \cdots$};
  \node[subsv, label = left:$\frac{35}{16}$] (2f) [below of=1f] {};
  \node[subsv, label = left:$\frac{99}{16}$] (3f) [below of=2f] {};
  \node[subsv, label = left:$\frac{195}{16}$, label = below:$\vdots$] (4f) [below of=3f] {};
 \path[every node/.style={font=\sffamily\small}]
	(1f) edge (2f)
	(3f) edge (2f)
	(3f) edge (4f);
  \node[main node, label = above:$\frac{3}{2}$] (1c) [right = 4cm of 1b] {};
  \node[] (nc) [above of =1c] {$\Fock{1,7} = \Fock{3,11} = \cdots$};	
  \node[main node, label = left:$3$] (2c) [below left of=1c] {};
  \node[csv, label = right:$5$] (3c) [below right of=1c] {};
  \node[subsv, label = left:$\frac{15}{2}$] (4c) [below of=2c] {};
  \node[subsv, label = right:$\frac{21}{2}$] (5c) [below of=3c] {};
  \node[main node, draw, label = left:$14$] (6c) [below of=4c] {};  
  \node[subsv, label = right:$18$] (7c) [below of=5c] {};
  \node[subsv, label = left:$\frac{45}{2}$, label = below:$\vdots$] (8c) [below of=6c] {};
  \node[subsv, label = right:$\frac{55}{2}$, label = below:$\vdots$] (9c) [below of=7c] {};
 \path[every node/.style={font=\sffamily\small}]
	(1c) edge (2c)
	(3c) edge (1c)
	(4c) edge (2c)
	(5c) edge (2c)
	(3c) edge (4c)
	(3c) edge (5c)
	(4c) edge (6c)
	(5c) edge (6c)
	(7c) edge (4c)
	(7c) edge (5c)
	(8c) edge (6c)
 	(7c) edge (8c)
	(7c) edge (9c);
  \node[main node, label = left:$\frac{15}{16}$] (1g) [right = 4cm of 1f] {};
  \node[] (ng) [above of =1g] {$\Fock{2,8} = \Fock{4,12} = \cdots$};
  \node[subsv, label = left:$\frac{143}{16}$] (2g) [below of=1g] {};
  \node[subsv, label = left:$\frac{399}{16}$] (3g) [below of=2g] {};
  \node[subsv, label = left:$\frac{783}{16}$, label = below:$\vdots$] (4g) [below of=3g] {};
\end{tikzpicture}
\caption{Examples of the structures of \ns{} Fock modules when $p=2$ and $p'=4$.} \label{fig:Fockc=0}
\vspace{1cm}
\begin{tikzpicture}[->,-stealth',scale = 0.8, transform shape, node distance=1.1cm]
  \node[main node, label = left:$0$] (1) [] {};
  \node[] (n) [above of =1] {$\Kac{1,1}$};	
  \node[subsv, label = left:$\frac{3}{2}$] (3) [below of=1] {};
 \path[every node/.style={font=\sffamily\small}]
   (1) edge (3);
  \node[main node, label = left:$0$] (1a) [right=4cm of 1] {};
  \node[] (22) [above of =1a] {$\Kac{1,3}$};	
  \node[main node, label = left:$\frac{1}{2}$] (3a) [below of=1a] {};
 \path[every node/.style={font=\sffamily\small}]
   (1a) edge (3a);
 \node[main node, label = left:$\frac{1}{2}$] (1b) [right=4cm of 1a] {};
  \node[] (22) [above of =1b] {$\Kac{1,5}$};	
 \node[main node, label = left:$5$] (3b) [below of=1b] {};
 \path[every node/.style={font=\sffamily\small}]
   (1b) edge (3b);
  \node[main node, label = left:$\frac{3}{2}$] (1c) [right=4cm of 1b] {};
  \node[] (22) [above of =1c] {$\Kac{1,7}$};	
  \node[main node, label = left:$3$] (3c) [below of=1c] {};
 \path[every node/.style={font=\sffamily\small}]
   (1c) edge (3c);
  \node[] (dummy1) [right = 2cm of 1] {};
  \node[main node, label = left:$\frac{3}{16}$] (E1) [below = 3cm of dummy1] {};
  \node[] (22) [above of =E1] {$\Kac{2,2}$};	
  \node[main node, label = left:$-\frac{1}{16}$] (E2) [right = 4cm of E1] {};
  \node[] (22) [above of =E2] {$\Kac{2,4}$};	
  \node[main node, label = left:$\frac{3}{16}$] (E3) [right = 4cm of E2] {};
  \node[] (22) [above of =E3] {$\Kac{2,6}$};	
  \node[main node, label = left:$\frac{35}{16}$] (E3a) [below of=E3] {};
 \path[every node/.style={font=\sffamily\small}]
   (E3) edge (E3a);
  \node[main node, label = left:$\frac{15}{16}$] (E3') [right = 4cm of E3] {};
  \node[] (22) [above of =E3'] {$\Kac{2,8}$};	
  \node[main node, label = left:$\frac{3}{2}$] (1d) [below=6cm of 1] {};
  \node[] (22) [above of =1d] {$\Kac{3,1}$};	
  \node[main node, label = left:$5$] (3d) [below of=1d] {};
 \path[every node/.style={font=\sffamily\small}]
   (1d) edge (3d);
  \node[main node, label = left:$\frac{1}{2}$] (1e) [right=4cm of 1d] {};
  \node[] (22) [above of =1e] {$\Kac{3,3}$};	
  \node[main node, label = left:$3$] (3e) [below of=1e] {};
 \path[every node/.style={font=\sffamily\small}]
   (1e) edge (3e);
  \node[main node, label = above:$0$] (1f) [right=4cm of 1e] {};
  \node[] (22) [above of =1f] {$\Kac{3,5}$};	
  \node[csv, label = left:$\frac{1}{2}$] (2f) [below left of=1f] {};
  \node[main node, label = right:$\frac{3}{2}$] (3f) [below right of=1f] {};
  \node[subsv, label = left:$3$] (4f) [below of=2f] {};
  \node[subsv, label = right:$5$] (5f) [below of=3f] {};
  \node[subsv, label = below:$\frac{21}{2}$] (6f) [below left of=5f] {};
 \path[every node/.style={font=\sffamily\small}]
   (2f) edge (1f)
   (1f) edge (3f)
   (2f) edge (4f)
   (2f) edge (5f)
   (4f) edge (3f)
   (5f) edge (3f)
   (5f) edge (6f)
   (4f) edge (6f);
 \node[main node, label = above:$0$] (1g) [right = 4cm of 1f] {};
  \node[] (22) [above of =1g] {$\Kac{3,7}$};	
 \node[main node, label = left:$\frac{1}{2}$] (2g) [below left of=1g] {};
 \node[csv, label = right:$\frac{3}{2}$] (3g) [below right of=1g] {};
 \node[subsv, label = left:$3$] (4g) [below of=2g] {};
 \node[subsv, label = right:$5$] (5g) [below of=3g] {};
 \node[main node, draw, label = below:$\frac{15}{2}$] (6g) [below right of=4g] {};  
 \path[every node/.style={font=\sffamily\small}]
   (1g) edge (2g)
   (3g) edge (1g)
   (4g) edge (2g)
   (5g) edge (2g)
   (3g) edge (4g)
   (3g) edge (5g)
   (4g) edge (6g)
   (5g) edge (6g);
  \node[main node, label = left:$\frac{35}{16}$] (E4) [below = 7.5cm of E1] {};
  \node[] (22) [above of =E4] {$\Kac{4,2}$};	
  \node[main node, label = left:$\frac{15}{16}$] (E5) [right = 4cm of E4] {};
  \node[] (22) [above of =E5] {$\Kac{4,4}$};	
  \node[main node, label = left:$\frac{3}{16}$] (E6) [right = 4cm of E5] {};
  \node[] (22) [above of =E6] {$\Kac{4,6}$};	
  \node[main node, label = left:$\frac{35}{16}$] (E6a) [below of=E6] {};
  \node[main node, label = left:$\frac{99}{16}$] (E6b) [below of=E6a] {};
 \path[every node/.style={font=\sffamily\small}]
   (E6a) edge (E6)
   (E6a) edge (E6b);
  \node[main node, label = left:$-\frac{1}{16}$] (E7) [right = 4cm of E6] {};
  \node[] (22) [above of =E7] {$\Kac{4,8}$};	
  \node[main node, label = left:$\frac{63}{16}$] (E7a) [below of =E7] {};	
\end{tikzpicture}
\caption{Examples of the structures of \ns{} Kac modules when $p=2$ and $p'=4$.} \label{fig:Kacc=0}
\end{figure}

\begin{figure}
\scalebox{0.5}{
\begin{tabular}{|D{3cm}|D{0.5cm}|D{3cm}|D{0.5cm}|D{3cm}|D{0.5cm}|D{3cm}|D{0.5cm}|D{3cm}}
\hline
\IKL 
\begin{tikzpicture} 
[->,thick,-stealth', transform shape, node distance=1.1cm]
\node[main node] (1) [] {};
\node[main node] (3) [below of=1] {};
\path[]
   (1) edge node {} (3);
\end{tikzpicture}
& \BKL 
\begin{tikzpicture} 
\node[main node] (1) [] {};
\end{tikzpicture}
& \IKL 
\begin{tikzpicture} 
[->,thick,-stealth', transform shape, node distance=1.1cm]
\node[main node] (1) [] {};
\node[main node] (3) [below of=1] {};
\path[]
   (1) edge node {} (3);
\end{tikzpicture}
& \BKL 
\begin{tikzpicture} 
\node[main node] (1) [] {};
\end{tikzpicture}
& \IKL 
\begin{tikzpicture} 
[->,thick,-stealth', transform shape, node distance=1.1cm]
\node[main node] (1) [] {};
\node[main node] (3) [below of=1] {};
\path[]
   (1) edge node {} (3);
\end{tikzpicture}
& \BKL 
\begin{tikzpicture} 
\node[main node] (1) [] {};
\end{tikzpicture}
& \IKL 
\begin{tikzpicture} 
[->,thick,-stealth', transform shape, node distance=1.1cm]
\node[main node] (1) [] {};
\node[main node] (3) [below of=1] {};
\path[]
   (1) edge node {} (3);
\end{tikzpicture}
& \BKL 
\begin{tikzpicture} 
\node[main node] (1) [] {};
\end{tikzpicture}
& \IKL 
\begin{tikzpicture} 
[->,thick,-stealth', transform shape, node distance=1.1cm]
\node[main node] (1) [] {};
\node[main node] (3) [below of=1] {};
\path[]
   (1) edge node {} (3);
\end{tikzpicture}
\\
\hline
\BKL
\begin{tikzpicture}
\node[main node] (1) [] {};
\end{tikzpicture}
& 
\begin{tikzpicture}
\node[main node] (1) [] {};
\end{tikzpicture}
& \BKL 
\begin{tikzpicture} 
[->,thick,-stealth', transform shape, node distance=1.1cm]
\node[main node] (1) [] {};
\node[main node] (3) [below of=1] {};
\path[]
   (1) edge node {} (3);
\end{tikzpicture}
& 
\begin{tikzpicture}
\node[main node] (1) [] {};
\end{tikzpicture}
& \BKL
\begin{tikzpicture}
 [->,thick,-stealth', transform shape, node distance=1.1cm]
 \node[main node] (E6) [] {};
 \node[main node] (E6a) [below of=E6] {};
 \path[]
   (E6) edge node {} (E6a);
\end{tikzpicture}
& 
\begin{tikzpicture}
 \node[main node] (1) [] {};
\end{tikzpicture}
&\BKL
\begin{tikzpicture} 
[->,thick,-stealth', transform shape, node distance=1.1cm]
\node[main node] (1) [] {};
\node[main node] (3) [below of=1] {};
\path[]
   (1) edge node {} (3);
\end{tikzpicture}
& 
\begin{tikzpicture}
 \node[main node] (1) [] {};
\end{tikzpicture}
&\BKL
\begin{tikzpicture} 
[->,thick,-stealth', transform shape, node distance=1.1cm]
\node[main node] (1) [] {};
\node[main node] (3) [below of=1] {};
\path[]
   (1) edge node {} (3);
\end{tikzpicture}
\\
\hline
\IKL 
\begin{tikzpicture} 
[->,thick,-stealth', transform shape, node distance=1.1cm]
\node[main node] (1) [] {};
\node[main node] (3) [below of=1] {};
\path[]
   (1) edge node {} (3);
\end{tikzpicture}
& \BKL 
\begin{tikzpicture} 
[->,thick,-stealth', transform shape, node distance=1.1cm]
\node[main node] (1) [] {};
\node[main node] (3) [below of=1] {};
\path[]
   (1) edge node {} (3);
\end{tikzpicture}
& \IKL 
\begin{tikzpicture} 
[->,thick,-stealth', transform shape, node distance=1.1cm]
 \node[main node] (1f) [] {};
 \node[csv] (2f) [below left of=1f] {};
 \node[main node] (3f) [below right of=1f] {};
 \node[subsv] (4f) [below of=2f] {};
 \node[subsv] (5f) [below of=3f] {};
 \node[subsv] (6f) [below left of=5f] {};
 \path[]
   (2f) edge node {} (1f)
   (1f) edge node {} (3f)
   (2f) edge node {} (4f)
   (2f) edge node {} (5f)
   (4f) edge node {} (3f)
   (5f) edge node {} (3f)
   (5f) edge node {} (6f)
   (4f) edge node {} (6f);
\end{tikzpicture}
& \BKL 
\begin{tikzpicture}
 [->,thick,-stealth', transform shape, node distance=1.1cm]
 \node[main node] (E6) [] {};
 \node[main node] (E6a) [below of=E6] {};
 \node[main node] (E6b) [below of=E6a] {};
 \path[]
   (E6a) edge node {} (E6)
   (E6a) edge node {} (E6b);
\end{tikzpicture}
& \IKL
\begin{tikzpicture} 
[->,thick,-stealth', transform shape, node distance=1.1cm]
 \node[main node] (1f) [] {};
 \node[csv] (2f) [below left of=1f] {};
 \node[main node] (3f) [below right of=1f] {};
 \node[subsv] (4f) [below of=2f] {};
 \node[subsv] (5f) [below of=3f] {};
 \node[subsv] (6f) [below left of=5f] {};
 \path[]
   (2f) edge node {} (1f)
   (1f) edge node {} (3f)
   (2f) edge node {} (4f)
   (2f) edge node {} (5f)
   (4f) edge node {} (3f)
   (5f) edge node {} (3f)
   (5f) edge node {} (6f)
   (4f) edge node {} (6f);
\end{tikzpicture}
& \BKL 
\begin{tikzpicture}
 [->,thick,-stealth', transform shape, node distance=1.1cm]
 \node[main node] (E6) [] {};
 \node[main node] (E6a) [below of=E6] {};
 \node[main node] (E6b) [below of=E6a] {};
 \path[]
   (E6a) edge node {} (E6)
   (E6a) edge node {} (E6b);
\end{tikzpicture}
&\IKL
\begin{tikzpicture} 
[->,thick,-stealth', transform shape, node distance=1.1cm]
 \node[main node] (1f) [] {};
 \node[csv] (2f) [below left of=1f] {};
 \node[main node] (3f) [below right of=1f] {};
 \node[subsv] (4f) [below of=2f] {};
 \node[subsv] (5f) [below of=3f] {};
 \node[subsv] (6f) [below left of=5f] {};
 \path[]
   (2f) edge node {} (1f)
   (1f) edge node {} (3f)
   (2f) edge node {} (4f)
   (2f) edge node {} (5f)
   (4f) edge node {} (3f)
   (5f) edge node {} (3f)
   (5f) edge node {} (6f)
   (4f) edge node {} (6f);
\end{tikzpicture}
& \BKL 
\begin{tikzpicture}
 [->,thick,-stealth', transform shape, node distance=1.1cm]
 \node[main node] (E6) [] {};
 \node[main node] (E6a) [below of=E6] {};
 \node[main node] (E6b) [below of=E6a] {};
 \path[]
   (E6a) edge node {} (E6)
   (E6a) edge node {} (E6b);
\end{tikzpicture}
&\IKL
\begin{tikzpicture} 
[->,thick,-stealth', transform shape, node distance=1.1cm]
 \node[main node] (1f) [] {};
 \node[csv] (2f) [below left of=1f] {};
 \node[main node] (3f) [below right of=1f] {};
 \node[subsv] (4f) [below of=2f] {};
 \node[subsv] (5f) [below of=3f] {};
 \node[subsv] (6f) [below left of=5f] {};
 \path[]
   (2f) edge node {} (1f)
   (1f) edge node {} (3f)
   (2f) edge node {} (4f)
   (2f) edge node {} (5f)
   (4f) edge node {} (3f)
   (5f) edge node {} (3f)
   (5f) edge node {} (6f)
   (4f) edge node {} (6f);
\end{tikzpicture}
\\
\hline
\BKL 
\begin{tikzpicture}
\node[main node] (1) [] {};
\end{tikzpicture}
& 
\begin{tikzpicture}
\node[main node] (1) [] {};
\end{tikzpicture}
& \BKL 
\begin{tikzpicture}
 [->,thick,-stealth', transform shape, node distance=1.1cm]
 \node[main node] (E6) [] {};
 \node[main node] (E6a) [below of=E6] {};
 \node[main node] (E6b) [below of=E6a] {};
 \path[]
   (E6a) edge node {} (E6)
   (E6a) edge node {} (E6b);
\end{tikzpicture}
& 
\begin{tikzpicture}
 [->,thick,-stealth', transform shape, node distance=1.1cm]
 \node[main node] (E6) [] {};
 \node[main node] (E6a) [below of=E6] {};
\end{tikzpicture}
& \BKL 
\begin{tikzpicture}
 [->,thick,-stealth', transform shape, node distance=1.1cm]
 \node[main node] (E6) [] {};
 \node[main node] (E7) [above of=E6] {};
 \node[main node] (E6a) [below of=E6] {};
 \node[main node] (E6b) [below of=E6a] {};
 \path[]
   (E7) edge node {} (E6)
   (E6a) edge node {} (E6)
   (E6a) edge node {} (E6b);
\end{tikzpicture}
& 
\begin{tikzpicture}
 [->,thick,-stealth', transform shape, node distance=1.1cm]
 \node[main node] (E6) [] {};
 \node[main node] (E6a) [below of=E6] {};
\end{tikzpicture}
&\BKL
\begin{tikzpicture}
 [->,thick,-stealth', transform shape, node distance=1.1cm]
 \node[main node] (E6) [] {};
 \node[main node] (E7) [above of=E6] {};
 \node[main node] (E6a) [below of=E6] {};
 \node[main node] (E6b) [below of=E6a] {};
 \path[]
   (E7) edge node {} (E6)
   (E6a) edge node {} (E6)
   (E6a) edge node {} (E6b);
\end{tikzpicture}
& 
\begin{tikzpicture}
 [->,thick,-stealth', transform shape, node distance=1.1cm]
 \node[main node] (E6) [] {};
 \node[main node] (E6a) [below of=E6] {};
\end{tikzpicture}
&\BKL
\begin{tikzpicture}
 [->,thick,-stealth', transform shape, node distance=1.1cm]
 \node[main node] (E6) [] {};
 \node[main node] (E7) [above of=E6] {};
 \node[main node] (E6a) [below of=E6] {};
 \node[main node] (E6b) [below of=E6a] {};
 \path[]
   (E7) edge node {} (E6)
   (E6a) edge node {} (E6)
   (E6a) edge node {} (E6b);
\end{tikzpicture}
\\
\hline
\IKL 
\begin{tikzpicture} 
[->,thick,-stealth', transform shape, node distance=1.1cm]
\node[main node] (1) [] {};
\node[main node] (3) [below of=1] {};
\path[]
   (1) edge node {} (3);
\end{tikzpicture}
& \BKL 
\begin{tikzpicture} 
[->,thick,-stealth', transform shape, node distance=1.1cm]
\node[main node] (1) [] {};
\node[main node] (3) [below of=1] {};
\path[]
   (1) edge node {} (3);
\end{tikzpicture}
& \IKL 
\begin{tikzpicture} 
[->,thick,-stealth', transform shape, node distance=1.1cm]
 \node[main node] (1f) [] {};
 \node[csv] (2f) [below left of=1f] {};
 \node[main node] (3f) [below right of=1f] {};
 \node[subsv] (4f) [below of=2f] {};
 \node[subsv] (5f) [below of=3f] {};
 \node[subsv] (6f) [below left of=5f] {};
 \path[]
   (2f) edge node {} (1f)
   (1f) edge node {} (3f)
   (2f) edge node {} (4f)
   (2f) edge node {} (5f)
   (4f) edge node {} (3f)
   (5f) edge node {} (3f)
   (5f) edge node {} (6f)
   (4f) edge node {} (6f);
\end{tikzpicture}
& \BKL 
\begin{tikzpicture}
 [->,thick,-stealth', transform shape, node distance=1.1cm]
 \node[main node] (E6) [] {};
 \node[main node] (E7) [above of=E6] {};
 \node[main node] (E6a) [below of=E6] {};
 \node[main node] (E6b) [below of=E6a] {};
 \path[]
   (E7) edge node {} (E6)
   (E6a) edge node {} (E6)
   (E6a) edge node {} (E6b);
\end{tikzpicture}
& \IKL
\begin{tikzpicture} 
[->,thick,-stealth', transform shape, node distance=1.1cm]
 \node[main node] (1f) [] {};
 \node[csv] (2f) [below left of=1f] {};
 \node[main node] (3f) [below right of=1f] {};
 \node[subsv] (4f) [below of=2f] {};
 \node[subsv] (5f) [below of=3f] {};
 \node[subsv] (6f) [below of=4f] {};
 \node[subsv] (7f) [below of=5f] {};
 \node[subsv] (8f) [below of =6f] {};
 \node[subsv] (9f) [below of =7f] {};
 \node[subsv] (10f) [below left of=9f] {};
 \path[]
   (2f) edge node {} (1f)
   (1f) edge node {} (3f)
   (2f) edge node {} (4f)
   (2f) edge node {} (5f)
   (4f) edge node {} (3f)
   (5f) edge node {} (3f)
   (6f) edge node {} (4f)
   (4f) edge node {} (7f)
   (6f) edge node {} (5f)
   (5f) edge node {} (7f)
   (6f) edge node {} (8f)
   (6f) edge node {} (9f)
   (8f) edge node {} (7f)
   (9f) edge node {} (7f)
   (8f) edge node {} (10f)
   (9f) edge node {} (10f);
\end{tikzpicture}
& \BKL 
\begin{tikzpicture}
 [->,thick,-stealth', transform shape, node distance=1.1cm]
 \node[main node] (E6) [] {};
 \node[main node] (E7) [above of=E6] {};
 \node[main node] (E8) [above of=E7] {};
 \node[main node] (E6a) [below of=E6] {};
 \node[main node] (E6b) [below of=E6a] {};
 \path[]
   (E7) edge node {} (E8)
   (E7) edge node {} (E6)
   (E6a) edge node {} (E6)
   (E6a) edge node {} (E6b);
\end{tikzpicture}
& \IKL
\begin{tikzpicture} 
[->,thick,-stealth', transform shape, node distance=1.1cm]
 \node[main node] (1f) [] {};
 \node[csv] (2f) [below left of=1f] {};
 \node[main node] (3f) [below right of=1f] {};
 \node[subsv] (4f) [below of=2f] {};
 \node[subsv] (5f) [below of=3f] {};
 \node[subsv] (6f) [below of=4f] {};
 \node[subsv] (7f) [below of=5f] {};
 \node[subsv] (8f) [below of =6f] {};
 \node[subsv] (9f) [below of =7f] {};
 \node[subsv] (10f) [below left of=9f] {};
 \path[]
   (2f) edge node {} (1f)
   (1f) edge node {} (3f)
   (2f) edge node {} (4f)
   (2f) edge node {} (5f)
   (4f) edge node {} (3f)
   (5f) edge node {} (3f)
   (6f) edge node {} (4f)
   (4f) edge node {} (7f)
   (6f) edge node {} (5f)
   (5f) edge node {} (7f)
   (6f) edge node {} (8f)
   (6f) edge node {} (9f)
   (8f) edge node {} (7f)
   (9f) edge node {} (7f)
   (8f) edge node {} (10f)
   (9f) edge node {} (10f);
\end{tikzpicture}
& \BKL 
\begin{tikzpicture}
 [->,thick,-stealth', transform shape, node distance=1.1cm]
 \node[main node] (E6) [] {};
 \node[main node] (E7) [above of=E6] {};
 \node[main node] (E8) [above of=E7] {};
 \node[main node] (E6a) [below of=E6] {};
 \node[main node] (E6b) [below of=E6a] {};
 \path[]
   (E7) edge node {} (E8)
   (E7) edge node {} (E6)
   (E6a) edge node {} (E6)
   (E6a) edge node {} (E6b);
\end{tikzpicture}
& \IKL
\begin{tikzpicture} 
[->,thick,-stealth', transform shape, node distance=1.1cm]
 \node[main node] (1f) [] {};
 \node[csv] (2f) [below left of=1f] {};
 \node[main node] (3f) [below right of=1f] {};
 \node[subsv] (4f) [below of=2f] {};
 \node[subsv] (5f) [below of=3f] {};
 \node[subsv] (6f) [below of=4f] {};
 \node[subsv] (7f) [below of=5f] {};
 \node[subsv] (8f) [below of =6f] {};
 \node[subsv] (9f) [below of =7f] {};
 \node[subsv] (10f) [below left of=9f] {};
 \path[]
   (2f) edge node {} (1f)
   (1f) edge node {} (3f)
   (2f) edge node {} (4f)
   (2f) edge node {} (5f)
   (4f) edge node {} (3f)
   (5f) edge node {} (3f)
   (6f) edge node {} (4f)
   (4f) edge node {} (7f)
   (6f) edge node {} (5f)
   (5f) edge node {} (7f)
   (6f) edge node {} (8f)
   (6f) edge node {} (9f)
   (8f) edge node {} (7f)
   (9f) edge node {} (7f)
   (8f) edge node {} (10f)
   (9f) edge node {} (10f);
\end{tikzpicture}
\\
\hline
\BKL 
\begin{tikzpicture}
\node[main node] (1) [] {};
\end{tikzpicture}
& 
\begin{tikzpicture}
\node[main node] (1) [] {};
\end{tikzpicture}
& \BKL 
\begin{tikzpicture}
 [->,thick,-stealth', transform shape, node distance=1.1cm]
 \node[main node] (E6) [] {};
 \node[main node] (E6a) [below of=E6] {};
 \node[main node] (E6b) [below of=E6a] {};
 \path[]
   (E6a) edge node {} (E6)
   (E6a) edge node {} (E6b);
\end{tikzpicture}
& 
\begin{tikzpicture} 
[->,thick,-stealth', transform shape, node distance=1.1cm]
\node[main node] (1) [] {};
\node[main node] (3) [below of=1] {};
\end{tikzpicture}
& \BKL 
\begin{tikzpicture}
 [->,thick,-stealth', transform shape, node distance=1.1cm]
 \node[main node] (E6) [] {};
 \node[main node] (E7) [above of=E6] {};
 \node[main node] (E8) [above of=E7] {};
 \node[main node] (E6a) [below of=E6] {};
 \node[main node] (E6b) [below of=E6a] {};
 \path[]
   (E7) edge node {} (E8)
   (E7) edge node {} (E6)
   (E6a) edge node {} (E6)
   (E6a) edge node {} (E6b);
\end{tikzpicture}
& 
\begin{tikzpicture}
 [->,thick,-stealth', transform shape, node distance=1.1cm]
 \node[main node] (E6) [] {};
 \node[main node] (E6a) [below of=E6] {};
 \node[main node] (E6b) [below of=E6a] {};
\end{tikzpicture}
&\BKL
\begin{tikzpicture}
 [->,thick,-stealth', transform shape, node distance=1.1cm]
 \node[main node] (E6) [] {};
 \node[main node] (E7) [above of=E6] {};
 \node[main node] (E8) [above of=E7] {};
 \node[main node] (E9) [above of=E8] {};
 \node[main node] (E6a) [below of=E6] {};
 \node[main node] (E6b) [below of=E6a] {};
 \path[]
   (E9) edge node {} (E8)
   (E7) edge node {} (E8)
   (E7) edge node {} (E6)
   (E6a) edge node {} (E6)
   (E6a) edge node {} (E6b);
\end{tikzpicture}
& 
\begin{tikzpicture}
 [->,thick,-stealth', transform shape, node distance=1.1cm]
 \node[main node] (E6) [] {};
 \node[main node] (E6a) [below of=E6] {};
 \node[main node] (E6b) [below of=E6a] {};
\end{tikzpicture}
&\BKL
\begin{tikzpicture}
 [->,thick,-stealth', transform shape, node distance=1.1cm]
 \node[main node] (E6) [] {};
 \node[main node] (E7) [above of=E6] {};
 \node[main node] (E8) [above of=E7] {};
 \node[main node] (E9) [above of=E8] {};
 \node[main node] (E6a) [below of=E6] {};
 \node[main node] (E6b) [below of=E6a] {};
 \path[]
   (E9) edge node {} (E8)
   (E7) edge node {} (E8)
   (E7) edge node {} (E6)
   (E6a) edge node {} (E6)
   (E6a) edge node {} (E6b);
\end{tikzpicture}
\end{tabular}
}
\caption{A depiction of the structures of the Kac modules $\Kac{r,s}$ as $(r,s)$ varies over (a part of) the extended Kac table.  The genuine Kac table, bounded by $1 \le r \le p-1$ and $1 \le s \le p'-1$, is represented by the dark grey rectangle in the upper-left corner.  These are interior points of the extended Kac table and the light grey and white areas correspond to boundary and corner type labels as in \cref{fig:KacTables}.  If $p=1$ or $p'=1$ (or both), then the possible structures correspond to removing the rows or columns (or both) that contain interior labels.} \label{fig:KacStructures}
\end{figure}
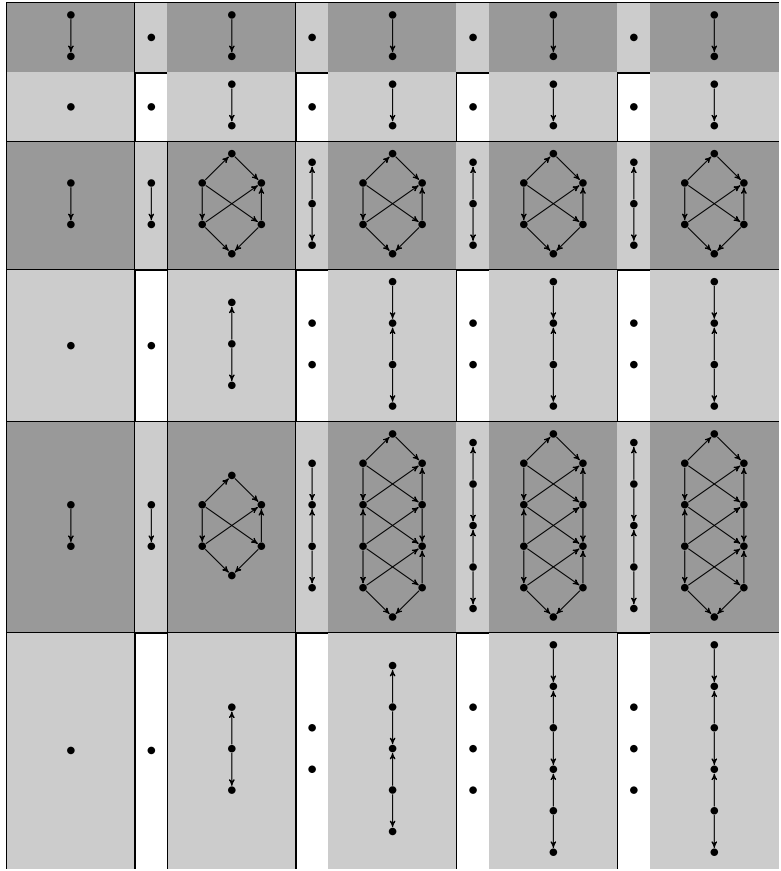

\section{A Verlinde formula} \label{sec:Verlinde}

For rational theories, a relatively efficient route to computing the fusion rules is to first determine the modular S-transforms of the characters of the simple modules and then apply the Verlinde formula.  For logarithmic theories, such a formula cannot compute the fusion multiplicities themselves, because characters cannot distinguish between an indecomposable module and the direct sum of its composition factors.  Instead, it is natural to require that a logarithmic Verlinde formula computes the structure constants of the projection of the fusion ring onto an appropriate ring of characters.\footnote{Other, typically model-dependent, Verlinde-like formulae have been proposed for logarithmic conformal field theories in \cite{FucNon04,FloVer07,GabFro08,GaiRad09,PeaGro10,RasVer10,PeaCos11}.}  In many cases, this character ring may be identified with the Grothendieck ring of fusion, in which indecomposable modules are identified with the sum of their composition factors.\footnote{This identification assumes that the characters of the simple modules are linearly independent.  This is indeed the case here, provided that one does not distinguish modules if they only differ through the parity of their \hws{}.  If one does want to distinguish these modules, then supercharacters should be used as well as characters.  In the latter setup, the sign with which a character appears becomes relevant.}  The expectation is therefore that the Verlinde formula determines the character, or equivalently the composition factors, of a fusion product.

There are a few provisos to this claim.  Technically, a Grothendieck ring of fusion will only exist if fusing with any given \ns{} module defines an exact functor from our chosen category of \ns{} modules to itself.  While this is not true for general \ns{} modules (the first example over the Virasoro algebra was given in \cite{GabFus09}), the modules for which it is true form a subring of the fusion ring \cite{KazTenIV94} (assuming that fusion defines a tensor structure on the appropriate category of \ns{} modules).  It has been conjectured that modules defining boundary sectors in a consistent boundary \cft{} do define exact functors under fusing.  We will therefore assume that the Kac modules have this property; consequently, this property will be shared by the modules that are generated from the Kac modules by fusion.  Thus, the subring of the \ns{} fusion ring that is generated by the Kac modules is hereby assumed to possess a well defined Grothendieck ring of fusion.  This subring is central to our investigations and we shall begin our analysis by exploring the detailed structure of its Grothendieck ring.

\subsection{Modular transformations} \label{sec:Mod}

We start with the study of the modular transformation properties of \ns{} characters.  Those of the Fock spaces are particularly accessible:
\begin{equation}
\fch{\Fock{\lambda}}{\tau} = q^{h_{\lambda} - c/24} \prod_{j=1}^{\infty} \frac{1+q^{j-1/2}}{1-q^j} = \frac{q^{\brac{\lambda - Q/2}^2 / 2}}{\func{\eta}{q}} \sqrt{\frac{\fjth{3}{1;q}}{\func{\eta}{q}}} \qquad \text{(\(q = \ee^{2 \pi \ii \tau}\)).}  
\label{chF}
\end{equation}
Here, $h_{\lambda}$ denotes the minimal conformal weight of $\Fock{\lambda}$ and we employ the Dedekind eta and Jacobi theta functions
\begin{equation}
\func{\eta}{q} = q^{1/24} \prod_{j=1}^{\infty} \brac{1-q^j}, \qquad 
\fjth{3}{z;q} = \prod_{j=1}^{\infty} \brac{1+zq^{j-1/2}} \brac{1-q^j} \brac{1+z^{-1}q^{j-1/2}}.
\end{equation}
Note that the character formula \eqref{chF} does not distinguish between $\Fock{\lambda}$ and its contragredient dual $\Fock{Q-\lambda}$.

The modular S-transformation of these characters is then
\begin{equation}
\fch{\Fock{\lambda}}{-1/\tau} = \int_{Q/2}^{\infty} \Smat{\Fock{\lambda}}{\Fock{\mu}} \fch{\Fock{\mu}}{\tau} \: \dd \mu, \qquad 
\Smat{\Fock{\lambda}}{\Fock{\mu}} = 2 \cos \tsqbrac{2 \pi \brac{\lambda - Q/2} \brac{\mu - Q/2}},
\end{equation}
where we restrict the range of the integral to $[Q/2,\infty)$ 
in order to integrate over linearly independent characters.  This may be easily verified using a standard gaussian integral, convergent for $\Im \tau > 0$.  However, because this S-transform requires a continuous family of Fock space characters, we expect to encounter singular distributions at some point.  As it can be confusing to allow endpoints to the integration domain when computing with these generalised functions, we will redefine the above S-transformation, once and for all, so that the integration range is open:
\begin{equation} \label{eq:DefS}
\fch{\Fock{\lambda}}{-1/\tau} = \int_{-\infty}^{\infty} \Smat{\Fock{\lambda}}{\Fock{\mu}} \fch{\Fock{\mu}}{\tau} \: \dd \mu, \qquad 
\Smat{\Fock{\lambda}}{\Fock{\mu}} = \cos \tsqbrac{2 \pi \brac{\lambda - Q/2} \brac{\mu - Q/2}}.
\end{equation}
The price to pay for extending to all $\mu \in \RR$ is that we must now identify $\Fock{\mu}$ and its contragredient $\Fock{Q-\mu}$ in all computations involving characters, in particular when we employ the Verlinde formula.

We remark that the kernel $\Smat{\Fock{\lambda}}{\Fock{\mu}}$ of this S-transform, the analogue of the S-matrix for rational theories, is symmetric 
($\Smat{\Fock{\lambda}}{\Fock{\mu}} = \Smat{\Fock{\mu}}{\Fock{\lambda}}$) and unitary:
\begin{align}
\int_{-\infty}^{\infty} \Smat{\Fock{\lambda}}{\Fock{\mu}} \Smat{\Fock{\nu}}{\Fock{\mu}}^* \: \dd \mu 
&= \int_{-\infty}^{\infty} \cos \tsqbrac{2 \pi \brac{\lambda - Q/2} \brac{\mu - Q/2}} \cos \tsqbrac{2 \pi \brac{\nu - Q/2} \brac{\mu - Q/2}} \: \dd \mu \notag \\
&= \frac{1}{2} \int_{-\infty}^{\infty} \brac{\cos \tsqbrac{2 \pi \brac{\lambda - \nu} \mu} + \cos \tsqbrac{2 \pi \brac{\lambda + \nu - Q} \mu}} \: \dd \mu \notag \\
&= \frac{1}{2} \tsqbrac{\func{\delta}{\nu = \lambda} + \func{\delta}{\nu = Q - \lambda}} = \func{\delta}{\nu = \lambda}.
\end{align}
Here, the final equality is justified by the fact that we must identify $\lambda$ and $Q - \lambda$ in such computations.  Because the S-transformation kernel is real, this computation also proves that $\modS$ squares to the identity operator.  In other words, conjugation is trivial in this theory (at the level of characters), as one would expect.

These results accord with our general expectations for good modular properties and indicate that we may expect meaningful results from the Verlinde formula.  Indeed, in the general formalism for the modularity of logarithmic \cfts{} proposed in \cite{CreLog13} and refined in \cite{RidVer14}, the $\Fock{\lambda}$ may be regarded as the \emph{standard modules}:  They are simple for almost all $\lambda \in \RR$, with respect to the Lebesgue measure used in \eqref{eq:DefS}; the simple $\Fock{\lambda}$ are termed the \emph{typical modules}.  The \emph{atypical modules} are then those corresponding to $\lambda = \lambda_{r,s}$, for $r,s \in \ZZ$.  In particular, the $\Fock{r,s}$ are not simple, hence are atypical.

The Kac modules $\Kac{r,s}$ of \cref{sec:Kac} are therefore also atypical and their characters are given by
\begin{equation} \label{ch:Kac}
\ch{\Kac{r,s}} = \ch{\Fock{\lambda_{r,s}}} - \ch{\Fock{\lambda_{-r,s}}}.
\end{equation}
We immediately obtain their S-transformations:
\begin{equation}
\begin{gathered}
\fch{\Kac{r,s}}{-1/\tau} = \int_{-\infty}^{\infty} \Smat{\Kac{r,s}}{\Fock{\mu}} \fch{\Fock{\mu}}{\tau} \: \dd \mu, \\
\Smat{\Kac{r,s}}{\Fock{\mu}} = \Smat{\Fock{\lambda_{r,s}}}{\Fock{\mu}} - \Smat{\Fock{\lambda_{-r,s}}}{\Fock{\mu}} = 2 \sin \tsqbrac{2 \pi r \alpha' \brac{\mu - Q/2}} \sin \tsqbrac{2 \pi s \alpha \brac{\mu - Q/2}}.
\end{gathered}
\end{equation}
Note that the result of S-transforming the character of a Kac module is an integral over the Fock space characters.  As the Fock spaces are the standard modules of the theory, their characters give the canonical topological basis in which to express all characters.\footnote{We refer to this basis as topological because one might only recover a given character, for example that of a simple atypical module $\Irr{r,s}$, as an infinite sum of Fock space characters.  In this case, the convergence is that of formal power series --- the contribution to the multiplicity of a given weight space is zero for all but finitely many terms in the sum.}  In particular, it is not clear that the quantity $\Smat{\Kac{r,s}}{\Kac{r',s'}}$ is well defined.  We will therefore perform all subsequent computations in the basis of standard characters without further comment.

Finally, we recall that Kac modules were only defined for $r,s \in \ZZ_+$.  If the character formula \eqref{ch:Kac} is extended to general $r,s \in \ZZ$, then we obtain
\begin{equation} \label{eq:KacCharSymm}
\ch{\Kac{-r,s}} = -\ch{\Kac{r,s}} = \ch{\Kac{r,-s}}, \qquad 
\ch{\Kac{r,0}} = \ch{\Kac{0,s}} = 0, \qquad 
\ch{\Kac{-r,-s}} = \ch{\Kac{r,s}}.
\end{equation}
These formulae are important for interpreting the results of general Verlinde computations.

\subsection{Verlinde products} \label{sec:VerProd}

Assuming, as discussed above, that the fusion product ``$\fuse$'' descends to a well defined product ``$\Grfuse$'' on the Grothendieck ring of characters, we may decompose a character product into a linear combination of Fock space characters:
\begin{equation} \label{eq:TheVerlindeFormula1}
\ch{\mathcal{M} \fuse \mathcal{N}} = \ch{\mathcal{M}} \Grfuse \ch{\mathcal{N}} = \int_{-\infty}^{\infty} \fuscoeff{\mathcal{M}}{\mathcal{N}}{\Fock{\nu}} \ch{\Fock{\nu}} \: \dd \nu.
\end{equation}
Here, $\mathcal{M}$ and $\mathcal{N}$ are \ns{} modules and we recall that the Fock space characters are the preferred (topological) basis for the space of all \ns{} characters.  The multiplicities $\fuscoeff{\mathcal{M}}{\mathcal{N}}{\Fock{\nu}}$ are the Verlinde coefficients and are computed, in terms of the S-transformation kernel, by the Verlinde formula:
\begin{equation} \label{eq:TheVerlindeFormula2}
\fuscoeff{\mathcal{M}}{\mathcal{N}}{\Fock{\nu}} = \int_{-\infty}^{\infty} \frac{\Smat{\mathcal{M}}{\Fock{\rho}} \Smat{\mathcal{N}}{\Fock{\rho}} \Smat{\Fock{\nu}}{\Fock{\rho}}^*}{\Smat{\Kac{1,1}}{\Fock{\rho}}} \: \dd \rho.
\end{equation}
Here, $\Kac{1,1}$ plays the role of the vacuum module.  Note that this Kac module is generated by a \hws{} of conformal weight $0$ that is annihilated by both $L_{-1}$ and $G_{-1/2}$.

It is now straight-forward to compute the character of fusion products.  The unitarity of the S-transformation implies that the vacuum module $\Kac{1,1}$ is the unit of the character product:  
\begin{equation}
\ch{\Kac{1,1} \fuse \mathcal{N}} = \ch{\Kac{1,1}} \Grfuse \ch{\mathcal{N}} = \ch{\mathcal{N}}.
\end{equation}
A somewhat less trivial example involves the fusion of $\Kac{3,1}$ with an arbitrary Fock space:
\begin{subequations}
\begin{align}
\fuscoeff{\Kac{3,1}}{\Fock{\mu}}{\Fock{\nu}} &= \int_{-\infty}^{\infty} \frac{\Smat{\Kac{3,1}}{\Fock{\rho}} \Smat{\Fock{\mu}}{\Fock{\rho}} \Smat{\Fock{\nu}}{\Fock{\rho}}^*}{\Smat{\Kac{1,1}}{\Fock{\rho}}} \: \dd \rho \notag \\
&= \int_{-\infty}^{\infty} \frac{\sin \sqbrac{6 \pi \alpha' \rho}}{\sin \sqbrac{2 \pi \alpha' \rho}} \cos \sqbrac{2 \pi \brac{\mu - Q/2} \rho} \cos \sqbrac{2 \pi \brac{\nu - Q/2} \rho} \: \dd \rho \notag \\
&= \frac{1}{2} \int_{-\infty}^{\infty} \brac{1 + 2 \cos \sqbrac{4 \pi \alpha' \rho}} \brac{\cos \sqbrac{2 \pi \brac{\mu - \nu} \rho} + \cos \sqbrac{2 \pi \brac{\mu + \nu - Q} \rho}} \: \dd \rho \notag \\
&= \func{\delta}{\nu = \mu - 2 \alpha'} + \func{\delta}{\nu = \mu} + \func{\delta}{\nu = \mu + 2 \alpha'} \notag \\
\Ra \qquad \ch{\Kac{3,1}} \Grfuse \ch{\Fock{\mu}} &= \ch{\Fock{\mu - 2 \alpha'}} + \ch{\Fock{\mu}} + \ch{\Fock{\mu + 2 \alpha'}}.
\intertext{Similar computations result in}
\ch{\Kac{1,3}} \Grfuse \ch{\Fock{\mu}} &= \ch{\Fock{\mu - 2 \alpha}} + \ch{\Fock{\mu}} + \ch{\Fock{\mu + 2 \alpha}}, \\
\ch{\Kac{2,2}} \Grfuse \ch{\Fock{\mu}} &= \ch{\Fock{\mu - \alpha' - \alpha}} + \ch{\Fock{\mu - \alpha' + \alpha}} + \ch{\Fock{\mu + \alpha' - \alpha}} + \ch{\Fock{\mu + \alpha' + \alpha}}.
\end{align}
\end{subequations}
Using \cref{ch:Kac}, we obtain the corresponding products with arbitrary Kac modules:
\begin{subequations}
\begin{align}
\ch{\Kac{3,1}} \Grfuse \ch{\Kac{r,s}} &= \ch{\Kac{r-2,s}} + \ch{\Kac{r,s}} + \ch{\Kac{r+2,s}}, \label{GrFR:K31xK} \\
\ch{\Kac{1,3}} \Grfuse \ch{\Kac{r,s}} &= \ch{\Kac{r,s-2}} + \ch{\Kac{r,s}} + \ch{\Kac{r,s+2}}, \label{GrFR:K13xK} \\
\ch{\Kac{2,2}} \Grfuse \ch{\Kac{r,s}} &= \ch{\Kac{r-1,s-1}} + \ch{\Kac{r-1,s+1}} + \ch{\Kac{r+1,s-1}} + \ch{\Kac{r+1,s+1}}. \label{GrFR:K22xK}
\end{align}
\end{subequations}
Here, we must employ \eqref{eq:KacCharSymm} if the labels on the Kac modules of the \rhs{} are not positive integers.

It follows from these character products that the Kac characters span a unital subring of the Grothendieck ring of \ns{} characters and that this subring is generated by $\ch{\Kac{3,1}}$, $\ch{\Kac{2,2}}$ and $\ch{\Kac{1,3}}$.\footnote{There is a simple exception when $c=\frac{3}{2}$ ($p=p'=1$) because then $\ch{\Kac{3,1}} = \ch{\Kac{1,3}}$ and $\ch{\Kac{2,2}} = \ch{\Kac{1,1}} + \ch{\Kac{1,3}}$.}  Associativity then leads to an explicit general formula for Kac character products:
\begin{equation} \label{GrFR:KxK}
\ch{\Kac{r,s}} \Grfuse \ch{\Kac{r',s'}} = \sideset{}{'} \sum_{r'' = \abs{r-r'}+1}^{r+r'-1} \ \sideset{}{'} \sum_{s'' = \abs{s-s'}+1}^{s+s'-1} \ch{\Kac{r'',s''}}.
\end{equation}
Here, the primed sums indicate that the summation variable increases in steps of two.  We mention, for later purposes, the following special case:
\begin{equation} \label{GrFR:Kr1xK1s}
\ch{\Kac{r,1}} \Grfuse \ch{\Kac{1,s}} = \ch{\Kac{r,s}}.
\end{equation}

\section{An explicit fusion product} \label{sec:TheExample}

In this section, we use an example to illustrate the steps involved in completely decomposing a fusion product and identifying its (indecomposable) direct summands.  To construct the fusion product itself, we utilise the Nahm-Gaberdiel-Kausch fusion algorithm, referring to \cref{app:Fusion} for further details concerning this technology.  The example is fairly involved and we have chosen it in order to illustrate a wide variety of the features and methods that we employ to analyse more general fusion rules.  Some rather more simple arguments are presented in \cref{sec:Kr1K1s} (though stripped of the explicit \NGK{} computations).

As our example, we consider the fusion of the \ns{} Kac module $\Kac{1,3}$ with itself at central charge $c=0$ ($p=2$ and $p'=4$).  A part of the extended \ns{} Kac table for $c=0$ appears in \cref{fig:KacTables}.  We remark that $\Kac{1,3}$, unlike most Kac modules, is a \hwm{}; indeed, it is generated by a \hws{} $v$ of conformal weight $h_{1,3} = 0$.  We may therefore identify $\Kac{1,3}$ as the quotient of the Verma module $\Ver{0}$ by the submodule generated by the \sv{} of conformal weight $\frac{3}{2}$.  Thus,
\begin{equation} \label{SV:K13}
\brac{L_{-1} G_{-1/2} - \frac{1}{2} G_{-3/2}} v = 0
\end{equation}
in $\Kac{1,3}$.  For simplicity, we shall assume throughout that $v$ is even.

First, we determine the character of the fusion product using the Verlinde formula.  Specifically, \cref{GrFR:K13xK} gives
\begin{equation}
\ch{\Kac{1,3}} \Grfuse \ch{\Kac{1,3}} = \ch{\Kac{1,1}} + \ch{\Kac{1,3}} + \ch{\Kac{1,5}}.
\end{equation}
However, each of the Kac modules appearing on the \rhs{} is reducible, with two (simple) composition factors each, so we learn that the fusion product has six composition factors in all:
\begin{equation} \label{CompFact:K13xK13}
\ch{\Kac{1,3} \fuse \Kac{1,3}} = 2 \: \ch{\Irr{0}} + 2 \: \ch{\Irr{1/2}} + \ch{\Irr{3/2}} + \ch{\Irr{5}}.
\end{equation}
Here, $\Irr{h}$ denotes the simple \hwm{} whose \hws{} has conformal weight $h$, as in \cref{sec:hwms}.  To understand how these six simple modules are glued together to form the fusion product, we will partially construct the product module and explicitly analyse the action of the \ns{} algebra upon it.

To construct the fusion product of $\Kac{1,3}$ with itself, we first calculate its special subspace.  This is defined (see \cref{app:NGK}) to be the (vector space) quotient of $\Kac{1,3}$ by the action of the algebra generated by the Virasoro modes $L_n$, with $n \le -2$, and superfield modes $G_j$, with $j \le -\frac{3}{2}$, leaving only linear combinations of vectors in which $L_{-1}$ and $G_{-1/2}$ act on $v$.  Imposing the \sv{} relation \eqref{SV:K13}, we find that $L_{-1} G_{-1/2} v = \frac{1}{2} G_{-3/2} v$ must be set to $0$ in the special subspace; generalising this shows that the special subspace is three-dimensional:
\begin{equation}
\spsub{\Kac{1,3}} = \vspn \set{v, G_{-1/2}v, L_{-1}v}.
\end{equation}

We will first determine the depth $0$ truncation of the fusion product $\Kac{1,3} \fuse \Kac{1,3}$.  Naturally enough, this requires the depth $0$ truncation of $\Kac{1,3}$.  This subspace is obtained by quotienting by the action of all \ns{} monomials with negative indices:
\begin{equation}
\Kac{1,3}^0 = \vspn \set{v}.
\end{equation}
The depth $0$ truncated fusion product is constructed within the tensor product of these two spaces.  Thus,
\begin{equation} \label{eq:BasisK13xK13}
\sqbrac{\Kac{1,3} \fuse \Kac{1,3}}^0 \subseteq \spsub{\Kac{1,3}} \otimes_{\CC} \Kac{1,3}^0 = \vspn \set{v \otimes v, L_{-1} v \otimes v \,\middle\vert\, G_{-1/2} v \otimes v}.
\end{equation}
The basis vectors here have been partitioned into even and odd parities, using a vertical delimiter, recalling that $v$ has been assumed to be even.  To determine which subspace of this three-dimensional tensor product is the depth $0$ fusion product, we search for spurious states.  These are linear dependence relations (see \cref{app:NGK}) that may be derived in $\spsub{\Kac{1,3}} \otimes_{\CC} \Kac{1,3}^0$ when we impose the \ns{} algebra action defined by the fusion coproduct formulae.  Inspection of the composition factors \eqref{CompFact:K13xK13} shows that there must be two vectors of conformal weight $0$ in the depth $0$ product, hence there can be at most one spurious state.

We search for spurious states by implementing the \sv{} relation \eqref{SV:K13}.  This will require the following cases of the master formulae \eqref{eq:Master} derived in \cref{app:Coprod}:
\begin{subequations}
\begin{align}
\coproduct{G_{-1/2}} &= G_{-1/2} \otimes \wun + \mu_1 \: \wun \otimes G_{-1/2}, \label{M1:G-1/2} \\
\coproduct{L_{-1}} &= L_{-1} \otimes \wun + \wun \otimes L_{-1}, \label{M1:L-1} \\
\coproduct{G_{-3/2}} &= G_{-1/2} \otimes \wun - \cdots + \mu_1 \: \wun \otimes G_{-3/2}, \label{M2:G-3/2} \\
G_{-3/2} \otimes \wun &= \coproduct{G_{-3/2}} + \cdots + \mu_1 \sqbrac{\wun \otimes G_{-1/2} + \cdots}. \label{M3:G-3/2}
\end{align}
\end{subequations}
Here, $\coproductsymb$ is the fusion coproduct, $\mu_1 = \pm 1$ is the parity of $w_1$ when the formula is applied to $w_1 \otimes w_2$, and the dots stand for infinite numbers of omitted terms which will not contribute to this calculation.  First, we note that all Virasoro and \ns{} modes, except $L_0$, will act as the zero operator on a depth $0$ space.  In particular, $\coproduct{G_{-1/2}} = \coproduct{L_{-1}} = \coproduct{G_{-3/2}} = 0$, so that
\begin{align} \label{eq:ExampleCalculation}
0 &= \coproduct{G_{-3/2}} v \otimes v = G_{-1/2} v \otimes v + v \otimes G_{-3/2} v = G_{-1/2} v \otimes v + 2 \: v \otimes L_{-1} G_{-1/2} v \notag \\
&= G_{-1/2} v \otimes v - 2 \: L_{-1} v \otimes G_{-1/2} v = G_{-1/2} v \otimes v + 2 \: G_{-1/2} L_{-1} v \otimes v = G_{-1/2} v \otimes v + 2 \: L_{-1} G_{-1/2} v \otimes v \notag \\
&= G_{-1/2} v \otimes v + G_{-3/2} v \otimes v = G_{-1/2} v \otimes v + v \otimes G_{-1/2} v = 0.
\end{align}
In this calculation, we have used \eqref{M2:G-3/2}, \eqref{SV:K13}, \eqref{M1:L-1}, \eqref{M1:G-1/2}, then the commutation relations \eqref{eq:CommN=1}, \eqref{SV:K13} again, \eqref{M3:G-3/2}, and finally \eqref{M1:G-1/2} again.  We have also assumed that the \hws{} $v$ has even parity.  In any case, the \rhs{} of \eqref{eq:ExampleCalculation} is identically zero which means that we have failed to find a spurious state.  Replacing $v \otimes v$ by $G_{-1/2} v \otimes v$ or $L_{-1} v \otimes v$ in this calculation likewise fails to uncover any spurious states.

We therefore assert that there are no spurious states to find and that the depth $0$ fusion product is three-dimensional.  It only remains to determine the action of $L_0$, that of the other modes being trivial.  To this end, we need three additional auxiliary formulae:
\begin{subequations}
\begin{align}
\coproduct{L_0} &= L_{-1} \otimes \wun + L_0 \otimes \wun + \wun \otimes L_0, \label{M1:L0} \\
L_{-2} \otimes \wun &= \coproduct{L_{-2}} + \cdots + \wun \otimes L_{-1} - \wun \otimes L_0 + \cdots, \label{M3:L-2} \\
L_{-1}^2 v &= -\frac{1}{2} G_{-3/2} G_{-1/2} v + L_{-2} v. \label{SV':K13}
\end{align}
\end{subequations}
The last is a consequence of the \sv{} relation \eqref{SV:K13}.  The action of $L_0$ is now given by
\begin{subequations}
\begin{align}
\coproduct{L_0} v \otimes v &= L_{-1} v \otimes v, \\
\coproduct{L_0} G_{-1/2} v \otimes v &= L_{-1} G_{-1/2} v \otimes v + \frac{1}{2} G_{-1/2} v \otimes v = \frac{1}{2} G_{-3/2} v \otimes v + \frac{1}{2} G_{-1/2} v \otimes v \notag \\
&= \frac{1}{2} v \otimes G_{-1/2} v + \frac{1}{2} G_{-1/2} v \otimes v = 0, \\
\coproduct{L_0} L_{-1} v \otimes v &= L_{-1}^2 v \otimes v + L_{-1} v \otimes v = -\frac{1}{2} G_{-3/2} G_{-1/2} v \otimes v + L_{-2} v \otimes v + L_{-1} v \otimes v \notag \\
&= \frac{1}{2} G_{-1/2} v \otimes G_{-1/2} v + v \otimes L_{-1} v + L_{-1} v \otimes v = \frac{1}{2} G_{-1/2}^2 v \otimes v = \frac{1}{2} L_{-1} v \otimes v.
\end{align}
\end{subequations}
With respect to the ordered basis \eqref{eq:BasisK13xK13}, we have
\begin{equation}
\coproduct{L_0} = 
\begin{amatrix}{cc|c}
0 & 0           & 0 \\
1 & \frac{1}{2} & 0 \\
\hline
0 & 0           & 0
\end{amatrix}
,
\end{equation}
where we have partitioned the matrix to indicate the separation into even and odd basis elements.  We conclude that the depth $0$ fusion product is spanned by two vectors of conformal weight $0$, one even and one odd, and one even vector of weight $\frac{1}{2}$.

The depth $0$ result therefore accounts for three of the six composition factors of the fusion product, namely both of the $\Irr{0}$ factors and one of the $\Irr{1/2}$ factors.  The remaining factors, $\Irr{1/2}$, $\Irr{3/2}$ and $\Irr{5}$, must appear as descendants of these via the action of the negative modes; otherwise, they would have appeared in the depth $0$ calculation.  The factor $\Irr{1/2}$ can only descend from one of the $\Irr{0}$ factors, but once this is fixed there are still three consistent possibilities, ignoring parities, for identifying $\Irr{3/2}$ and $\Irr{5}$ as descendants:
\begin{equation}
\parbox[c]{0.85\textwidth}{
\scalebox{0.85}{
\begin{tikzpicture}[->,node distance=1cm,>=stealth',semithick]
  \node[] (1) {$0$:};
  \node[] (1a) [below of=1] {$\frac{1}{2}$:};
  \node[] (1b) [below of=1a] {$\frac{3}{2}$:};
  \node[] (1c) [below of=1b] {$5$:};
  \node[sv] (2) [right = 1cm of 1] {};
  \node[sv] (2b) [right = 1cm of 1b] {};
  \path[] (2) edge (2b);
  \node[sv] (3) [right = 2.5cm of 1] {};
  \node[sv] (3a) [right = 2.5cm of 1a] {};
  \path[] (3) edge (3a);
  \node[sv] (4a) [right = 4cm of 1a] {};
  \node[sv] (4c) [right = 4cm of 1c] {};
  \path[] (4a) edge (4c);
  \node[sv] (5) [right = 6.5cm of 1] {};
  \node[sv] (6) [right = 8cm of 1] {};
  \node[sv] (6a) [right = 8.5cm of 1a] {};
  \node[sv] (6b) [right = 7.5cm of 1b] {};
  \path[] (6) edge (6a)
          (6) edge (6b);
  \node[sv] (7a) [right = 9.5cm of 1a] {};
  \node[sv] (7c) [right = 9.5cm of 1c] {};
  \path[] (7a) edge (7c);
  \node[sv] (8) [right = 12cm of 1] {};
  \node[sv] (9) [right = 13.5cm of 1] {};
  \node[sv] (9a) [right = 14cm of 1a] {};
  \node[sv] (9b) [right = 13cm of 1b] {};
  \node[sv] (9c) [right = 13.5cm of 1c] {};
  \path[] (9) edge (9a)
          (9) edge (9b)
          (9a) edge (9c)
          (9b) edge (9c);
  \node[sv] (10a) [right = 15cm of 1a] {};
\end{tikzpicture}
\ .}} \label{pic:3Poss}
\end{equation}
To distinguish between them, we must construct the fusion product to greater depth.

We therefore turn to the depth $\frac{1}{2}$ calculation in which vectors are set to zero if they may be obtained from other vectors by acting with linear combinations of \ns{} monomials whose indices are negative and sum to at most $-1$.  The special subspace of $\Kac{1,3}$ does not change, but now we consider its depth $\frac{1}{2}$ truncation which is spanned by $v$ and $G_{-1/2} v$.  Thus, the depth $\frac{1}{2}$ fusion product will be contained within a six-dimensional space:
\begin{equation} \label{eq:BasisK13xK13'}
\sqbrac{\Kac{1,3} \fuse \Kac{1,3}}^{1/2} \subseteq \vspn \set{v \otimes v, L_{-1} v \otimes v, G_{-1/2} v \otimes G_{-1/2} v \,\middle\vert\, G_{-1/2} v \otimes v, v \otimes G_{-1/2} v, L_{-1} v \otimes G_{-1/2} v}.
\end{equation}
Comparing with the three possible structures \eqref{pic:3Poss} for the fusion product, we see that allowing descendants by $G_{-1/2}$, as well as the depth zero vectors, always leads to five depth $\frac{1}{2}$ vectors with conformal weights $0$, $0$, $\frac{1}{2}$, $\frac{1}{2}$ and $1$.  This indicates that there is precisely one spurious state to find.

The calculation proceeds in much the same manner as before.  The difference is that because we are computing to depth $\frac{1}{2}$, we may no longer assert that $\coproduct{G_{-1/2}} = 0$ (nor that $\coproduct{G_{+1/2}} = 0$).  Using \eqref{M2:G-3/2}, \eqref{SV:K13} and \eqref{M1:L-1}, we quickly arrive at
\begin{align}
0 &= \coproduct{G_{-3/2}} v \otimes v = G_{-1/2} v \otimes v + v \otimes G_{-3/2} v = G_{-1/2} v \otimes v + 2 \: v \otimes L_{-1} G_{-1/2} v \notag \\
&= G_{-1/2} v \otimes v - 2 \: L_{-1} v \otimes G_{-1/2} v.
\end{align}
The \rhs{} has been expressed in terms of the basis elements \eqref{eq:BasisK13xK13'} and the fact that it does not vanish identically means that we have found a spurious state.  More precisely, it means that this relation must be imposed in the depth $\frac{1}{2}$ fusion product.  We have searched for more independent spurious states, but found none in accord with the structural arguments above.

Imposing this relation reduces the dimension of the depth $\frac{1}{2}$ fusion product from $6$ to $5$.  Computing the action of $L_0$ on this space is now straight-forward.  With respect to the ordered basis consisting of the first five elements of the \rhs{} of \eqref{eq:BasisK13xK13'}, we obtain
\begin{equation}
\coproduct{L_0} = 
\begin{amatrix}{ccc|cc}
0 & 0           & 0            & 0           & 0           \\
1 & 0           & -\frac{1}{2} & 0           & 0           \\
0 & \frac{1}{2} & 1            & 0           & 0           \\
\hline
0 & 0           & 0            & \frac{1}{2} & \frac{1}{2} \\
0 & 0           & 0            & \frac{1}{2} & \frac{1}{2}
\end{amatrix}
\sim
\begin{amatrix}{ccc|cc}
0 & 0           & 0           & 0 & 0 \\
0 & \frac{1}{2} & 1           & 0 & 0 \\
0 & 0           & \frac{1}{2} & 0 & 0 \\
\hline
0 & 0           & 0           & 0 & 0 \\
0 & 0           & 0           & 0 & 1
\end{amatrix}
,
\end{equation}
where we also indicate the Jordan canonical form.  While we do find the expected conformal weights, more interesting is the presence of a rank 2 Jordan block for the weight $\frac{1}{2}$, indicating the presence of a staggered submodule in the fusion product (see \cref{app:StagMod}).

To determine which of the three possibilities of \eqref{pic:3Poss} is realised by the fusion product, we can repeat the above computations to depth $\frac{3}{2}$ and show that the fusion product has no submodule isomorphic to $\Irr{0}$, that is that no weight $0$ vector is annihilated by both singular combinations $G_{-1/2}$ and $L_{-1} G_{-1/2} - \frac{1}{2} G_{-3/2}$.  This fact implies that the fusion product corresponds to the leftmost possibility in \eqref{pic:3Poss}.  We mention that the required computation is rather tedious by hand, involving one spurious state in a $12$-dimensional space, but is practically instantaneous in our computer algebra implementation.

The full structure of the fusion product is
\begin{equation} \label{FR:K13xK13}
\parbox[c]{0.25\textwidth}{
\scalebox{0.75}{
\begin{tikzpicture}[->,node distance=1cm,>=stealth',semithick]
  \node[] (1) {$0$:};
  \node[] (1a) [below of=1] {$\frac{1}{2}$:};
  \node[] (1b) [below of=1a] {$\frac{3}{2}$:};
  \node[] (1c) [below of=1b] {$5$:};
  \node[sv] (2) [right = 1cm of 1] {};
  \node[sv] (2b) [right = 1cm of 1b] {};
  \path[] (2) edge (2b);
  \node[sv] (3) [right = 2.5cm of 1] {};
  \node[sv] (3a) [right = 2.5cm of 1a] {};
  \path[] (3) edge (3a);
  \node[sv] (4a) [right = 4cm of 1a] {};
  \node[sv] (4c) [right = 4cm of 1c] {};
  \path[] (4a) edge (4c)
          (4a) edge (3)
          (4a) edge (3a)
          (4c) edge (3a);
\end{tikzpicture}
}}
\qquad \Ra \qquad \Kac{1,3} \fuse \Kac{1,3} = \Kac{1,1} \oplus \Stag{1,4}{0,1},
\end{equation}
where $\Stag{1,4}{0,1}$ denotes a staggered module described by the short exact sequence
\begin{equation} \label{ses:K13K15}
\dses{\Kac{1,3}}{}{\Stag{1,4}{0,1}}{}{\Kac{1,5}}.
\end{equation}
The notation here derives from \eqref{ses:K13K15} in that $\Stag{1,4}{0,1}$ has a submodule isomorphic to the Kac module with labels $(r,s) = (1,4) - (0,1) = (1,3)$ and the quotient by this submodule is isomorphic to the Kac module with labels $(r,s) = (1,4) + (0,1) = (1,5)$.  We will use the obvious extension of this notation to describe more general staggered modules in what follows (see \cref{app:Stag}).\footnote{We emphasise that this notation differs from a similar notation $\mathcal{R}_{r,s}^{a,b}$ that has been used to indicate certain modules over the Virasoro \cite{RasFus07,RasFus07b} and \ns{} \cite{PeaLog14} algebras. These modules are believed to arise in the continuum scaling limit of certain statistical models via a lattice fusion prescription and are conjectured to have Jordan blocks for $L_0$ of rank $2$, if exactly one of $a$ and $b$ is non-zero, and rank $3$, if both $a$ and $b$ are non-zero. If the rank-$2$ module $\mathcal{R}_{r,s}^{a,b}$ is staggered, then the two notations are believed to agree: $\mathcal{R}_{r,s}^{a,b}=\Stag{r,s}{a,b}$.}

With our depth $\frac{3}{2}$ computation, we are now able to check every aspect of \eqref{FR:K13xK13} except for explicitly verifying the arrow from the \ssv{} of conformal weight $5$ to the $L_0$-eigenvector of weight $\frac{1}{2}$.  Unfortunately, this would require computing fusion truncations to depth $5$ which is well beyond the current limits of our computer.  Below, we will discuss an alternative means of checking that this arrow is present.

However, having determined that the fusion product involves a staggered module $\Stag{1,4}{0,1}$, we have to determine if the structure depicted in \eqref{FR:K13xK13}, equivalently if the short exact sequence \eqref{ses:K13K15}, completely specifies its isomorphism class.  The general theory states that this isomorphism class is characterised by its logarithmic coupling $\logcoup{1,4}{0,1}$ which may be determined, in this example, as follows (see \cref{app:LogCoup} for generalities):  Let $x \in \Stag{1,4}{0,1}$ have conformal weight $0$, so that $G_{-1/2} x$ is singular.  Choose any $y \in \Stag{1,4}{0,1}$ satisfying $\brac{L_0 - \frac{1}{2}} y = G_{-1/2} x$ and note that $G_{1/2} y$ must be proportional to $x$.  The constant of proportionality is $\logcoup{1,4}{0,1}$.

We may compute $\logcoup{1,4}{0,1}$ within the depth $\frac{1}{2}$ fusion product by computing $\coproduct{G_{-1/2}}$ and $\coproduct{G_{+1/2}}$.  In the basis consisting of the first five elements of the \rhs{} of \eqref{eq:BasisK13xK13'}, we find that
\begin{equation}
\coproduct{G_{-1/2}} = 
\begin{amatrix}{ccc|cc}
0 & 0           & 0           & 0  & 0  \\
0 & 0           & 0           & 1  & -1 \\
0 & 0           & 0			  & -1 & 1  \\
\hline
1 & \frac{1}{2} & \frac{1}{2} & 0  & 0  \\
0 & \frac{1}{2} & \frac{1}{2} & 0  & 0		 		   		  
\end{amatrix}
, \qquad \coproduct{G_{+1/2}} = 
\begin{amatrix}{ccc|cc}
0 & 0           & 0           & 0 & 0 \\
0 & 0           & 0           & 1 & 0 \\
0 & 0           & 0			  & 0 & 1 \\
\hline
1 & 1           & \frac{1}{2} & 0 & 0 \\
0 & \frac{1}{2} & 0           & 0 & 0		 		   		  
\end{amatrix}
.
\end{equation}
The element $x \in \Stag{1,4}{0,1}$ may be identified (in the depth $\frac{1}{2}$ product) with the vector $(0,0,0\,\vert\, 1,-1)^T$.  We then solve $\brac{L_0 - \frac{1}{2}} y = G_{-1/2} x$, giving $y = (0,-2,-2\,\vert\, 0,0)^T$ modulo arbitrary multiples of $G_{-1/2} x$ and so $G_{1/2} y = (0,0,0\,\vert\, -3,-1)^T = -x-2z$, where $z = (0,0,0\,\vert\, 1,1)^T$ has conformal weight $1$.  Now, this appears to contradict the fact that $G_{1/2} y$ must be proportional to $x$.  This is down to a subtlety with the computation of $\coproduct{G_{1/2}}$.  This mode should not be regarded as mapping the depth $\frac{1}{2}$ fusion product into itself, but rather as a map from the depth $\frac{1}{2}$ product into the depth $0$ product.  The vector $z$, being of conformal weight $1$ and thus not in the depth $0$ product, should therefore be set to $0$ in order to arrive at the correct result:  $G_{1/2} y = -x$.  We therefore conclude that $\logcoup{1,4}{0,1} = -1$ and identify the fusion product as
\begin{equation} \label{FR:K13xK13'}
\Kac{1,3} \fuse \Kac{1,3} = \Kac{1,1} \oplus \Stag{1,4}{0,1}(-1).
\end{equation}

Actually, we can refine this even further by keeping track of parities.  We assumed in our computations that the minimal conformal weight vectors of both copies of $\Kac{1,3}$ were even.  The same is true for the summand $\Kac{1,1}$ found above, though the vector $x$ of minimal conformal weight in $\Stag{1,4}{0,1}(-1)$ is odd.  A maximally precise version of \eqref{ses:K13K15} and \eqref{FR:K13xK13'} is therefore
\begin{equation}
\Kac{1,3}^+ \fuse \Kac{1,3}^+ = \Kac{1,1}^+ \oplus \Stag{1,4}{0,1}(-1)^-, \qquad
\dses{\Kac{1,3}^-}{}{\Stag{1,4}{0,1}(-1)^-}{}{\Kac{1,5}^+}.
\end{equation}

We have also confirmed the logarithmic coupling $\logcoup{1,4}{0,1} = -1$ using the method described in \cite{RidLog07}.  This succeeds because, in this case, the exact sequence \eqref{ses:K13K15} fixes the isomorphism class of the staggered module $\Stag{1,4}{0,1}$ completely.   We omit the rather tedious calculations.  Instead, we indicate how to confirm this value using the heuristic formula \eqref{eq:LogCoupForm} developed in \cite{VasInd11}.  For this, we first perturb the parameter $t$, hence the central charge $c$ and Kac weights $h_{r,s}$, as in \eqref{eq:Perturb} to $\func{t}{\eps} = t + \eps = \frac{1}{2} + \eps$.  We then compute the scalar product
\begin{equation}
\inner{\func{x}{\eps}}{G_{1/2} G_{-1/2} \func{x}{\eps}} = \inner{\func{x}{\eps}}{2 L_0 \func{x}{\eps}} = 2 \func{h_{1,3}}{\eps} = 2 \eps
\end{equation}
in the (poorly characterised) perturbed theory in which $\func{x}{\eps}$ is a \hws{} of conformal weight $\func{h_{1,3}}{\eps} = \eps$.  Substituting into \eqref{eq:LogCoupForm}, we obtain
\begin{equation}
\logcoup{1,4}{0,1} = \frac{8t^2}{0 - \brac{5^2-3^2} t^2} \lim_{\eps \ra 0} \frac{2 \eps}{\eps} = -1,
\end{equation}
in agreement with the explicit depth $\frac{1}{2}$ fusion construction.

Finally, we point out that one can also arrive at the leftmost possibility in \eqref{pic:3Poss}, without performing depth $\frac{3}{2}$ calculations, by instead appealing to the theory of staggered modules (\cref{app:StagMod}).  This is generally far more efficient than explicitly constructing the truncated fusion product, especially when the depth required for a complete identification becomes large.  First, there are two independent vectors of conformal weight $\frac{1}{2}$ and we know that only one, $w$ say, is descended from a weight $0$ vector.  As the weight $0$ vectors are eigenvectors of $L_0$, so is $w$.  Thus, $w$ is the $L_0$-eigenvector in the Jordan block at weight $\frac{1}{2}$.  Now, choose a Jordan partner $y$ for $w$ so that $\brac{L_0 - \frac{1}{2}} y = w$.  Then, \cref{prop:StagAnn} shows that $w$ cannot have a singular descendant of weight $5$ unless this descendant also has a Jordan partner (take $U$ there to be the singular combination of weight $\frac{9}{2}$ that annihilates $\pi y$).  This rules out the rightmost diagram in \eqref{pic:3Poss} --- the corresponding staggered module simply does not exist.

One can similarly rule out the middle diagram using some deeper structural results for staggered modules.  In this case, the purported staggered module $\Stag{}{}(-1)$ would be described by the following exact sequence:
\begin{equation} \label{es:NonExistentStag}
\dses{\frac{\Ver{0}}{\Ver{3} + \Ver{5}}}{}{\Stag{}{}(-1)}{}{\frac{\Ver{1/2}}{\Ver{3}}}.
\end{equation}
Here, we indicate the required logarithmic coupling which was obtained from a depth $\frac{1}{2}$ computation.  If we replace the third module in this sequence by its Verma cover $\Ver{1/2}$, then a staggered module $\Stag{}{}(-1)'$ with this new exact sequence may be shown to exist (and be unique up to isomorphism) using the same methods that were employed in \cite{RidSta09} for Virasoro staggered modules.  Moreover, a staggered module $\Stag{}{}(-1)$ with sequence \eqref{es:NonExistentStag} will exist if and only if there exists a \sv{} of conformal weight $3$ in $\Stag{}{}(-1)'$.  It is not difficult to check this explicitly with a computer implementation of $\Stag{}{}(-1)'$ --- the result is that no such \sv{} exists, hence that $\Stag{}{}(-1)$ does not exist either.  This rules out the middle diagram.

A similar calculation may be used to verify the presence of the arrow in \eqref{FR:K13xK13} from the conformal weight $5$ \ssv{} to the weight $\frac{1}{2}$ \sv{} $G_{-1/2} x$.  If this arrow were not present, then this weight $5$ vector would have to be singular in $\Stag{1,4}{0,1}(-1)$.  Again, an explicit search for a \sv{} of this weight leads to no solutions, thus verifying the arrow.  In principle, the arrow could instead point to the weight $0$ \sv{} $x \in \Stag{1,4}{0,1}(-1)$.  However, this is easy to rule out because all of the positive weight vectors in the submodule generated by $x$ actually belong to that generated by $G_{-1/2} x$.  More generally, we may appeal to the \ns{} generalisation of the Projection Lemma \cite[Lem.~5.1]{RidSta09} to identify the targets of such arrows (assuming we have shown that said arrows exist).

To summarise, we have completely identified the fusion product considered above by combining information obtained from four distinct sources.  First, the Verlinde formula \eqref{eq:TheVerlindeFormula2} decomposed the corresponding Grothendieck fusion product, giving us the composition factors of the fusion product.  Second, applying the \NGK{} algorithm to depth $0$ indicated which composition factors were descendants of others and which were not, while the depth $\frac{1}{2}$ computation uncovered a rank $2$ Jordan block for the action of $L_0$ and determined the corresponding logarithmic coupling.  We also determined the parities of the indecomposable direct summands of the fusion product from these computations.  Third, the structure theory for staggered modules allowed us to completely fix the rest of the structure of the fusion product, except for one arrow in our diagrammatic representation of the structure that may or may not have been present.  Fourth, the presence of this arrow was confirmed by explicitly showing the non-existence of a \sv{}, of the appropriate conformal weight, in the staggered module.

The logic that led us to the structure of this fusion product is fairly typical.  We have employed it to analyse many further examples of \ns{} fusion products.  The conclusions that we have drawn from these analyses are reported in the next section.

\section{Results} \label{sec:Results}

In this section, we summarise the results that we have obtained by combining the character product rules \eqref{GrFR:KxK} with explicit \NGK{} fusion computations and the structure theory of staggered modules.  As was explained in the previous section, this combination allows us to significantly reduce the depth to which the fusion algorithm must be applied in order to completely identify the product.  For brevity, we have only considered fusion rules between Kac modules, restricting to the central charges $c = \frac{3}{2}$, $-\frac{5}{2}$, $-\frac{81}{10}$, $0$, $-\frac{21}{4}$ and $\frac{7}{10}$, corresponding to $(p,p') = (1,1)$, $(1,3)$, $(1,5)$, $(2,4)$, $(2,8)$ and $(3,5)$, respectively.  The results obtained suggest conjectures for certain classes of general Kac fusion rules which we describe below.

\subsection{Fusing $\bm{\Kac{r,1}}$ with $\bm{\Kac{1,s}}$} \label{sec:Kr1K1s}

Perhaps the simplest Kac module fusion products are those involving a ``first row'' module and a ``first column'' one.  In this case, the proposed fusion formalism for the underlying lattice models \cite{PeaLog14,MorKac15} 
requires the following fusion rule for consistency:
\begin{equation} \label{FR:Kr1K1s}
\Kac{r,1}\fuse \Kac{1,s} = \Kac{r,s} \qquad \text{(\(r,s \in 2 \ZZ_+ - 1\)).}
\end{equation}
This is certainly consistent with the corresponding character product rule \eqref{GrFR:Kr1xK1s} and we have verified it explicitly, using the \NGK{} algorithm, in many cases (see below).  The evidence is, in our opinion, sufficient to conjecture that \eqref{FR:Kr1K1s} holds in complete generality. Whilst this accords with the proposed lattice fusion calculations, we view the result as also confirming, indirectly, that we have made the correct abstract definition for Kac modules.  

We remark that when confirming the fusion product \eqref{FR:Kr1K1s}, the case in which $\Kac{r,s}$ is a corner type module, hence is semisimple, is the most computationally intensive.  To explicitly verify that each composition factor splits off as a direct sum, thereby forming the required collection of islands (as indicated in \cref{fig:KacStructures}), we must compute to the depth given by the maximal difference between the conformal weights of consecutive composition factors (when they are ordered by their conformal weight).  This quickly becomes infeasible with our implementation as $r$ and $s$ grow, so the direct evidence for corner type modules is somewhat less compelling than for the other cases.

We also note here that the computational complexity of the fusion algorithm means that we were only able to successfully confirm \eqref{FR:Kr1K1s} when the required depth was at most $4$.  However, as $p$ or $p'$ increases, the labels $r$ and $s$ requiring a given depth calculation tend to increase leading to an overall steady decrease in the feasible depths due to the increasing complexity of the \svs{} for high $r$ and $s$.  For $(p,p') = (2,8)$ and $(3,5)$, we were therefore limited to depths at most $2$.

Including parity in the fusion rule \eqref{FR:Kr1K1s} is easy:  In each case, explicit computation confirms that
\begin{equation}
\Kac{r,1}^+\fuse \Kac{1,s}^+ = \Kac{r,s}^+ \qquad \text{(\(r,s \in 2 \ZZ_+ - 1\)).}
\end{equation}
More generally, if the parities of $\Kac{r,1}$ and $\Kac{1,s}$ coincide (differ), then that of $\Kac{r,s}$ will be even (odd).  
This observation does not appear to have a simple explanation in terms of the fusion algorithm, although it is in accord with the well known principle of conservation of fermion numbers.  We will show in \cite{CanFusII15} that it follows readily from a fermionic version of the Verlinde formula.  From the lattice, it follows from considerations of the parity of the system size \cite{PeaLog14}.

To illustrate some of the simpler issues that arise with fusion computations of the form \eqref{FR:Kr1K1s}, we consider the example $\Kac{3,1}\fuse \Kac{1,3}$, for $(p,p') = (2,4)$ ($c=0$). The Verlinde formula tells us that the character is that of $\Kac{3,3}$ which means that the fusion product has two composition factors, $\Irr{1/2}$ and $\Irr{3}$. However, there are three inequivalent structural possibilities:
\begin{equation}
\parbox[c]{0.3\textwidth}{
\begin{tikzpicture}[->,node distance=1cm,scale = 0.8, transform shape,>=stealth',semithick]
  \node[main node] (1) {};
  \node[main node] (1a) [below of=1] {};
  \path[] (1) edge (1a);
  \node[sv] (2) [right = 1.5cm of 1] {};
  \node[sv] (2a) [right = 1.5cm of 1a] {};
  \path[] (2a) edge (2);
  \node[sv] (3) [right = 1.5cm of 2] {};
  \node[sv] (3a) [right = 1.5cm of 2a] {};
  \node (0) [left = 1cm of 1] {$\tfrac{1}{2}$:};
  \node     [below of=0] {$3$:};
\end{tikzpicture}
\ .}
\end{equation}
The first possibility is $\Kac{3,3}$, the expected result, the second is its contragredient dual, and the third is the direct sum $\Irr{1/2} \oplus \Irr{3}$. Constructing $\Kac{3,1}\fuse \Kac{1,3}$ to depth $0$ leads to a one-dimensional truncated space.  Comparing with the possible structures, we see that this is consistent with $\Kac{3,3}$ where the $\Irr{3}$ factor is descended from $\Irr{1/2}$, hence does not appear at depth $0$. Moreover, the depth $0$ truncations of the other two structural possibilities are two-dimensional. For this reason, a depth 0 calculation alone is sufficient to confirm that $\Kac{3,1}\fuse \Kac{1,3} = \Kac{3,3}$.

A slightly more complicated example is the $(p,p') = (1,3)$ ($c=-5/2$) fusion product $\Kac{3,1} \fuse \Kac{1,5}$.  The character is that of $\Kac{3,5}$, so the composition factors of the product are $\Irr{1/2}$, $\Irr{5/2}$ and $\Irr{4}$.  Depth $0$ computations reveal a two-dimensional truncated space with conformal weights $1/2$ and $5/2$.  It follows from the general structure theory that $\Irr{4}$ is descended from $\Irr{5/2}$, so it only remains to decide whether $\Irr{1/2}$ splits off as a direct summand or whether it is generated from $\Irr{5/2}$ through the action of the positive modes.  This requires computing to depth $2$ and the result indicates that $\Irr{1/2}$ is not a direct summand.  The structure is therefore
\begin{equation}
\parbox[c]{0.15\textwidth}{
\begin{tikzpicture}  [->,-stealth', scale = 0.7, transform shape, node distance=1.1cm]
  \node[main node, label=left:$\frac{1}{2}$] (3) [] {};
  \node[csv] (3a) [below of=3, label=left:$\frac{5}{2}$] {};
  \node[main node] (3b) [below of=3a, label=left:$4$] {};  
  \path[] (3a) edge node {} (3)
          (3a) edge node {} (3b);
\end{tikzpicture}
\ ,}
\end{equation}
confirming the fusion rule $\Kac{3,1} \fuse \Kac{1,5} = \Kac{3,5}$.

In the braided case, there may be further obstacles to overcome in completely identifying the structure of the fusion product.  For example, the product $\Kac{3,1} \fuse \Kac{1,5}$ at $(p,p') = (2,4)$ ($c=0$) involves six composition factors:  $\Irr{0}$, $\Irr{1/2}$, $\Irr{3/2}$, $\Irr{3}$, $\Irr{5}$ and $\Irr{21/2}$.  A depth $0$ calculation indicates that all but $\Irr{0}$ and $\Irr{1/2}$ are descendants and a depth $\frac{1}{2}$ calculation shows that $G_{1/2}$ maps the weight $\frac{1}{2}$ vector to that of weight $0$.  The (partial) structure of the fusion product is thus
\begin{equation} \label{Str:K31xK15c=0}
\parbox[c]{0.15\textwidth}{
\begin{tikzpicture}  [->,-stealth', scale = 0.7, transform shape, node distance=1.1cm]
 \node[main node, label = above:$0$] (1f) [] {};
 \node[csv, label = left:$\frac{1}{2}$] (2f) [below left of=1f] {};
 \node[main node, label = right:$\frac{3}{2}$] (3f) [below right of=1f] {};
 \node[subsv, label = left:$3$] (4f) [below of=2f] {};
 \node[subsv, label = right:$5$] (5f) [below of=3f] {};
 \node[subsv, label = below:$\frac{21}{2}$] (6f) [below left of=5f] {};
 \path[]
   (2f) edge node {} (1f)
   (1f) edge node {} (3f)
   (2f) edge node {} (4f)
   (2f) edge node {} (5f)
   (5f) edge node {} (6f)
   (4f) edge node {} (6f);
\end{tikzpicture}
\ ,}
\end{equation}
where it only remains to determine if there are upwards-pointing arrows emanating from the three lowest nodes.  The $N=1$ version of the Projection Lemma of \cite[Lem.~5.1]{RidSta09} rules out any such arrow from the node labelled by $\frac{21}{2}$ and from the nodes labelled by $3$ and $5$ to that labelled by $0$.  There are thus only two possible arrows:  those from $3$ or $5$ to $\frac{3}{2}$.  The presence of these arrows may be ascertained as in \cite{RidLog07}, see also \cite[Sec.~4.2.2]{MorKac15}.  If they are absent, then the fusion product would possess a \sv{} of weight $3$ or $5$, respectively.  By computing within the most general abstract module with structure \eqref{Str:K31xK15c=0} (an arbitrary extension of the \hwm{} $\Ver{1/2} / \Ver{15/2}$ by the \hwm{} $\Kac{1,1} \cong \Ver{0} / \Ver{1/2}$), we can explicitly verify that such \svs{} do not exist.  This is a relatively efficient calculation for a computer; in particular, we do not need to invoke the \NGK{} algorithm.  In this way, we arrive at the structure
\begin{equation}
\parbox[c]{0.15\textwidth}{
\begin{tikzpicture}  [->,-stealth', scale = 0.7, transform shape, node distance=1.1cm]
 \node[main node, label = above:$0$] (1f) [] {};
 \node[csv, label = left:$\frac{1}{2}$] (2f) [below left of=1f] {};
 \node[main node, label = right:$\frac{3}{2}$] (3f) [below right of=1f] {};
 \node[subsv, label = left:$3$] (4f) [below of=2f] {};
 \node[subsv, label = right:$5$] (5f) [below of=3f] {};
 \node[subsv, label = below:$\frac{21}{2}$] (6f) [below left of=5f] {};
 \path[]
   (2f) edge node {} (1f)
   (1f) edge node {} (3f)
   (2f) edge node {} (4f)
   (2f) edge node {} (5f)
   (4f) edge node {} (3f)
   (5f) edge node {} (3f)
   (5f) edge node {} (6f)
   (4f) edge node {} (6f);
\end{tikzpicture}
\ ,}
\end{equation}
from which we conclude that $\Kac{3,1} \fuse \Kac{1,5} = \Kac{3,5}$.

\subsection{Fusing near the edge} \label{sec:Near}

We do not have enough data to make conjectures concerning the general fusion rules of the Kac modules (see the following section for some complicated examples).  However, we have observed some patterns that seem to be followed when the Kac modules to be fused lie sufficiently close to the edges of the extended Kac table.  Below, we will indicate what this means precisely. For brevity, we will refer to this situation as fusing ``near the edge'' (of the extended Kac table).

Recall the following specialisations of the Kac character rules \eqref{GrFR:KxK}:
\begin{subequations}
\begin{align}
\ch{\Kac{r,1}} \Grfuse \ch{\Kac{r',s'}} &= \sideset{}{'} \sum_{r'' = \abs{r-r'}+1}^{r+r'-1} \ch{\Kac{r'',s'}}, \label{GrFR:Kr1} \\
\ch{\Kac{1,s}} \Grfuse \ch{\Kac{r',s'}} &= \sideset{}{'} \sum_{s'' = \abs{s-s'}+1}^{s+s'-1} \ch{\Kac{r',s''}}. \label{GrFR:K1s}
\end{align}
\end{subequations}
As usual, primed sums indicate that the index increases in steps of two.  If the Kac modules $\Kac{r'',s'}$ ($\Kac{r',s''}$) that appear in the above decompositions all satisfy either $r'' \le p$ or $s' \le p'$ ($r' \le p$ or $s'' \le p'$), then we conjecture that the fusion rule corresponding to \cref{GrFR:Kr1} (\cref{GrFR:K1s}) may be determined through the following procedure:
\begin{enumerate}[leftmargin=*,label=\arabic*)]
\item Write down a list of all the Kac modules $\Kac{r'',s'}$ ($\Kac{r',s''}$) from the decomposition \eqref{GrFR:KxK} in order of increasing $r''$ ($s''$). \label{it:KacList}
\item Starting from the \emph{smallest} value of $r''$ ($s''$), check whether there exists a $\Kac{\rho'',s'}$ ($\Kac{r',\sigma''}$) in the list which is the reflection of $\Kac{r'',s'}$ ($\Kac{r',s''}$) about the next boundary.  This means that $\rho''$ ($\sigma''$) must satisfy $0 < \rho'' - r'' < 2p$ ($0 < \sigma'' -s'' < 2p'$) and $\Kac{\frac{1}{2} (r'' + \rho''),s'}$ ($\Kac{r',\frac{1}{2} (s'' + \sigma'')}$) must be of boundary or corner type.\footnote{It is worth mentioning here that this boundary or corner type module may well be from the Ramond sector.}
\item If there does, then replace $\Kac{r'',s'}$ and $\Kac{\rho'',s'}$ ($\Kac{r',s''}$ and $\Kac{r',\sigma''}$) in the list by the staggered module $\Stag{\frac{1}{2} (\rho'' + r''), s'}{\frac{1}{2} (\rho'' - r''),0}$ ($\Stag{r',\frac{1}{2} (\sigma'' + s'')}{0,\frac{1}{2} (\sigma'' - s'')}$).  Any logarithmic coupling must be determined through other means.
\item Repeat with $\Kac{r'',s'}$ ($\Kac{r',s''}$), where $r''$ ($s''$) is the next-highest value.  Once all values are exhausted, the list consists of the direct summands of the fusion product.
\end{enumerate}
Similar conjectures were made for certain classes of Virasoro Kac modules in \cite{EbeVir06,RasFus07,RasFus07b,RidPer07}.  We have checked that this procedure gives results that are consistent (up to the values of any logarithmic couplings) with our explicit fusion computations.  \cref{app:Results} lists, for each central charge considered (except $c=\frac{3}{2}$ for reasons that are explained in \cref{sec:32}), those computations for which these checks have been performed.

We illustrate this procedure with a few examples.  First, take $(p,p')=(1,3)$ ($c=-\frac{5}{2}$) and consider $\Kac{1,3} \fuse \Kac{1,5}$.  The list, for this product, is $\Kac{1,3}$, $\Kac{1,5}$, $\Kac{1,7}$, according to \eqref{GrFR:K1s}.  Since $\Kac{1,3}$ is of corner type, it has no reflection.  However, $\Kac{1,5}$ reflects onto $\Kac{1,7}$ about the (Ramond) corner type module $\Kac{1,6}$, so they are replaced by $\Stag{1,6}{0,1}$.  Computing the logarithmic coupling using \NGK{} fusion or \eqref{eq:LogCoupForm} then gives
\begin{equation} \label{FR:K13K15}
\Kac{1,3} \fuse \Kac{1,5} = \Kac{1,3} \oplus \Stag{1,6}{0,1}(-2) \qquad \text{(\(c = -\tfrac{5}{2}\)).}
\end{equation}
Similar arguments for $(p,p')=(2,4)$ result in
\begin{equation}
\Kac{3,1} \fuse \Kac{3,3} = \Stag{2,3}{1,0}(\tfrac{1}{2}) \oplus \Kac{5,3}, \quad
\Kac{1,5} \fuse \Kac{1,5} = \Stag{1,4}{0,3}(-4) \oplus \Stag{1,4}{0,1}(-1) \oplus \Kac{1,9} \qquad \text{(\(c = 0\)).}
\end{equation}
We mention that in these examples, the original list of Kac modules did not by itself uniquely determine which list members are combined to form a staggered module.  This is where it is important to start the above procedure with the smallest value of the appropriate Kac label.  For example, with $\Kac{1,5} \fuse \Kac{1,5}$, we combined $\Kac{1,1}$ with $\Kac{1,7}$ to correctly identify $\Stag{1,4}{0,3}(-4)$ as a direct summand, instead of combining $\Kac{1,7}$ with $\Kac{1,9}$.

The parities of the modules obtained by fusing near the edge of the extended Kac table are easily determined.  Assuming that the Kac modules being fused are both assigned an even parity, we find that the parities of the Kac modules in the ordered list constructed in step \ref{it:KacList} above always alternate, starting (and therefore ending) with even parity.  This is also consistent with lattice expectations \cite{PeaLog14}.  As an example, take $(p,p')=(2,4)$ ($c=0$) and consider the fusion product $\Kac{1,3}^+ \fuse \Kac{2,4}^+$.  This yields the list $\Kac{2,2}^+$, $\Kac{2,4}^-$, $\Kac{2,6}^+$ from which we deduce that the first and last Kac modules combine to form a staggered module.  As $h_{2,2} = h_{2,6}$ and both $\Kac{2,2}$ and $\Kac{2,6}$ are \hwms{} (see \cref{fig:Kacc=0}), there is no logarithmic coupling to find and the final result is
\begin{equation}
\Kac{1,3}^+ \fuse \Kac{2,4}^+ = {\Stag{2,4}{0,2}}^+ \oplus \Kac{2,4}^- \qquad \text{(\(c = 0\)).}
\end{equation}
As in the previous section, this parity rule will be derived in \cite{CanFusII15} from a fermionic Verlinde formula.

The (conjectural) procedure described above for fusing near the edge of the extended Kac table implies that the resulting fusion products may only admit Jordan blocks of rank at most $2$ for the action of $L_0$.  More precisely, it implies that these products always decompose as direct sums of Kac modules and staggered modules.  We have observed such staggered modules in every model considered except for that with $(p,p')=(1,1)$ ($c=\frac{3}{2}$).  This exceptional case is discussed separately in \cref{sec:32} below.

Finally, the procedure proposed above, for determining Kac fusion rules near the edge of the extended Kac table, does not require restricting to fusing with either $\Kac{r,1}$ or $\Kac{1,s}$, although the Kac module labels $(r'',s'')$ appearing in the character product \eqref{GrFR:KxK} may need to satisfy either $r'' \le p$ or $s'' \le p'$.  An example illustrating this is $(p,p')=(2,4)$ ($c=0$) with $\Kac{2,2} \fuse \Kac{2,2}$.  The corresponding list in this case is $\Kac{1,1}$, $\Kac{3,1}$, $\Kac{1,3}$, $\Kac{3,3}$ and an explicit calculation shows that the first two and last two members combine to form staggered modules as one would expect from the boundary reflection principle:
\begin{equation}
\Kac{2,2} \fuse \Kac{2,2} = \Stag{2,1}{1,0}(\tfrac{3}{8}) \oplus \Stag{2,3}{1,0}(\tfrac{1}{2}) \qquad \text{(\(c = 0\)).}
\end{equation}
Further examples like this are common for larger $p$ and $p'$ and we have checked in several cases that the fusion decompositions do lead to staggered modules whenever two Kac modules in the list (not ordered in these examples) are related by reflection about a boundary, see \cref{app:Results}.  However, the number of examples that we are able to fully analyse is not particularly large, explaining why we have not included these observations in the conjectured procedure.

\subsection{Fusing away from the edge} \label{sec:Away}

We first consider an interesting product, $\Kac{1,3} \fuse \Kac{2,2}$ at $(p,p')=(1,3)$ ($c=-\frac{5}{2}$), which takes us slightly out of our comfort zone near the edge of the corresponding extended Kac table.  The Kac module list given by \eqref{GrFR:K1s} is $\Kac{2,2}$, $\Kac{2,4}$ (since the character of $\Kac{2,0}$ is formally $0$, by \eqref{eq:KacCharSymm}) and we note that $(r',s'')=(2,4)$ is not near the edge --- it fails to satisfy both $r' \le p=1$ and $s'' \le p'=3$.  However, $\Kac{2,3}$ is of corner type, so one might expect that the result is the staggered module $\Stag{2,3}{0,1}$ (with some logarithmic coupling).  It is, at first, surprising that the fusion product is actually found to be the staggered module $\Stag{1,6}{0,1}(-2)$ of \eqref{FR:K13K15}.  However, there is no contradiction as the structures are identical:
\begin{equation}
\parbox[c]{0.3\textwidth}{
\begin{tikzpicture}[->,-stealth', scale = 0.8, transform shape, node distance=1cm]
 \node[main node, label = left:$\frac{1}{2}$] (l1) [] {};
 \node[main node, label = right:$\frac{1}{2}$] (r1) [right=2cm of l1] {};
 \node[main node, label = right:$0$] (r0) [above of=r1] {};
 \node[main node, label = right:$\frac{5}{2}$] (r2) [below of=r1] {};
 \node (tmp) [left of=r0] {};
 \node [above of=tmp] {\scalebox{1.25}{$\ses{\Kac{2,2}}{\Stag{2,3}{0,1}}{\Kac{2,4}}$}};
 \path[]
   (r1) edge node {} (l1)
   (r0) edge node {} (l1)
   (r1) edge node {} (r0)
   (r1) edge node {} (r2)
   (r2) edge node {} (l1);
\end{tikzpicture}
}
\hspace{0.05\textwidth}
\parbox[c]{0.3\textwidth}{
\begin{tikzpicture}  [->,-stealth', scale = 0.8, transform shape, node distance=1cm]
 \node[main node, label = left:$\frac{1}{2}$] (l1) [] {};
 \node[main node, label = right:$\frac{1}{2}$] (r1) [right=2cm of l1] {};
 \node[main node, label = left:$0$] (l0) [above of=l1] {};
 \node[main node, label = right:$\frac{5}{2}$] (r2) [below of=r1] {};
 \node (tmp) [right of=l0] {};
 \node [above of=tmp] {\scalebox{1.25}{$\ses{\Kac{1,5}}{\Stag{1,6}{0,1}}{\Kac{1,7}}$}};
 \path[]
   (r1) edge node {} (l1)
   (l0) edge node {} (l1)
   (r1) edge node {} (l0)
   (r1) edge node {} (r2)
   (r2) edge node {} (l1);
\end{tikzpicture}
\ .}
\end{equation}
The fact that $G_{-1/2}$ annihilates the weight $0$ vector of $\Kac{2,4}$ means that its action on the preimage of this vector in $\Stag{2,3}{0,1}$ will be proportional to the weight $\frac{1}{2}$ vector of the $\Kac{2,2}$ submodule.  The fusion algorithm merely shows that the proportionality constant is not zero.  In this case, the conjectured procedure of the previous section does predict the correct answer.  However, it indicates strongly that a theory of staggered extensions of Kac modules, rather than \hwms{}, will be needed to properly understand the Kac fusion rules away from the edge.

The next example is a much more structurally intricate fusion product for which the conjecture of the previous section fails spectacularly.  Here, we consider $\Kac{2,2} \fuse \Kac{2,4}$ at $(p,p')=(2,4)$ ($c=0$), the result of which is sufficiently complicated that we have not unravelled its full structure.  First, the corresponding character rule suggests that the product should involve the following Kac modules (which take us well away from the edge of the extended Kac table):
\begin{equation} \label{FR:K22xK24}
\parbox[c]{0.75\textwidth}{
\begin{tikzpicture}[->,scale=0.8,transform shape,node distance=1cm,>=stealth',semithick]
  \node[main node, label = left:$0$](1a){};
  \node[] (1aa) [above of =1a] {$\Kac{1,3}$};	
  \node[main node, label = left:$\frac{1}{2}$](1b)[below of =1a]{};
  \path[] (1a) edge (1b);
  \node[main node, label = left:$\frac{1}{2}$](2a)[right=4cm of 1a]{};
  \node[] (2aa) [above of =2a] {$\Kac{1,5}$};	
  \node[main node, label = left:$3$](2b)[below of =2a]{};
  \path[] (2a) edge (2b);
  \node[main node, label = left:$\frac{1}{2}$](3a)[right=4cm of 2a]{};
  \node[] (3aa) [above of =3a] {$\Kac{3,3}$};	
  \node[main node, label = left:$5$](3b)[below of =3a]{};
  \path[] (3a) edge (3b);
  \node[main node, label = above:$0$] (4a) [right=4cm of 3a] {};
  \node[] (22) [above of =4a] {$\Kac{3,5}$};	
  \node[csv, label = left:$\frac{1}{2}$] (4b) [below left of=4a] {};
  \node[main node, label = right:$\frac{3}{2}$] (4c) [below right of=4a] {};
  \node[subsv, label = left:$3$] (4d) [below of=4b] {};
  \node[subsv, label = right:$5$] (4e) [below of=4c] {};
  \node[subsv, label = below:$\frac{21}{2}$] (4f) [below left of=4e] {};
 \path[every node/.style={font=\sffamily\small}]
   (4b) edge (4a)
   (4a) edge (4c)
   (4b) edge (4d)
   (4b) edge (4e)
   (4d) edge (4c)
   (4e) edge (4c)
   (4e) edge (4f)
   (4d) edge (4f);
\end{tikzpicture}
\ .}
\end{equation}
In particular, this implies that the composition factors (including multiplicities) of the fusion product are $2 \: \Irr{0}$, $4 \: \Irr{1/2}$, $\Irr{3/2}$, $2 \: \Irr{3}$, $2 \: \Irr{5}$ and $\Irr{21/2}$.  If the result were the direct sum of these Kac modules, then we would see five linearly independent eigenvectors when computing with the \NGK{} algorithm to depth 0, three of eigenvalue $\frac{1}{2}$ and two of eigenvalue $0$. However, the special subspace of $\Kac{2,2}$ is only four-dimensional, so such a direct sum is ruled out.  Indeed, explicit computation shows that the depth $0$ truncation of the fusion product has conformal weights $0$ and $\frac{1}{2}$, both with multiplicity $2$, and that $L_0$ possesses two rank $2$ Jordan blocks.

Extending to depth $\frac{1}{2}$, we encounter a new feature:  our first rank $3$ Jordan block.  The placement of the $\Irr{3/2}$ factor is now determined by looking at the staggered submodule generated by the copies of $\Irr{0}$ and applying \cref{prop:StagAnn} to rule out one of the two possibilities:
\begin{equation}
\parbox[c]{0.35\textwidth}{
\begin{tikzpicture}[->,node distance=1cm,scale=0.8,transform shape,>=stealth',semithick]
  \node[] (1) {$0$:};
  \node[] (1a) [below of=1] {$\frac{1}{2}$:};
  \node[] (1b) [below of=1a] {$\frac{3}{2}$:};
  \node[sv] (2) [right = 1cm of 1] {};
  \node[sv] (2a) [right = 0.5cm of 1a] {};
  \node[cross] (2b) [right = 1.5cm of 1b] {};
  \path[] (2) edge (2a)
		  (2) edge (2b);
  \node[sv] (3) [right = 2.5cm of 1] {};
  \node[sv] (3a) [right = 2cm of 1a] {};
  \node[sv] (3b) [right = 3cm of 1b] {};
  \path[]   (3) edge (2)
			(3) edge (3a)
			(3) edge (3b)
			(3a) edge (2a);
  \node[sv] (4a) [right = 4cm of 1a] {};
  \path[] (4a) edge [bend left] (2a);
  \node[sv] (5a) [right = 5.5cm of 1a] {};
  \path [] (5a) edge (4a);
\end{tikzpicture}
\ .}
\end{equation}
We have confirmed this placement by computing to depth $\frac{3}{2}$.  However, determining the placement of the $\Irr{3}$ factors using the fusion algorithm is too computationally demanding at present.  We are therefore left with the question of where to place the composition factors $2\:\Irr{3}$, $2\:\Irr{5}$ and $\Irr{21/2}$.  The first thing to note is that $\Irr{21/2}$ can only appear as a descendant of \emph{both} an $\Irr{3}$ and an $\Irr{5}$. We are able to further restrict the possibilities by appealing to staggered module theory.\footnote{Here, we look at the constraints on the existence of the staggered subquotients of this module.  Some of these are simple consequences of \cref{prop:StagAnn}, but more powerful non-existence theorems follow from the \ns{} analogues of \cite[Sec.~7]{RidSta09}.  Similar arguments (in the Virasoro case) may be found in \cite[Sec.~4.2.2]{MorKac15}.}  The Jordan block structure thereby narrows the structural possibilities down to just two (though this is before upwards-pointing arrows are considered):
\begin{equation}
\parbox[c]{0.92\textwidth}{
\begin{tikzpicture}[->,scale =0.8, transform shape, node distance=1cm,>=stealth',semithick]
  \node[] (1) {};
  \node[] (1a) [below of=1] {};
  \node[] (1b) [below of=1a] {};
  \node[] (1c) [below of=1b] {};
  \node[] (1d) [below of=1c] {};
  \node[] (1e) [below of=1d] {};
  \node[sv] (2) [right = 1.5cm of 1] {};
  \node[sv] (2a) [right = 1cm of 1a] {};
  \node[cross] (2b) [right = 2cm of 1b] {};
  \node[sv] (3) [right = 3.5cm of 1] {};
  \node[sv] (3a) [right = 3cm of 1a] {};
  \node[sv] (3b) [right = 4cm of 1b] {};
  \node[sv] (3c) [right = 3cm of 1c] {};
  \node[sv] (3d) [right = 4cm of 1d] {};
  \node[sv] (3e) [right = 3.5cm of 1e] {};
  \node[sv] (4a) [right = 6cm of 1a] {};
  \node[sv] (5a) [right = 7.5cm of 1a] {};
  \node[sv] (5b) [right = 7cm of 1c] {};
  \node[sv] (5c) [right = 8cm of 1d] {};
  \path[] (2) edge (2a)
  		  (2) edge (2b);
  \path[] (3) edge (2)
   (3a) edge (2a)
   (3) edge (3a)
   (3) edge (3b)
   (3a) edge (3c)
   (3a) edge (3d)
   (3b) edge (3c)
   (3b) edge (3d)
   (3d) edge (3e)
   (3c) edge (3e);
  \path[] (4a) edge [bend left] (2a);
  \path[] (5a) edge (4a)
   (5a) edge (5b)
   (5a) edge (5c);
  \node[] (11) [right = 9cm of 1] {:$0$:};
  \node[] (11a) [below of=11] {:$\frac{1}{2}$:};
  \node[] (11b) [below of=11a] {:$\frac{3}{2}$:};
  \node[] (11c) [below of=11b] {:$3$:};
  \node[] (11d) [below of=11c] {:$5$:};
  \node[] (11e) [below of=11d] {:$\frac{21}{2}$:};
  \node[sv] (12) [right = 1.5cm of 11] {};
  \node[sv] (12a) [right = 1cm of 11a] {};
  \node[cross] (12b) [right = 2cm of 11b] {};
  \node[sv] (13) [right = 3.5cm of 11] {};
  \node[sv] (13a) [right = 3cm of 11a] {};
  \node[sv] (13b) [right = 4cm of 11b] {};
  \node[sv] (14a) [right = 5.5cm of 11a] {};
  \node[sv] (14c) [right = 5cm of 11c] {};
  \node[sv] (14d) [right = 6cm of 11d] {};
  \node[sv] (15a) [right = 7.5cm of 11a] {};
  \node[sv] (15c) [right = 7cm of 11c] {};
  \node[sv] (15d) [right = 8cm of 11d] {};
  \node[sv] (15e) [right = 7.5cm of 11e] {};
  \path[] (12) edge (12a)
		  (12) edge (12b);
  \path[] (13) edge (12)
   (13a) edge (12a)
   (13) edge (13a)
   (13) edge (13b);
  \path[] (14a) edge [bend left] (12a)
   (14a) edge (14c)
   (14a) edge (14d);
  \path[] (15a) edge (14a)
   (15a) edge (15c)
   (15a) edge (15d)
   (15d) edge (15e)
   (15c) edge (15e)
   (15c) edge (14c)
   (15d) edge (14d);
\end{tikzpicture}
\ .}
\end{equation}
However, we note that even if we were able to decide between these structures (by computing to greater depths, for example), then it is still not clear if there are additional parameters, generalising the logarithmic couplings of staggered modules, required to completely identify the isomorphism class of this indecomposable module. 

The existence of indecomposable \ns{} modules like these is not unexpected given that similar modules have been investigated over the Virasoro algebra \cite{EbeVir06}.  However, a full understanding of a single example of these modules is still missing (see \cite[Sec.~6.4]{CreLog13} for a recent discussion), hence attempts to further explore these structures here would be misguided.

Nevertheless, an obvious feature of $\Kac{1,3} \fuse \Kac{2,2}$ at $c=0$ is that the Kac modules \eqref{FR:K22xK24} suggested by the character product are related by boundary reflections of \emph{different} orientations.  For example, $\Kac{1,3}$ reflects onto $\Kac{1,5}$ about the boundary type (Ramond) module $\Kac{1,4}$, but it also reflects onto $\Kac{3,3}$ about the boundary module $\Kac{2,3}$.  With this observation, it is not surprising that the procedure conjectured in \cref{sec:Near} does not suffice.  It would be interesting to further investigate whether this procedure gives correct fusion results whenever this multiple orientations phenomenon is absent.  An issue to address here is what happens when the (expected) structure of the fusion product includes non-cyclic Kac modules, meaning that the module is not generated by a single \ssv{}.  Non-cyclicity is generic for Kac modules (see \cref{fig:KacStructures}), but current computational limits mean that our fusion calculations do not shed any light on this issue.

\subsection{Fusion subrings}\label{sec:fundamental}

As a (finitely generated) fusion ring is commutative and associative, it may be presented as a quotient of the polynomial ring over $\ZZ$ in which the indeterminates are the generators.  In this section, we derive conjectures for such presentations of the fusion subrings generated by either $\Kac{3,1}$ or $\Kac{1,3}$, generalising similar conjectures for Virasoro fusion subrings made in \cite{RPpoly08}.

We begin by discussing polynomial presentations for the ring of Kac characters equipped with the Verlinde product of \cref{sec:VerProd}.  In this case, the character product \eqref{GrFR:KxK} exhibits an obvious $\SLA{sl}{2}$ symmetry that suggests a presentation in terms of Chebyshev polynomials.  Indeed, it is easy to check that the subrings generated by $\ch{\Kac{3,1}}$ and $\ch{\Kac{1,3}}$ are isomorphic to $\ZZ[x]$ and $\ZZ[y]$, respectively, and that explicit isomorphisms are given by
\begin{equation} \label{eq:ChIsoPoly}
\ch{\Kac{2i+1,1}} \leftrightarrow \ChebyU{2i}{\tfrac{1}{2} \sqrt{x+1}}, \qquad \ch{\Kac{1,2j+1}} \leftrightarrow \ChebyU{2j}{\tfrac{1}{2} \sqrt{y+1}},
\end{equation}
where $\ChebyPolyU{n}$ denotes the $n$-th Chebyshev polynomial of the second kind.  Despite the square roots, the images of the Kac characters are polynomials of degrees $i$ and $j$, in $x$ and $y$, respectively.

We remark that the subring generated by $\ch{\Kac{3,1}}$ and $\ch{\Kac{1,3}}$ is obviously isomorphic to $\ZZ[x,y]$, but that this does not contain the \ns{} Kac characters with even indices.  To rectify this, one may introduce a new generator $z$ corresponding to $\ch{\Kac{2,2}}$ and show that the full Kac character ring admits the following presentation:
\begin{equation}
\ideal{\ch{\Kac{3,1}}, \ch{\Kac{2,2}}, \ch{\Kac{1,3}}} \cong \frac{\ZZ[x,y,z]}{\ideal{z^2-xy-x-y-1}}.
\end{equation}

Before we can lift \eqref{eq:ChIsoPoly} to presentations of the fusion subrings generated by either $\Kac{3,1}$ or $\Kac{1,3}$, there are a few matters to address.  First, we mention that $p=1$ implies that the fusion subring generated by $\Kac{3,1}$ is semisimple (see \cref{sec:32} below).  In particular, it is isomorphic to the character ring generated by $\ch{\Kac{3,1}}$.  The story is identical for $p'=1$ and $\Kac{1,3}$.  Second, the case $p=2$ also requires a separate treatment because fusing two (non-trivial) elements of $\ideal{\Kac{3,1}}$, such as $\Kac{3,1} \fuse \Kac{3,1}$, never gives $\Kac{3,1}$ as a direct summand.  Again, the story is identical for $p'=2$ and $\Kac{1,3}$.  For simplicity, we will therefore assume that $p\ge3$ or $p'\ge3$ when analysing the fusion subring generated by $\Kac{3,1}$ or $\Kac{1,3}$, respectively.\footnote{The mild pathology exhibited when $p$ or $p'$ is equal to $2$ is a manifestation of the fact that we are ignoring the Ramond sector fusion.  Nevertheless, the results of this section can also be applied in these cases with minor modifications.}

From here on, we will restrict to the fusion subring generated by $\Kac{3,1}$ as the results for the subring generated by $\Kac{1,3}$ then follow by interchanging labels and swapping $p$ for $p'$.  Consider the $n$-fold iterated fusion product $\Kac{3,1}^{\fuse n} \equiv \Kac{3,1} \fuse \cdots \fuse \Kac{3,1}$.  The (conjectured) procedure of \cref{sec:Near} states that for $n \le \lfloor \frac{p-1}{2} \rfloor$, this iterated fusion product will decompose as a direct sum of Kac modules $\Kac{2i+1,1}$, with $i \in \set{0, 1, \ldots, \lfloor \frac{p-1}{2} \rfloor}$.  However, when $n = \lfloor \frac{p+1}{2} \rfloor$, the fusion product will have a staggered module as a direct summand.  Unfortunately, we have not computed any fusion rules in which one of the modules being fused is staggered.

If $p$ is sufficiently large, however, the fusion rules of the staggered modules may be derived using associativity.  For example, if $p \ge 5$ is odd, then
\begin{equation}
\Kac{3,1} \fuse \Kac{p,1} = \Kac{p,1} \oplus \Stag{p,1}{2,0}
\end{equation}
implies that
\begin{equation}
\Kac{3,1} \fuse \Stag{p,1}{2,0} = (\Kac{3,1} \fuse \Kac{3,1} \fuse \Kac{p,1}) \ominus (\Kac{3,1} \fuse \Kac{p,1}) = \Kac{p,1} \oplus (\Kac{5,1} \fuse \Kac{p,1}) = 2 \: \Kac{p,1} \oplus \Stag{p,1}{2,0} \oplus \Stag{p,1}{4,0},
\end{equation}
where we ignore any logarithmic couplings.\footnote{There are further questions to address here concerning whether the coupling $\logcoup{p,1}{2,0}$ obtained, for example, from $\Kac{3,1} \fuse \Kac{p,1}$ is the same as that obtained from $\Kac{5,1} \fuse \Kac{p,1}$.  The evidence that we have collated to date indicates that the answers are always yes.}  A similar calculation shows that if $p \ge 6$ is even, then
\begin{equation}
\Kac{3,1} \fuse \Stag{p,1}{1,0} = 2 \: \Stag{p,1}{1,0} \oplus \Stag{p,1}{3,0}.
\end{equation}
One can also obtain the fusion rules involving other Kac modules and those of staggered modules with one another (for sufficiently large $p$).  Calculations of this nature lead us to conjecture that the indecomposable modules appearing in the fusion subring generated by $\Kac{3,1}$ are precisely as follows:
\begin{equation} \label{eq:SubringSpan}
\ideal{\Kac{3,1}} = \vspn_{\ZZ} \set{\Kac{2i+1,1}, \Kac{mp,1}, \Stag{kp,1}{a,0} \st i=0, 1, \ldots, \lfloor \tfrac{p-3}{2} \rfloor; \ k,m \in \ZZ_+; \ 0 < a < p; \ kp-a = mp = 1 \bmod{2}}.
\end{equation}
Here, we again ignore any logarithmic couplings.  Note that if $p$ is even, the $\Kac{mp,1}$ do not belong to this subring.

Combining the (conjectured) results of \cref{sec:Near} with associativity leads to a general conjecture for the fusion rules of the (unital) subring generated by $\Kac{3,1}$.  This is conveniently encoded in the following polynomial ring structure, generalising that presented above for the corresponding character subring.  In particular, we lift \eqref{eq:ChIsoPoly} from characters to the Kac modules in \eqref{eq:SubringSpan}:
\begin{equation} \label{eq:ModIsoPoly1}
\Kac{2i+1,1} \leftrightarrow \ChebyU{2i}{\tfrac{1}{2} \sqrt{x+1}} \quad \text{(\(i=0, 1, \ldots, \lfloor \tfrac{p-3}{2} \rfloor\)),} \qquad
\Kac{mp,1} \leftrightarrow \ChebyU{mp-1}{\tfrac{1}{2} \sqrt{x+1}} \quad \text{(\(m \in \ZZ_+\), \(mp=1 \bmod{2}\)).}
\end{equation}
To obtain the map for the staggered modules $\Stag{kp,1}{a,0}$ in \eqref{eq:SubringSpan}, we recall that \eqref{eq:ChIsoPoly} maps their characters to
\begin{align}
\ch{\Stag{kp,1}{a,0}} &= \ch{\Kac{kp-a,1}} + \ch{\Kac{kp+a,1}} \notag \\
&\leftrightarrow \ChebyU{kp-1-a}{\tfrac{1}{2} \sqrt{x+1}} + \ChebyU{kp-1+a}{\tfrac{1}{2} \sqrt{x+1}} = 2 \: \ChebyT{a}{\tfrac{1}{2} \sqrt{x+1}} \ChebyU{kp-1}{\tfrac{1}{2} \sqrt{x+1}},
\end{align}
where $\ChebyPolyT{n}$ denotes the $n$-th Chebyshev polynomial of the first kind.  Lifting this to modules, we arrive at
\begin{equation} \label{eq:ModIsoPoly2}
\Stag{kp,1}{a,0} \leftrightarrow 2 \: \ChebyT{a}{\tfrac{1}{2} \sqrt{x+1}} \ChebyU{kp-1}{\tfrac{1}{2} \sqrt{x+1}} \qquad \text{(\(k \in \ZZ_+\), \(0 < a < p\), \(kp-a = 1 \bmod{2}\)).}
\end{equation}
As the degrees of the images in \eqref{eq:ModIsoPoly1} and \eqref{eq:ModIsoPoly2} are $i$, $\frac{1}{2} (mp-1)$ and $\frac{1}{2} (kp-1+a)$, respectively, and as the coefficient of the highest degree term of each image is $1$, we see that these maps define a $\ZZ$-module isomorphism from the fusion subring $\ideal{\Kac{3,1}}$ to $\ZZ[x]$.  Our conjecture for the fusion rules is that this is actually a ring isomorphism.  Similarly, we conjecture that the map from $\ideal{\Kac{1,3}}$ to $\ZZ[y]$, defined by
\begin{equation} \label{eq:FRy}
\begin{aligned}
\Kac{1,2j+1} &\leftrightarrow \ChebyU{2j}{\tfrac{1}{2} \sqrt{y+1}} & &\text{(\(j=0, 1, \ldots, \lfloor \tfrac{p'-3}{2} \rfloor\)),} \\
\Kac{1,np'} &\leftrightarrow \ChebyU{np'-1}{\tfrac{1}{2} \sqrt{y+1}} & &\text{(\(n \in \ZZ_+\), \(np'=1 \bmod{2}\)),} \\
\Stag{1,\ell p'}{0,b} &\leftrightarrow (2 - \delta_{b=0}) \: \ChebyT{b}{\tfrac{1}{2} \sqrt{y+1}} \ChebyU{\ell p'-1}{\tfrac{1}{2} \sqrt{y+1}} & &\text{(\(\ell \in \ZZ_+\), \(0 \le b < p'\), \(\ell p'-b = 1 \bmod{2}\)),}
\end{aligned}
\end{equation}
is a ring isomorphism.

To illustrate these conjectures, we consider how \eqref{eq:FRy} allows us to compute the fusion rule $\Kac{1,3} \fuse \Kac{1,3}$ when $(p,p')=(2,4)$:
\begin{equation}
 \Kac{1,3}\fuse\Kac{1,3} \leftrightarrow yy=1+(y^2-1) \leftrightarrow \Kac{1,1}\oplus\Stag{1,4}{0,1} \qquad \text{(\(c=0\)).}
\end{equation}
The result agrees with \eqref{FR:K13xK13'}, up to the logarithmic coupling of the staggered module.  It is very easy now to predict more complicated fusion rules, as long as the modules being fused belong to either of the subrings considered above.  For example, \eqref{eq:FRy} predicts the $(p,p') = (2,4)$ fusion rule
\begin{align}
 \Stag{1,4}{0,1}\fuse\Stag{1,4}{0,3} &\leftrightarrow (y^2-1)(y^3-2y^2-y+2) \nonumber\\
 &=2(y^2-1)+2(y^4-2y^3-2y^2+2y+1)+(y^5-4y^4+2y^3+6y^2-3y-2) \nonumber\\
 &\leftrightarrow 2 \: \Stag{1,4}{0,1}\oplus 2\: \Stag{1,8}{0,1}\oplus\Stag{1,8}{0,3} \qquad \text{(\(c=0\)).}
\end{align}

We remark that these conjectures may be incorporated into the procedure described in \cref{sec:Near} for fusion near the edge of the extended Kac table.  If either of the modules to be fused is staggered, $\Stag{r,s}{a,b}$ say, then one replaces it by $\Kac{r-a,s-b} \oplus \Kac{r+a,s+b}$ when applying the character decomposition \eqref{GrFR:KxK} in the first step.

\subsection{An exceptional case: $\bm{c=\frac{3}{2}}$}\label{sec:32}

In the exceptional case where $(p,p')=(1,1)$, every Kac module has the islands structure, hence they are all semisimple.  The fusion rules therefore reduce to those for the simple Kac modules and these are exhausted, up to isomorphism, by the $\Kac{1,s}$ (in particular, $\Kac{r,1} \cong \Kac{1,r}$).  According to the conjectured prescription in \cref{sec:Near}, the fusion products of these simple modules always decompose into a direct sum of Kac modules $\Kac{1,s''}$ because corner type modules do not have reflections about boundaries.  It follows that, in this case, the general character rules \eqref{GrFR:KxK} lift to the genuine fusion rules
\begin{equation}
\Kac{r,s} \fuse \Kac{r',s'} = \sideset{}{'} \bigoplus_{r'' = \abs{r-r'}+1}^{r+r'-1} \ \sideset{}{'} \bigoplus_{s'' = \abs{s-s'}+1}^{s+s'-1} \Kac{r'',s''} \qquad \text{(\(c = \tfrac{3}{2}\)),} \label{KxK}
\end{equation}
where the primes on the direct sum symbols indicate that the summation variables increase in steps of two.  The fusion coefficients corresponding to this rule were previously derived in \cite{MilFus02}.

We remark that one can derive this result from the isomorphism $\Kac{r,s}\cong\Kac{s,r}$ (specific to this model), the conjectured fusion rule \eqref{FR:Kr1K1s}, and the definition of corner type Kac modules.  Indeed, we have
\begin{align}
 \Kac{r,s}\fuse\Kac{r',s'}&=\Big(\sideset{}{'} \bigoplus_{\rho= \abs{r-s}+1}^{r+s-1}\Kac{\rho,1}\Big)
 \fuse\Big(\sideset{}{'} \bigoplus_{\sigma=\abs{r'-s'}+1}^{r'+s'-1}\Kac{1,\sigma}\Big)
=\sideset{}{'}\bigoplus_{\rho=\abs{r-s}+1}^{r+s-1}\ \sideset{}{'} \bigoplus_{\sigma=\abs{r'-s'}+1}^{r'+s'-1}\Kac{\rho,\sigma}
\nonumber\\
 &=\sideset{}{'}\bigoplus_{\rho=\abs{r-s}+1}^{r+s-1}\ \sideset{}{'} \bigoplus_{\sigma=\abs{r'-s'}+1}^{r'+s'-1}\
  \sideset{}{'}\bigoplus_{\sigma'=\abs{\rho-\sigma}+1}^{\rho+\sigma-1}\Kac{1,\sigma'}
 =\sideset{}{'} \bigoplus_{r'' = \abs{r-r'}+1}^{r+r'-1} \ \sideset{}{'} \bigoplus_{s'' = \abs{s-s'}+1}^{s+s'-1}\
 \sideset{}{'}\bigoplus_{\sigma''=\abs{r''-s''}+1}^{r''+s''-1}\Kac{1,\sigma''} \nonumber \\
 &= \sideset{}{'} \bigoplus_{r'' = \abs{r-r'}+1}^{r+r'-1} \ \sideset{}{'} \bigoplus_{s'' = \abs{s-s'}+1}^{s+s'-1} \Kac{r'',s''}.
\end{align}
Of course, semisimplicity also proves that the character product \eqref{GrFR:Kr1xK1s} lifts to the fusion rule \eqref{FR:Kr1K1s}, when $c=\frac{3}{2}$. It follows, in particular, that the Kac modules form a closed fusion ring, without the need to introduce any staggered modules.  This does not mean, however, that staggered modules do not exist at $c=\frac{3}{2}$.  For example, one can construct staggered self-extensions of every (simple) \ns{} Kac module $\Kac{1,s}$, with $s \neq 1$, by following the arguments of \cite[Prop~7.5]{RidSta09}.

\section{Discussion and Outlook}

The results presented in this article demonstrate that the paradigm of fusion product computations using the \NGK{} algorithm is as successful for the \ns{} algebra as for the Virasoro algebra.  In particular, our computations have allowed us to formulate several conjectures for Kac module fusion rules in the logarithmic $N=1$ superconformal minimal models.  We have seen that staggered \ns{} modules, upon which $L_0$ has rank $2$ Jordan blocks, are readily encountered and that more complicated indecomposable modules can be generated on which $L_0$ has rank $3$ blocks.  These results provide, among other things, strong evidence for the conjectures made in \cite{PeaLog14} from numerical lattice-theoretic studies.  On the other hand, the reader may have noticed that our results and conjectures bear a striking resemblance to their counterparts for the Virasoro algebra, see \cite{EbeVir06,RasFus07b,RidPer07,MorKac15}.  This accords with expectations, as there are many instances in which the representation theories of the \ns{} and Virasoro algebras mirror one another.

From this point of view, the key advance made in this paper is not so much the explicit results themselves, but rather the development of the formalism that led to them.  This should be understood in the context of the extension of this formalism, to include the Ramond algebra, that will appear in \cite{CanFusII15}.  In this case, the standard module development that led to the logarithmic Verlinde formula will need to be generalised to accommodate supercharacters, ultimately resulting in a fermionic Verlinde formula generalising that of \cite{EhoFus94}.  More tellingly, the \NGK{} fusion algorithm will require significant modifications because the Ramond sector, consisting of twisted modules, forces one to deal with non-integer field-theoretic monodromies.  The required changes have been partially addressed in \cite{GabFus97}, but this work requires refinement (and simplification) before the twisted version of this algorithm can be efficiently coded and utilised.

Let us recall that the representation theory of the Ramond algebra can be significantly more involved than that of the Virasoro or \ns{} algebras.  In particular, the relation $G_0^2 = L_0 - c/24$ means that $G_0$ may act nilpotently, hence possibly non-semisimply, on weight spaces.  Such a non-semisimple $G_0$-action is observed, for example, on the Ramond Verma module of \hw{} $c/24$.  Because it seems that the majority of the difficulties in the Ramond sector may be traced back to this observation, and because nilpotent actions are the norm for the odd generators of the $N>1$ superconformal algebras, it is reasonable to regard the $N=1$ Ramond sector as a non-trivial, but accessible, toy model for the representation theory of these important superalgebras.  This gives a separate motivation for studying logarithmic $N=1$ models:  They give us an idea of what one can expect in the $N>1$ case.  As logarithmic behaviour is now recognised to be generic (rather than pathological) for \cfts{}, it makes sense to look for logarithmic structures when investigating poorly understood representation theories.\footnote{We mention an example of where this looking has paid off:  The longstanding issue of negative ``fusion coefficients'' for fractional level $\AKMA{sl}{2}$ models, see \cite{RidSL208,CreMod12}.}

We mention that such superconformal algebra studies fit rather naturally into a bigger research programme that is currently being actively pursued (see \cite{RidRel15} for further details).  Here, the point is that the superconformal algebras may be naturally constructed as quantum hamiltonian reductions of certain affine Lie superalgebras \cite{KacQua03}.  For the $N=1$ superconformal algebras, the corresponding affine Lie superalgebra is $\AKMSA{osp}{1}{2}$.  The \cfts{} corresponding to these affine superalgebras are expected to be logarithmic, in general, and the relatively tractable affine symmetry suggests that these theories will be very valuable for understanding general logarithmic behaviour.  Unfortunately, little is known about these models (but see \cite{RozQua92,SalGL106,SalSU207,GotWZN07,CreRel11,CreWAl11}).  A more thorough understanding of these affine theories is therefore warranted and we expect that the structural features of the logarithmic superconformal models will aid in this undertaking (and vice-versa).

This work also relates to affine models through the celebrated coset construction of Goddard, Kent and Olive \cite{GKO85,GKO86}.  It is well known that the simple modules of the $N=1$ minimal models arise when considering admissible \hwms{} \cite{KacWak88,KacWak89} of the $\AKMA{sl}{2}$ components of the coset.  However, as proposed in \cite{PRcoset13}, one expects to similarly realise $N=1$ Kac modules, in particular, by extending these considerations to the analogues of Kac modules over $\AKMA{sl}{2}$.  As this approach is also applicable to the infinite hierarchy of extended minimal models \cite{Date86,Date87a,Date87b,Ahn91,Ber97,MDPR14}, this will give important insight into their logarithmic counterparts. These ideas are also currently being actively pursued.

Finally, it would be of interest to explore whether our results admit a W-extended picture, generalising the
situation \cite{GabRat96,FeiLog06,ReaAss07,SemNot07,GabFus09,RasPol09,RasIrr10,WooFus10,RasWKac11,TsuTen12} for the Virasoro logarithmic minimal models. From the lattice \cite{PRR1p08,RPWperc08,RasWLM09}, this would amount to identifying a W-extended vacuum boundary condition for the Neveu-Schwarz algebra and devising the appropriate lattice implementation of fusion of the W-extended representations. Presumably, this would give rise to W-extended Neveu-Schwarz representations whose characters could then be compared with the recent results of Adamovi\'{c} and Milas \cite{AdaMil09,AdaMil08,AdaMil09b}.

\section*{Acknowledgements}

MC is supported by an Australian Postgraduate Award from the Australian Government.  JR was supported by the Australian Research Council under the Future Fellowship scheme, project number FT100100774.  DR's research is supported by the Australian Research Council Discovery Project DP1093910.  The authors thank Pierre Mathieu and Simon Wood for helpful correspondence and discussions.

\appendix

\section{Fusion and the \NGK{} Algorithm} \label{app:Fusion}

In this appendix, we review an algorithmic approach to fusion that was proposed by Nahm \cite{NahQua94} and then significantly generalised by Gaberdiel and Kausch \cite{GabInd96}.  This approach aims to construct the fusion product of two (\voa{}) modules by realising it as a vector space, in fact as a quotient of the vector space tensor product of the modules being fused, upon which the chiral algebra acts through explicitly given coproduct formulae.  These formulae are deduced from fairly straight-forward manipulations involving \opes{} \cite{GabFus94,GabFus94b} and the aforementioned quotient corresponds to imposing the condition that two seemingly different means of arriving at these coproduct formulae actually give identical results.  As we will see, this imposition follows from the mutual locality requirement for conformal fields.

The algorithm that has come to be known as \NGK{} fusion then observes that this tensor product realisation can be consistently truncated so as to explicitly construct only a quotient, preferably a finite-dimensional one, of the fusion product vector space.  It is often the case that one can completely identify a given fusion product by analysing the coproduct action on a sufficiently large truncation.  This allows one to perform fusion calculations explicitly with a computer algebra system.  There are several advantages to this approach to fusion over traditional approaches, the most important of which is that it facilitates the exploration of the new classes of representations that fusion produces.  For example, the result of fusing two simple modules need not be a direct sum of simple modules in general; in fact, the fusion product need not even be \hw{}.  The main disadvantage is that the algorithm amounts to a computationally intensive brute force construction and is therefore not well suited for general fusion calculations or theoretical studies.  There is also the issue of determining when one has determined all of the so-called \emph{spurious states}.  The algorithm itself does not guarantee termination here, so it is important to be able to check this independently.  In this article, we employ a Verlinde formula for this purpose; in some cases, one can instead use correlation function computations to confirm the result \emph{a posteriori}.

Readers who are already familiar with the \NGK{} fusion algorithm, as introduced in \cite{GabInd96}, will find little that is new here, though they may wish to skim what follows in order to familiarise themselves with our notation and nomenclature.  The \ns{} supersymmetry does lead to the appearance of additional signs, as one expects with a graded tensor product, but the ideas and implementation follow the bosonic case.  On the other hand, fusing with a Ramond module requires significant modifications to the algorithm.  We shall defer a detailed discussion of these modifications to \cite{CanFusII15}.

\subsection{Coproduct Formulae} \label{app:Coprod}

Here, we review the derivation, following \cite{GabFus94b}, of the coproduct formulae that define the action of the chiral algebra on the fusion product.  We provide almost all of the details because the generalisation to the twisted sector, necessary for Ramond fusion, is significantly more involved.  In fact, our treatment of the twisted case, which will be reported in \cite{CanFusII15}, simplifies the known formulae \cite{GabFus97} considerably.  We mention that the derivation for negative modes is only valid in the limit when one of the insertion points is sent to zero.  This is not important for practical purposes, but is expected to be relevant to demonstrating that these fusion coproducts define a tensor structure on appropriate categories of vertex operator algebra modules.

Let $\tfunc{S^{(j)}}{z}$ be chiral fields of conformal weight $h^{(j)}$ whose mode expansions take the form
\begin{equation}
\func{S^{(j)}}{z} = \sum_{n \in \ZZ - h^{(j)}} S^{(j)}_n z^{-n-h^{(j)}}.
\end{equation}
For the purposes of this article, the relevant fields are $\func{T}{z}$, with $h=2$, and $\func{G}{z}$, with $h=\frac{3}{2}$.  In what follows, we will omit the index $j$ labelling the chiral field for simplicity.  We also introduce arbitrary fields $\tfunc{\psi_1}{w_1}$ and $\tfunc{\psi_2}{w_2}$ that are local with respect to $\tfunc{S}{z}$:
\begin{equation} \label{eq:UntwistedLocality}
\func{S}{z} \func{\psi_i}{w_i} = \mu_i \func{\psi_i}{w_i} \func{S}{z}.
\end{equation}
Here, $\mu_i \in \CC \setminus \set{0}$ is the mutual locality index of $S$ with $\psi_i$ (this index also depends on $S$; the notation likewise keeps this implicit for simplicity).  

We deduce coproduct formulae for fusion by determining the natural action of the modes $S_n$ on the (radially ordered) products $\tfunc{\psi_1}{w_1} \tfunc{\psi_2}{w_2}$.  This will define an action of the $S_n$ on the tensor product of the corresponding states $\psi_1 \otimes \psi_2$, the latter being interpreted (after quotienting) as a state in the fusion product.  The starting point of the computation is the contour integral
\begin{equation} \label{eq:UntwistedCoprodInt}
\oint_{\Gamma} \inner{\phi}{\func{S}{z} \func{\psi_1}{w_1} \func{\psi_2}{w_2} \Omega} z^{n+h-1} \: \frac{\dd z}{2 \pi \ii},
\end{equation}
where the contour $\Gamma$ encloses $0$, $w_1$ and $w_2$, $\Omega$ is the vacuum, and $\phi$ is an arbitrary spectator state.  We remark that $\phi$ may even depend on other insertion points, noting that because of radial ordering, such insertion points are not enclosed by $\Gamma$.

We will assume that the fields $\tfunc{\psi_1}{w_1}$ and $\tfunc{\psi_2}{w_2}$ correspond to untwisted representations of the chiral algebra, referring to \cite{GabFus97,CanFusII15} for the twisted case.  In other words, their \opes{} with $\tfunc{S}{z}$ are characterised by modes $S_m$ with $m \in \ZZ - h$:
\begin{equation}
\func{S}{z} \func{\psi_i}{w_i} = \sum_{m \in \ZZ - h} \func{(S_m \psi_i)}{w_i} \brac{z-w_i}^{-m-h}.
\end{equation}
Inserting these \opes{} into \eqref{eq:UntwistedCoprodInt}, we see that there are no branch cuts in the integrand.  We may therefore split the contour into three simple contours around each of the (potential) singularities $w_1$, $w_2$ and $0$.  The residue at $w_1$ is computed by substituting the \ope{} of $\tfunc{S}{z}$ and $\tfunc{\psi_1}{w_1}$:
\begin{align} \label{eq:UntwContz=w1}
\oint_{w_1} \inner{\phi}{\func{S}{z} \func{\psi_1}{w_1} \func{\psi_2}{w_2} \Omega} z^{n+h-1} \: \frac{\dd z}{2 \pi \ii} &= \sum_{m \in \ZZ - h} \inner{\phi}{\func{(S_m \psi_1)}{w_1} \func{\psi_2}{w_2} \Omega} \oint_{w_1} \brac{z-w_1}^{-m-h} z^{n+h-1} \: \frac{\dd z}{2 \pi \ii} \notag \\
&= \sum_{m=-h+1}^{\infty} \binom{n+h-1}{m+h-1} w_1^{n-m} \inner{\phi}{\func{(S_m \psi_1)}{w_1} \func{\psi_2}{w_2} \Omega}.
\end{align}
The residue at $w_2$ is similarly computed, though we must first apply \eqref{eq:UntwistedLocality} before we are able to substitute the \ope{} $\tfunc{S}{z} \tfunc{\psi_2}{w_2}$:
\begin{equation} \label{eq:UntwContz=w2}
\oint_{w_2} \inner{\phi}{\func{S}{z} \func{\psi_1}{w_1} \func{\psi_2}{w_2} \Omega} z^{n+h-1} \: \frac{\dd z}{2 \pi \ii} = \mu_1 \sum_{m=-h+1}^{\infty} \binom{n+h-1}{m+h-1} w_2^{n-m} \inner{\phi}{\func{\psi_1}{w_1} \func{(S_m \psi_2)}{w_2} \Omega}.
\end{equation}
Note that \eqref{eq:UntwContz=w2} follows readily from \eqref{eq:UntwContz=w1} upon swapping $w_1$ with $w_2$ and then making the replacements $S_m \psi_1 \to \psi_1$ and $\psi_2 \to \mu_1 S_m \psi_2$.

If $n \ge -h+1$, then the integrand of \eqref{eq:UntwistedCoprodInt} has no pole at $z=0$ and the coproduct formula for $S_n$ is just
\begin{equation}
\parcoproduct{w_1,w_2}{S_n} = \sum_{m=-h+1}^n \binom{n+h-1}{m+h-1} \sqbrac{w_1^{n-m} \brac{S_m \otimes \wun} + \mu_1 w_2^{n-m} \brac{\wun \otimes S_m}} \qquad \text{(\(n \ge -h+1\)).}
\end{equation}
Here, we have extracted the action of $S_m$ on the fields in the correlator so that, for example,
\begin{equation}
\inner{\phi}{\func{(S_m \psi_1)}{w_1} \func{\psi_2}{w_2} \Omega} \lra \brac{S_m \otimes \wun}.
\end{equation}
It should be clear now why we have insisted on the spectator state $\phi$ even though it plays no role whatsoever in the analysis.  If $n \le -h$, then there is a pole at $z=0$ and we can evaluate the corresponding residue by using either \ope{}:
\begin{subequations}
\begin{align}
\oint_0 \inner{\phi}{\func{S}{z} \func{\psi_1}{w_1} \func{\psi_2}{w_2} \Omega} z^{n+h-1} \: \frac{\dd z}{2 \pi \ii} &= \sum_{m \in \ZZ - h} \oint_0 \brac{z-w_1}^{-m-h} z^{n+h-1} \: \frac{\dd z}{2 \pi \ii} \brac{S_m \otimes \wun} \label{eq:InsertOPE1} \\
&= \mu_1 \sum_{m \in \ZZ - h} \oint_0 \brac{z-w_2}^{-m-h} z^{n+h-1} \: \frac{\dd z}{2 \pi \ii} \brac{\wun \otimes S_m}. \label{eq:InsertOPE2}
\end{align}
\end{subequations}
We remark that $z$ is supposed to make a small circle around $0$, while inserting the appropriate \ope{} requires that $z$ be close to either $w_1$ or $w_2$.  We conclude that \eqref{eq:InsertOPE1} is only valid when $w_1$ is close to $0$ whereas \eqref{eq:InsertOPE2} is only valid when $w_2$ is close to $0$.  We will later send $w_1$ or $w_2$ to $0$ to arrive at a simplified form for the fusion coproducts.

Evaluating the integral in \eqref{eq:InsertOPE1}, we see that the $z=0$ contribution to the fusion coproduct $\parcoproduct{w_1,w_2}{S_n}$, with $n \le -h$, takes the form
\begin{multline} \label{eq:UntwContz=0}
\sum_{m \in \ZZ - h} \binom{-m-h}{-n-h} \brac{-w_1}^{n-m} \brac{S_m \otimes \wun} \\
= -\sum_{m=-h+1}^{\infty} \binom{n+h-1}{m+h-1} w_1^{n-m} \brac{S_m \otimes \wun} + \sum_{m=-\infty}^{-h} \binom{-m-h}{-n-h} \brac{-w_1}^{n-m} \brac{S_m \otimes \wun}.
\end{multline}
Here, we have used the binomial coefficient identity
\begin{equation}
\binom{-m-h}{-n-h} = \brac{-1}^{-n-h} \binom{m-n-1}{-n-h} = \brac{-1}^{-n-h} \binom{m-n-1}{m+h-1} = \brac{-1}^{m-n-1} \binom{n+h-1}{m+h-1}
\end{equation}
which is valid when $m \ge -h+1$ and $n \le -h$.  It is clear that the first sum on the \rhs{} of \eqref{eq:UntwContz=0} precisely cancels the contribution \eqref{eq:UntwContz=w1} from $z=w_1$, when $n \le -h$.  Similarly, starting from the evaluation of \eqref{eq:InsertOPE2} shows that the contribution \eqref{eq:UntwContz=w2} from $z=w_2$ is cancelled.

In this way, we arrive at \emph{two} formulae for the fusion coproduct of $S_n$ with $n \le -h$:
\begin{subequations}
\begin{align}
\parNcoproduct{1}{w_1,w_2}{S_n} &= \sum_{m=-\infty}^{-h} \binom{-m-h}{-n-h} \brac{-w_1}^{n-m} \brac{S_m \otimes \wun} + \mu_1 \sum_{m=-h+1}^{\infty} \binom{n+h-1}{m+h-1} w_2^{n-m} \brac{\wun \otimes S_m} \\
\parNcoproduct{2}{w_1,w_2}{S_n} &= \sum_{m=-h+1}^{\infty} \binom{n+h-1}{m+h-1} w_1^{n-m} \brac{S_m \otimes \wun} + \mu_1 \sum_{m=-\infty}^{-h} \binom{-m-h}{-n-h} \brac{-w_2}^{n-m} \brac{\wun \otimes S_m}.
\end{align}
\end{subequations}
Of course, these formulae should coincide in some sense.  However, they are only supposed to be valid when $w_1$ or $w_2$ is small, respectively, so we make the substitutions $w_1 = 0$ and $w_2 = -w$ in the first and $w_1 = w$ and $w_2 = 0$ in the second to obtain
\begin{subequations}
\begin{align}
\parNcoproduct{1}{0,-w}{S_n} &= \brac{S_n \otimes \wun} + \mu_1 \sum_{m=-h+1}^{\infty} \binom{n+h-1}{m+h-1} \brac{-w}^{n-m} \brac{\wun \otimes S_m} \\
\parNcoproduct{2}{w,0}{S_n} &= \sum_{m=-h+1}^{\infty} \binom{n+h-1}{m+h-1} w^{n-m} \brac{S_m \otimes \wun} + \mu_1 \brac{\wun \otimes S_n}.
\end{align}
\end{subequations}
We remark that we can set $w_1=0$ or $w_2=0$ in the first or second coproduct formula above because the binomial coefficient $\tbinom{-m-h}{-n-h}$ vanishes if $m>n$.  We would not have been able to make these simplifications if the residue at $z=0$ did not partially cancel the contribution from $z=w_1$ or $z=w_2$.

To summarise, the fusion coproduct formulae (for untwisted representations) may be expressed as
\begin{subequations} \label{eq:UntwCoprods}
\begin{align}
\parNcoproduct{1}{0,-w}{S_n} &= \brac{S_n \otimes \wun} + \mu_1 \sum_{m=-h+1}^n \binom{n+h-1}{m+h-1} \brac{-w}^{n-m} \brac{\wun \otimes S_m} & &\text{(\(n \ge -h+1\)),} \\
\parNcoproduct{1}{0,-w}{S_{-n}} &= \brac{S_{-n} \otimes \wun} + \mu_1 \sum_{m=-h+1}^{\infty} \binom{m+n-1}{m+h-1} \brac{-1}^{-n+h-1} w^{-n-m} \brac{\wun \otimes S_m} & &\text{(\(n \ge h\))}
\intertext{or as}
\parNcoproduct{2}{w,0}{S_n} &= \sum_{m=-h+1}^n \binom{n+h-1}{m+h-1} w^{n-m} \brac{S_m \otimes \wun} + \mu_1 \brac{\wun \otimes S_n} & &\text{(\(n \ge -h+1\)),} \\
\parNcoproduct{2}{w,0}{S_{-n}} &= \sum_{m=-h+1}^{\infty} \binom{m+n-1}{m+h-1} \brac{-1}^{m+h-1} w^{-n-m} \brac{S_m \otimes \wun} + \mu_1 \brac{\wun \otimes S_{-n}} & &\text{(\(n \ge h\)).}
\end{align}
\end{subequations}
These formulae are related by translation:
\begin{equation} \label{eq:Translation}
\parNcoproduct{1}{0,-w}{S_n} = \parNcoproduct{2}{w,0}{\ee^{wL_{-1}} S_n \ee^{-wL_{-1}}}.
\end{equation}
Using $\comm{L_{-1}}{S_n} = \tbrac{-h+1-n} S_{n-1}$, we derive inductively the following formulae:
\begin{equation}
L_{-1}^k S_n = \sum_{j=0}^k \binom{k}{j} \brac{-h+1-n} \cdots \brac{-h+j-n} S_{n-j} L_{-1}^{k-j} \quad \Ra \quad 
\ee^{wL_{-1}} S_n = \sum_{j=0}^{\infty} \binom{-h+j-n}{j}w^j S_{n-j} \ee^{wL_{-1}}.
\end{equation}
Consequently, one has
\begin{equation}
\parNcoproduct{1}{0,-w}{S_{-n}} = \sum_{j=0}^{\infty} \binom{n-h+j}{j}w^j \parNcoproduct{2}{w,0}{S_{-n-j}} = \sum_{m=n}^{\infty} \binom{m-h}{m-n} w^{m-n} \parNcoproduct{2}{w,0}{S_{-m}}
\end{equation}
which leads to the three \emph{master equations} for untwisted fusion:
\begin{subequations} \label{eq:Master}
\begin{align}
\coproduct{S_n} &= \sum_{m=-h+1}^n \binom{n+h-1}{m+h-1}w^{n-m} \brac{S_m \otimes \wun} + \mu_1 \brac{\wun \otimes S_n} & &\text{(\(n \ge -h+1\))} \label{eq:Master1} \\
\coproduct{S_{-n}} &= \sum_{m=-h+1}^{\infty} \binom{m+n-1}{n-h} \brac{-1}^{m+h-1}w^{-n-m} \brac{S_m \otimes \wun} + \mu_1 \brac{\wun \otimes S_{-n}} & &\text{(\(n \ge h\)),} \label{eq:Master2} \\
S_{-n} \otimes \wun &= \sum_{m=n}^{\infty} \binom{m-h}{n-h} w^{m-n} \coproduct{S_{-m}} + \mu_1 \brac{-1}^{-n+h} \sum_{m=-h+1}^{\infty} \binom{m+n-1}{n-h} w^{-n-m}\brac{\wun \otimes S_m} & &\text{(\(n \ge h\)).} \label{eq:Master3}
\end{align}
\end{subequations}
Here, we let $\coproductsymb$ denote $\parNcoproductsymb{2}{w,0}$ for brevity.  In practice, such as when performing the explicit computations reported in this work, we would set $w$ to $1$ to further simplify these formulae.  However, this masks the natural grading of these formulae by conformal weight in the same way that choosing insertion points masks the conformal grading of \opes{}.

We remark that the interpretation of fusion as a quotient of the (vector space) tensor product is captured in \cref{eq:Translation}.  This formula is actually a requirement imposed by the locality of \opes{} on the coproduct formulae.  Thus, we may define the fusion product (as a vector space) of two chiral algebra modules $M$ and $N$ to be the quotient
\begin{equation} \label{eq:DefFusion}
M \fuse N = \frac{M \otimes_{\CC} N}{\left\langle \brac{\parNcoproduct{1}{0,-w}{S_n} - \parNcoproduct{2}{w,0}{\ee^{wL_{-1}} S_n \ee^{-wL_{-1}}}} \brac{M \otimes_{\CC} N} \right\rangle},
\end{equation}
where the ideal is the sum of the images for all chiral modes $S_n$ (and all $w$).  Of course, this is the same as the intersection of the corresponding kernels.  The point is that $M \fuse N$ is defined to be the \emph{largest} quotient of $M \otimes_{\CC} N$ upon which the coproduct actions coincide.  We expect that this can be interpreted as a universality property for fusion as is imposed in \cite{HuaLog07}.

Finally, we remark that the cancellation between the contributions to the integral \eqref{eq:UntwistedCoprodInt} from the residues at $z=0$ and $z=w_1$ or $z=w_2$ is best understood analytically in terms of the regularity of the integrand at infinity.  Explicitly, if we use the \ope{} $\tfunc{S}{z} \tfunc{\psi_1}{w_1}$, then the sum of the contributions is, for fixed $m \ge -h+1$ and $n \le -h$, proportional to the integral
\begin{align}
\sqbrac{\oint_0 + \oint_{w_1}} \brac{z-w_1}^{-m-h} z^{n+h-1} \: \frac{\dd z}{2 \pi \ii} &= -\oint_{\infty} \brac{z-w_1}^{-m-h} z^{n+h-1} \: \frac{\dd z}{2 \pi \ii} \notag \\
&= \oint_0 \brac{1 - w_1 y}^{-m-h} y^{m-n-1} \: \frac{\dd y}{2 \pi \ii} \notag \\
&= 0,
\end{align}
because $m-n-1 \ge -h+1 + h - 1 = 0$.  In the above derivation, we have chosen to instead derive the cancellation of these contributions algebraically, using binomial coefficient identities, because this method readily generalises to the twisted sector \cite{CanFusII15}, whereas the contour analysis becomes rather more subtle there due to the presence of branch cuts.

\subsection{The \NGK{} Fusion Algorithm} \label{app:NGK}

One important insight \cite{NahQua94,GabInd96} into the definition \eqref{eq:DefFusion} of the fusion product is that it admits consistent truncations which are easier to construct explicitly.  For this, we consider subalgebras $\alg{U}$ of the \uea{} of the chiral algebra.  One important example is that generated by all the chiral modes with index not greater than minus their conformal weight:
\begin{equation}
\spsub{\alg{U}} = \left\langle S^{(j)}_n \st n \le -h^{(j)} \right\rangle.
\end{equation}
For the \ns{} algebra, this subalgebra is generated by the $L_n$, with $n \le -2$, and the $G_j$, with $j \le -\frac{3}{2}$.  The second example, actually a family of examples labelled by $d \in \ZZ$, that we shall employ is that spanned by monomials of weight greater than $d$ in the chiral modes:
\begin{equation}
\alg{U}^d = \vspn \set{S^{(j_1)}_{n_1} S^{(j_2)}_{n_2} \cdots S^{(j_r)}_{n_r} \st r \in \ZZ_{\ge 0}, \ n_1 + n_2 + \cdots + n_r < -d}.
\end{equation}
The integer $d$ will be referred to as the depth.

The ability to consistently truncate the fusion product amounts to the following claim \cite{NahQua94,GabInd96} for the chiral algebra modules $M$ and $N$:
\begin{equation} \label{eq:NGKClaim}
\frac{M \fuse N}{\func{\alg{U}^d}{M \fuse N}} \subseteq \frac{M}{\func{\spsub{\alg{U}}}{M}} \otimes_{\CC} \frac{N}{\func{\alg{U}^d}{N}}.
\end{equation}
The first factor on the \rhs{} defines the \emph{special subspace} $\spsub{M}$ of $M$ and the second factor defines the depth $d$ subspace $N^d$ of $N$ (even though both are defined as quotients).  The claim is therefore succinctly expressed as the inclusion $\brac{M \fuse N}^d \subseteq \spsub{M} \otimes_{\CC} N^d$.  We remark that fusion is commutative, so one may swap the roles of $M$ and $N$ if desired.

The proof of the claim \eqref{eq:NGKClaim} amounts to showing that any $v \otimes w$ representing the \lhs{} may be written as a linear combination of elements of the \rhs{} by using the master coproduct formulae \eqref{eq:Master}.  This is demonstrated through the following algorithm which is applied, at each step, to each term $v \otimes w$ of the result of the previous step:
\begin{itemize}
\item If $v \notin \spsub{M}$, so $v = S_n v'$, with $n \leq -h$, then \eqref{eq:Master3}, perhaps followed by \eqref{eq:Master2}, may be used to replace $v \otimes w$ by a linear combination of terms of the form $v' \otimes w'$ and $S_m v' \otimes w$, where $m > -h$.  Repeat until none of the resulting terms $v \otimes w$ have $v = S_n v'$, with $n \leq -h$.  One may need to take into account relations in $M$ to arrive at terms with $v \in \spsub{M}$.
\item If $w \notin N^d$, so $w = U w'$, where $U$ is a monomial in the chiral modes of weight greater than $d$, then combine $\coproduct{U} = 0$, coming from the \lhs{} of \eqref{eq:NGKClaim}, with \eqref{eq:Master1} and \eqref{eq:Master2} to replace $v \otimes w$ by a linear combination of terms of the form $U' v \otimes w'$, where the $U'$ are monomials in the chiral modes of weight strictly smaller than that of $U$.  Repeat until each of the resulting terms have $w \in N^d$, using relations in $N$.
\item Repeat the above two steps as required.  Termination is guaranteed for modules whose weights are bounded below over a large class of chiral algebras, the \ns{} algebra included, because each step requires that the sum of the weights of the factors in each term strictly decreases.
\end{itemize}
We remark that these steps are also used when computing fusion products to a given depth $d$.  An explicit example illustrating this is detailed in \cref{sec:TheExample}.

The first goal in constructing a depth $d$ fusion product $\brac{M \fuse N}^d$ is to determine the subspace of $\spsub{M} \otimes_{\CC} N^d$ with which it may be identified.  This determination proceeds through the identification of \emph{spurious states} which are actually relations in the tensor product space that are derived from relations in $M$ and $N$ (or even in the chiral algebra).  Specifically, one combines $\coproduct{U} = 0$, for monomials of weight greater than $d$, with these relations; a spurious state arises if reducing the result to an element of $\spsub{M} \otimes_{\CC} N^d$ using the above algorithm does not yield zero identically.  We again refer to \cref{sec:TheExample} for examples of this process.  Quotienting by the spurious states then gives the fusion product to depth $d$.

Once the depth $d$ fusion product has been identified, its structure is analysed by computing the action of the chiral modes $S_n$ with $\abs{n} \le d$ (all other chiral modes must act as the zero operator).  For this, one applies the coproduct formulae \eqref{eq:Master1} and \eqref{eq:Master2} to a basis element of the depth $d$ fusion product, reducing the result to an element of $\spsub{M} \otimes_{\CC} N^d$ using the above algorithm, then to an element of $\brac{M \times N}^d$ by imposing the spurious state relations.  By analysing the structures obtained for various (small) values of $d$, one gets highly non-trivial information about the fusion product itself; in favourable cases, the information obtained is sufficient to completely identify the fusion product.  The example of \cref{sec:TheExample} illustrates such a case.

\section{Staggered Modules and Logarithmic Couplings} \label{app:StagMod}

Staggered modules form a particular class of indecomposable modules upon which the Virasoro zero mode $L_0$ acts non-semisimply.  In a sense, they form the simplest class of such modules and they are responsible for the logarithmic structure of most of the best understood \lcfts{}.  The term was introduced for Virasoro modules in \cite{RohRed96}, where some classification results were reported, shortly after the first examples had been exhibited \cite{GabInd96}.  A full classification of staggered modules over the Virasoro algebra was completed in \cite{RidSta09}.  This notion has been recently developed \cite{CreLog13} for other chiral algebras, though only a few results have been proven at this level of generality.  In this appendix, we discuss the situation for the \ns{} superconformal algebra and explain the results that we use in the course of identifying the fusion products reported in \cref{sec:Results}.

One particularly important feature of staggered modules is that they are generally not determined, up to isomorphism, by the structural diagrams commonly used to depict them.  In the simplest examples over the Virasoro algebra, one requires an additional parameter $\beta$ to fix the isomorphism class completely.  This was first recognised in \cite{GabInd96}, but a general invariant definition of $\beta$ does not seem to have appeared until \cite{RidPer07}, where it was christened the \emph{logarithmic coupling}.  Other commonly used nomenclature for $\beta$ includes ``beta-invariant'' \cite{RidSta09} and ``indecomposability parameter'' \cite{VasInd11}.  Here, we define logarithmic couplings for the \ns{} algebra and describe the methods we use to compute these parameters in this paper.

\subsection{Staggered Modules} \label{app:Stag}

A staggered module is normally defined to be an extension of \hwms{} upon which $L_0$ acts non-semisimply.  Given the results reported here, and in \cite{RasCla11,MorKac15}, we believe that it will be necessary to generalise this to extensions of Kac modules.  However, most of the rigorous results \cite{RohRed96,RidSta09} about staggered modules were proven for extensions of \hwms{} over the Virasoro algebra and we expect that these results will lift to the \ns{} algebra without difficulty.  Moreover, the simplest Kac modules are also \hwms{}, so that many of the staggered modules we explicitly analyse will be of \hw{} type.  We will therefore make use of Virasoro results, lifted to the \ns{} algebra, that apply to \hw{} type staggered modules, deferring rigorous proofs of these lifts to a future publication.

Let us therefore define a \emph{\ns{} staggered module} to be an extension of a Kac module by another Kac module upon which $L_0$ acts non-semisimply.  The short exact sequence characterising the staggered module $\Stag{}{}$ is then
\begin{equation} \label{ses:DefStag}
\dses{\Kac{r,s}}{\iota}{\Stag{}{}}{\pi}{\Kac{\rho,\sigma}},
\end{equation}
where the injection $\iota$ and surjection $\pi$ are \ns{} module homomorphisms.  We will customarily label the staggered module with two sets of indices $\Stag{i,j}{k,\ell}$ such that the Kac quotient is $\Kac{i+k,j+\ell}$ and the Kac submodule is $\Kac{i-k,j-\ell}$.  In the above exact sequence, we would therefore write
\begin{equation}
\Stag{}{} = \Stag{\frac{1}{2} \brac{\rho+r}, \frac{1}{2} \brac{\sigma+s}}{\frac{1}{2} \brac{\rho-r}, \frac{1}{2} \brac{\sigma-s}}.
\end{equation}
We remark that this labelling need not completely specify the staggered module up to isomorphism.  When necessary, we shall append the required additional labels in parentheses; for example, $\Stag{i,j}{k,\ell}(\beta)$.  We also mention that the staggered modules $\Stag{i,j}{k,\ell}$ that we have encountered in our fusion computations all have either $k=0$ or $\ell = 0$.  
However, staggered modules with both $k$ and $\ell$ non-zero do exist.

Most of the staggered modules analysed in this paper have the property that both the Kac submodule and quotient appearing in \eqref{ses:DefStag} are \hwms{} (an example where this is not the case is discussed in \cref{sec:Away}).  In this case, we have the following result.
\begin{prop}
If the Kac modules $K_{r,s}$ and $K_{\rho,\sigma}$ in \eqref{ses:DefStag} are \hw{}, then their minimal conformal weights must satisfy $h_{\rho,\sigma} \ge h_{r,s}$ and $h_{\rho,\sigma} - h_{r,s} \in \frac{1}{2} \ZZ$.
\end{prop}
If this condition on the minimal conformal weights is not met, then there is no staggered module $\Stag{}{}$ making \eqref{ses:DefStag} exact.  The requirement that the Kac modules be \hw{} is necessary as the example discussed in \cref{sec:Away} shows.  In contrast, the following results hold for general staggered modules.
\begin{prop}
The Jordan blocks of $L_0$, acting on $\Stag{}{}$, have rank at most $2$.
\end{prop}
\begin{proof}
Let $v$ belong to a Jordan block for $L_0$ where the (generalised) eigenvalue is $h$.  Then, $\brac{L_0 - h} v$ need not be zero, but $\pi \brac{L_0 - h} v = \brac{L_0 - h} \pi v = 0$, since $\pi v$ belongs to the Kac module $\Kac{\rho,\sigma}$.  By exactness, $\brac{L_0 - h} v = \iota w$ for some $w \in \Kac{r,s}$, hence $\brac{L_0 - h}^2 v = \brac{L_0 - h} \iota w = \iota \brac{L_0 - h} w = 0$.
\end{proof}
It follows from this proposition, and the definition, that the action of $L_0$ on a staggered module always possesses Jordan blocks of rank $2$.
\begin{prop} \label{prop:StagAnn}
Let $w,y \in \Stag{}{}$ be elements of a rank $2$ Jordan block for $L_0$ satisfying $\brac{L_0 - h} y = w$.  If $\pi y$ is annihilated by some $U$ in the \uea{} of the \ns{} algebra, then $Uw = 0$.  In particular, if $\pi y \in \Kac{\rho,\sigma}$ is singular of conformal weight $h$, then $w$ is singular or zero in $\iota (\Kac{r,s})$.
\end{prop}
\begin{proof}
We may assume, without loss of generality, that $U$ is homogeneous, meaning that $\comm{L_0}{U} = -n U$ for some $n \in \ZZ$.  Since $\pi U y = U \pi y = 0$, we have $Uy \in \iota (\Kac{r,s})$ by exactness, hence $Uy$ is an eigenvector of $L_0$ with eigenvalue $h-n$.  Thus, $Uw = U \brac{L_0 - h} y = \brac{L_0 - h + n} Uy = 0$.  The last statement now follows by combining this result for $U = L_0 - h$ with that for $U$ a positive mode.
\end{proof}
Note that if $\Kac{\rho,\sigma}$ is \hw{} with \hws{} $\pi y$ of conformal weight $h$, then $w = \brac{L_0 - h} y$ cannot be $0$ if there are to be any Jordan blocks at all.  It follows that, in this case, $w$ is singular.

\subsection{Logarithmic Couplings} \label{app:LogCoup}

In this section, we assume that the Kac submodule and quotient of each staggered module is \hw{}.  With this restriction, we can follow \cite{RidSta09,CreLog13} in discussing parametrisations of the isomorphism classes of the staggered modules \eqref{ses:DefStag} with $\Kac{r,s}$ and $\Kac{\rho,\sigma}$ fixed.  The parametrisations for general staggered modules are beyond the scope of this paper.

When the \hwss{} of the (\hw{}) modules $\Kac{r,s}$ and $\Kac{\rho,\sigma}$ have equal conformal weight, $h_{r,s} = h_{\rho,\sigma}$, then this sequence determines the staggered module up to isomorphism (assuming that it exists).  This follows easily from the fact that such staggered modules are quotients \cite[Cor.~4.7]{RidSta09} of universal ``Verma-like'' staggered modules.  However, the exact sequence typically does not determine the isomorphism class uniquely when $h_{\rho,\sigma} > h_{r,s}$.  This seems to have been first recognised in the Virasoro examples constructed in \cite{GabInd96}, where an additional parameter $\logcoup{}{}$ was introduced to specify the module structure.  The claim that this parameter determines the isomorphism class was subsequently demonstrated for a class of Verma-like Virasoro staggered modules in \cite{RohRed96}.  Extending this work to general Virasoro staggered modules required a general invariant definition of $\logcoup{}{}$ \cite{RidPer07} and was completed in \cite{RidSta09}.

The staggered modules typically encountered in fusion computations have the following structure:  Let $x$ denote the \hws{} of the submodule $\func{\iota}{\Kac{r,s}}$ and let $w = Ux$ denote its singular descendant of conformal weight $h_{\rho,\sigma}$ (unique up to rescaling).  The existence of $w$ is guaranteed by \cref{prop:StagAnn}, so we may choose $y$ such that $\brac{L_0 - h_{\rho,\sigma}} y = w$.  It follows that $U^{\dag} y$ must be proportional to $x$ and we define \cite{RidPer07} this constant to be the \emph{logarithmic coupling} $\logcoup{}{}$.  The choice of $y$ is generally only unique up to adding elements of $\func{\iota}{\Kac{r,s}}$ (with conformal weight $h_{\rho,\sigma}$); however, these are annihilated by $U^{\dag}$ as $Ux$ is singular.  The logarithmic coupling is therefore independent of such choices once we fix a normalisation for the \sv{} $w$.  For the \ns{} algebra, it is not hard to show that the coefficient of the monomial involving only $G_{-1/2}$ is non-zero \cite{AstStr97}, hence we let
\begin{equation} \label{eq:LogCoupNorm}
w = \brac{G_{-1/2}^{h_{\rho,\sigma} - h_{r,s}} + \cdots} x.
\end{equation}
With the (antilinear) adjoint generated by $L_n^{\dag} = L_{-n}$ and $G_j^{\dag} = G_{-j}$, we see that renormalising $w$ by a factor of $a$ leads to a renormalisation of $\logcoup{}{}$ by a factor of $\abs{a}^2$.  We will often denote the logarithmic coupling of a staggered module $\Stag{i,j}{k,\ell}$, obtained from fusion calculations, by $\logcoup{i,j}{k,\ell}$ for convenience.

Logarithmic couplings are therefore fundamental representation-theoretic quantities.  Their physical interest in \lcft{} lies in the fact that they also appear in the coefficients of certain \opes{} and correlation functions, often accompanied by the factors with logarithmic singularities.  One can compute the value of any logarithmic coupling by explicitly constructing the staggered module to sufficient depth as in the \NGK{} fusion algorithm.  However, this is, computationally, very intensive.  An alternative method proposed in \cite{RidLog07}, but typically limited to staggered modules involving braid type \hwms{}, is to consider staggered modules for which the quotient $\Kac{\rho,\sigma}$ is replaced by its Verma cover.  The desired staggered module may then be realised as a quotient of the Verma-like one if the latter possesses a \sv{} of the appropriate conformal weight.  Checking explicitly for the existence of such a \sv{} fixes $\logcoup{}{}$ (in the braid type case); see \cite{RidSta09} for the proofs.

To the best of our knowledge, the most efficient means of computing the logarithmic coupling in a fusion product is the field-theoretic, though somewhat heuristic, method outlined in \cite{VasInd11} (noting the important clarifications described in \cite[App.~D]{GaiLat13}).  This builds on earlier work \cite{GurCon04} addressing the so-called ``$c \ra 0$ catastrophe'' that is reviewed in \cite{CarLog13,GurLog13}.  The computation for the staggered module appearing in \eqref{ses:DefStag} realises its logarithmic coupling $\logcoup{}{}$ as a limit --- the parameter $t$ is perturbed away from the desired value $\func{t}{0}$, hence the perturbed central charge $c$ and Kac weights $h_{r,s}$ take the form
\begin{equation} \label{eq:Perturb}
\func{t}{\eps} = \func{t}{0} + \eps, \quad 
\func{c}{\eps} = \func{c}{0} - 3 \brac{1 - \frac{1}{\func{t}{0}^2}} \eps + \cdots, \quad 
\func{h_{r,s}}{\eps} = \func{h_{r,s}}{0} - \brac{\frac{r^2-1}{8 \func{t}{0}^2} - \frac{s^2-1}{8}} \eps + \cdots.
\end{equation}
Let $\func{x}{\eps}$ denote a \hws{} of central charge $\func{c}{\eps}$ and conformal weight $\func{h_{r,s}}{\eps}$.  We define $U$ such that $U \func{x}{0}$ is the singular descendant of conformal weight $\func{h_{u,v}}{0}$, normalised as in \eqref{eq:LogCoupNorm}, then let $\func{w}{\eps} = U \func{x}{\eps}$.  Note that $U$ does not depend on $\eps$.  The logarithmic coupling is then given by \cite{VasInd11}
\begin{equation} \label{eq:LogCoupForm}
\logcoup{}{} = -\lim_{\eps \ra 0} \frac{\inner{\func{w}{\eps}}{\func{w}{\eps}}}{\func{h_{u,v}}{\eps} - \func{h_{r,s}}{\eps} - \brac{\func{h_{u,v}}{0} - \func{h_{r,s}}{0}}} = \frac{8 \func{t}{0}^2}{u^2-r^2 - \brac{v^2-s^2} \func{t}{0}^2} \lim_{\eps \ra 0} \frac{\inner{\func{x}{\eps}}{U^{\dag} U \func{x}{\eps}}}{\eps},
\end{equation}
where $\func{x}{\eps}$ is normalised so that $\inner{\func{x}{\eps}}{\func{x}{\eps}} = 1$.  We provide an example illustrating the use of this formula in \cref{sec:TheExample}.

\section{Further explicit fusion computations} \label{app:Results}

In this appendix, we list a selection of the fusion products that we have computed by combining the information provided by the Verlinde formula with explicit \NGK{} calculations and staggered module theory.  As it can be difficult to verify explicitly that a sum is direct (the depths required can be very large), we introduce a symbol ``$\qplus$'' for these fusion rules to indicate that the sum indicated may, or may not, be direct.  We do not indicate parity --- see \cref{sec:Results} for this information.  All logarithmic couplings have been independently verified using \eqref{eq:LogCoupForm}.

\subsection*{$\bm{(p,p')=(1,3)}$ ($\bm{c=-\frac{5}{2}}$)}

\begin{equation}
\begin{gathered}
\begin{aligned}
  \Kac{1,3}\fuse \Kac{1,3} &= \Stag{1,3}{0,2}\oplus \Kac{1,3}, &
  \Kac{3,1}\fuse \Kac{3,1} &= \Kac{1,1} \oplus \Kac{3,1} \oplus \Kac{5,1}, \\ 
  \Kac{1,3}\fuse \Kac{1,5} &= \Kac{1,3}\oplus \Stag{1,6}{0,1}(-2), &
  \Kac{3,1}\fuse \Kac{5,1} &= \Kac{3,1} \oplus \Kac{5,1} \oplus \Kac{7,1}, \\
  \Kac{1,3}\fuse \Kac{1,7} &= \Stag{1,6}{0,1}(-2) \oplus \Kac{1,9}, &
  \Kac{3,1}\fuse \Kac{7,1} &= \Kac{5,1}\oplus \Kac{7,1} \oplus \Kac{9,1}, \\
  \Kac{1,3}\fuse \Kac{1,9} &= \Stag{1,9}{0,2}(-64/9) \oplus \Kac{1,9}, & 
  \Kac{2,2}\fuse \Kac{2,2} &=  \Kac{1,1} \oplus \Kac{1,3} \oplus \Kac{3,1} \oplus \Kac{3,3}, \\
  \Kac{1,3}\fuse \Kac{1,11} &= \Stag{1,12}{0,1}(-64) \oplus \Kac{1,9}, &
  \Kac{2,2}\fuse \Kac{4,2} &= \Kac{3,1} \oplus \Kac{3,3} \oplus \Kac{5,1} \oplus \Kac{5,3},
\end{aligned}
\\
  \Kac{1,5}\fuse \Kac{1,5} = \mathcal{R}_{1,3}^{0,2}\oplus \Kac{1,3} \oplus \Kac{1,7} \oplus \Kac{1,9}.
\end{gathered}
\end{equation}

\medskip

\subsection*{$\bm{(p,p')=(1,5)}$ ($\bm{c=-\frac{81}{10}}$)}

\begin{equation}   
\begin{aligned}
  \Kac{1,3}\fuse \Kac{1,3} &= \Kac{1,1}\oplus \Kac{1,3}\oplus \Kac{1,5}, &
  \Kac{1,5}\fuse \Kac{1,5} &= \Stag{1,5}{0,4} \oplus \Stag{1,5}{0,2} \oplus \Kac{1,5}, \\
  \Kac{1,3}\fuse \Kac{1,5} &= \Stag{1,5}{0,2} \oplus \Kac{1,5}, &
  \Kac{1,5}\fuse \Kac{1,9} &= \Stag{1,10}{0,3}(16/25) \oplus \Kac{1,5} \oplus \Stag{1,10}{0,1}(-4) , \\
  \Kac{1,3}\fuse \Kac{1,7} &= \Kac{1,5}\oplus \Kac{1,7}\oplus \Kac{1,9}, &
  \Kac{1,3}\fuse \Kac{1,9} &= \Kac{1,7} \oplus \Stag{1,10}{0,1}(-4), \\
  \Kac{1,3}\fuse \Kac{1,15} &= \Stag{1,15}{0,2}(-1152/25) \oplus \Kac{1,15}, &
  \Kac{1,3}\fuse \Kac{1,19} &= \Stag{1,20}{0,1}(-336) \oplus \Kac{1,17}.
\end{aligned}
\end{equation}

\medskip

\subsection*{$\bm{(p,p')=(2,4)}$ ($\bm{c=0}$)}

\begin{equation}
\begin{gathered}
\begin{aligned}
  \Kac{1,3}\fuse \Kac{1,3} &= \Kac{1,1}\oplus \Stag{1,4}{0,1}(-1), &
  \Kac{3,1}\fuse \Kac{3,1} &= \Stag{2,1}{1,0}(3/8) \qplus \Kac{5,1}, \\
  \Kac{1,3}\fuse \Kac{1,5} &= \Stag{1,4}{0,1}(-1)\oplus \Kac{1,7}, &
  \Kac{3,1}\fuse \Kac{5,1} &= \Stag{4,1}{1,0}(14175/32)\qplus \Kac{7,1} \\
  \Kac{1,3}\fuse \Kac{1,7} &= \Kac{1,5} \oplus \Stag{1,8}{0,1}(-15), &
  \Kac{3,1}\fuse \Kac{2,2} &= \Kac{2,2} \oplus \Kac{4,2}, \\
  \Kac{1,3}\fuse \Kac{1,11} &= \Kac{1,9} \oplus \Stag{1,12}{0,1}(-3780), &
  \Kac{3,1}\fuse \Kac{2,4} &= \Kac{2,4} \oplus \Kac{4,4}, \\
  \Kac{1,3}\fuse \Kac{2,2} &= \Kac{2,2} \oplus \Kac{2,4}, &
  \Kac{1,5}\fuse \Kac{1,5} &= \Stag{1,4}{0,3}(-1/4) \oplus \Stag{1,4}{0,1}(-1) \oplus \Kac{1,9}, \\
  \Kac{1,3}\fuse \Kac{2,4} &= \Stag{2,4}{0,2} \oplus \Kac{2,4}, & 
  \Kac{1,3}\fuse \Kac{2,6} &= \Kac{2,4} \oplus \Kac{2,6} \oplus \Kac{2,8},
\end{aligned}
\\
  \Kac{2,2}\fuse \Kac{2,2} = \Stag{2,1}{1,0}(3/8) \oplus \Stag{2,3}{1,0}(1/2).
\end{gathered}
\end{equation}

\medskip

\subsection*{$\bm{(p,p')=(2,8)}$ ($\bm{c=-\frac{21}{4}}$)}

\begin{equation}
\begin{gathered}
\begin{aligned}
  \Kac{1,3}\fuse \Kac{1,3} &= \Kac{1,1}\oplus \Kac{1,3}\oplus \Kac{1,5}, &
  \Kac{1,3}\fuse \Kac{2,2} &= \Kac{2,2} \oplus \Kac{2,4}, \\
  \Kac{1,3}\fuse \Kac{1,5} &= \Kac{1,3}\oplus \Kac{1,5}\oplus \Kac{1,7}, &
  \Kac{1,3}\fuse \Kac{2,4} &= \Kac{2,2} \oplus \Kac{2,4} \oplus \Kac{2,6}, \\
  \Kac{1,3}\fuse \Kac{1,7} &= \Stag{1,8}{0,1}(-3)\oplus \Kac{1,5}, &
  \Kac{1,3}\fuse \Kac{2,6} &= \Kac{2,4} \oplus \Kac{2,6} \oplus \Kac{2,8}, \\
  \Kac{1,3}\fuse \Kac{1,15} &=\Kac{1,13} \oplus \Stag{1,16}{0,1}(-165), &
  \Kac{1,3}\fuse \Kac{2,8} &= \Stag{2,8}{0,2} \oplus \Kac{2,8}, \\
  \Kac{1,3}\fuse \Kac{1,23} &= \Kac{1,21} \oplus \Stag{1,24}{0,1}(-163020), &
  \Kac{2,2}\fuse \Kac{2,4} &= \Stag{2,3}{1,0}(-105/256) \oplus \Stag{2,5}{1,0}(-15/64), \\
\end{aligned}
\\
  \Kac{2,2}\fuse \Kac{3,3} = \Kac{2,2}\oplus \Kac{2,4} \oplus \Kac{4,2}\oplus \Kac{4,4}.
\end{gathered}
\end{equation}

\medskip

\subsection*{$\bm{(p,p')=(3,5)}$ ($\bm{c=\frac{7}{10}}$)}

\begin{equation}
\begin{aligned}
  \Kac{1,3}\fuse \Kac{1,3} &= \Kac{1,1}\oplus \Kac{1,3}\oplus \Kac{1,5}, &
  \Kac{3,1}\fuse \Kac{2,2} &= \Stag{3,2}{1,0}(64/125), \\
  \Kac{1,3}\fuse \Kac{1,5} &= \Stag{1,5}{0,2}(-256/675)\oplus \Kac{1,5}, &
  \Kac{3,1}\fuse \Kac{4,2} &= \Stag{3,2}{1,0}(64/125) \oplus \Kac{6,2}, \\
  \Kac{1,3}\fuse \Kac{1,7} &= \Kac{1,5}\oplus \Kac{1,7} \oplus K_{1,9}, &
  \Kac{3,1}\fuse \Kac{3,3} &=  \Stag{3,3}{2,0}(256/1125)\oplus \Kac{3,3}, \\
  \Kac{1,3}\fuse \Kac{1,9} &= \Kac{1,7} \oplus \Stag{1,10}{0,1}(-11264/9), &
  \Kac{3,1}\fuse \Kac{2,4} &= \Stag{3,4}{1,0}(2/5), \\
  \Kac{1,3}\fuse \Kac{2,2} &= \Kac{2,2} \oplus \Kac{2,4}, &
  \Kac{3,1}\fuse \Kac{3,5} &= \Stag{3,5}{2,0} \oplus \Kac{3,5}, \\
  \Kac{1,3}\fuse \Kac{2,4} &= \Kac{2,2} \oplus \Stag{2,5}{0,1}(-2/3), &
  \Kac{2,2}\fuse \Kac{2,2} &= \Kac{1,1} \oplus \Kac{1,3} \oplus \Kac{3,1} \oplus \Kac{3,3}, \\
  \Kac{1,3}\fuse \Kac{2,6} &= \Stag{2,5}{0,1}(-2/3) \oplus \Kac{2,8}, &
  \Kac{2,2}\fuse \Kac{2,4} &= \Kac{1,3} \oplus \Kac{1,5} \oplus \Kac{3,3} \oplus \Kac{3,5}, \\
  \Kac{1,3}\fuse \Kac{3,3} &= \Kac{3,1} \oplus \Kac{3,3} \oplus \Kac{3,5}, &
  \Kac{2,2}\fuse \Kac{3,3} &=  \Stag{3,2}{1,0}(64/125) \oplus \Stag{3,4}{1,0}(2/5) , \\
  \Kac{1,3}\fuse \Kac{3,5} &= \Stag{3,5}{0,2} \oplus \Kac{3,5}, &
  \Kac{2,2}\fuse \Kac{4,2} &= \Kac{3,1} \oplus \Kac{3,3} \oplus \Kac{5,1} \oplus \Kac{5,3}.  
\end{aligned}
\end{equation}

\medskip

\raggedright
\singlespacing

\begin{thebibliography}{100}

\bibitem{FriCon86}
D~Friedan, E~Martinec, and S~Shenker.
\newblock Conformal invariance, supersymmetry and string theory.
\newblock {\em Nucl. Phys.}, B271:93--165, 1986.

\bibitem{FriSup85}
D~Friedan, Z~Qiu, and S~Shenker.
\newblock Superconformal invariance in two dimensions and the tricritical
  {Ising} model.
\newblock {\em Phys. Lett.}, B151:37--43, 1985.

\bibitem{NevFac71}
A~Neveu and J~Schwarz.
\newblock Factorizable dual model of pions.
\newblock {\em Nucl. Phys.}, B31:86--112, 1971.

\bibitem{RamDua71}
P~Ramond.
\newblock Dual theory for free fermions.
\newblock {\em Phys. Rev.}, D3:2415--2418, 1971.

\bibitem{KacCon79}
V~Kac.
\newblock Contravariant form for infinite dimensional algebras and
  superalgebras.
\newblock {\em Lecture Notes in Physics}, 94:441--445, 1979.

\bibitem{MeuHig86}
A~Meurman and A~Rocha-Caridi.
\newblock Highest weight representations of the {Neveu}-{Schwarz} and {Ramond}
  algebras.
\newblock {\em Commun. Math. Phys.}, 107:263--294, 1986.

\bibitem{AstStr97}
A~Astashkevich.
\newblock On the structure of {Verma} modules over {Virasoro} and
  {Neveu}-{Schwarz} algebras.
\newblock {\em Commun. Math. Phys.}, 186:531--562, 1997.
\newblock \textsf{arXiv:\mbox{hep-th}/9511032}.

\bibitem{IohRepI03}
K~Iohara and Y~Koga.
\newblock Representation theory of {Neveu}-{Schwarz} and {Ramond} algebras {I}:
  {Verma} modules.
\newblock {\em Adv. Math.}, 178:1--65, 2003.

\bibitem{GreDua90}
B~Greene and M~Plesser.
\newblock Duality in {Calabi}-{Yau} moduli space.
\newblock {\em Nucl. Phys.}, B338:15--37, 1990.

\bibitem{BeiRev12}
N~Beisert \emph{et al}.
\newblock Review of {AdS}/{CFT} integrability: An overview.
\newblock {\em Lett. Math. Phys.}, 99:3--32, 2012.
\newblock \textsf{arXiv:1012.3982 [\mbox{hep-th}]}.

\bibitem{EguNot11}
T~Eguchi, H~Ooguri, and Y~Tachikawa.
\newblock Notes on the {$K3$} surface and the {Mathieu} group {$M_{24}$}.
\newblock {\em Exper. Math.}, 20:91--96, 2011.
\newblock \textsf{arXiv:1004.0956 [\mbox{hep-th}]}.

\bibitem{CanFusII15}
M~Canagasabey and D~Ridout.
\newblock Fusion rules for the logarithmic {$N=1$} superconformal minimal
  models {II}: including the {Ramond} sector.
\newblock In preparation.

\bibitem{RidLog13}
A~Gainutdinov, D~Ridout, and I~Runkel (Eds).
\newblock Logarithmic conformal field theory.
\newblock {\em J. Phys.}, A46:490301, 2013.

\bibitem{PeaLog14}
P~Pearce, J~Rasmussen, and E~Tartaglia.
\newblock Logarithmic superconformal minimal models.
\newblock {\em J. Stat. Mech.}, 2014:P05001, 2014.
\newblock \textsf{arXiv:1312.6763 [\mbox{hep-th}]}.

\bibitem{KhoLog98}
M~Khorrami, A~Aghamohammadi, and A~Ghezelbash.
\newblock Logarithmic {$N=1$} superconformal field theories.
\newblock {\em Phys. Lett.}, B439:283--288, 1998.
\newblock \textsf{arXiv:\mbox{hep-th}/9803071}.

\bibitem{MavNev03}
N~Mavromatos and R~Szabo.
\newblock The {Neveu}-{Schwarz} and {Ramond} algebras of logarithmic
  superconformal field theory.
\newblock {\em JHEP}, 0301:041, 2003.
\newblock \textsf{arXiv:\mbox{hep-th}/0207273}.

\bibitem{RasLog04}
J~Rasmussen.
\newblock Logarithmic limits of minimal models.
\newblock {\em Nucl. Phys.}, B701:516--528, 2004.
\newblock \textsf{arXiv:hep-th/0405257}.

\bibitem{NagLog05}
J~Nagi.
\newblock Logarithmic primary fields in conformal and superconformal field
  theory.
\newblock {\em Nucl. Phys.}, B722:249--265, 2005.
\newblock \textsf{arXiv:\mbox{hep-th}/0504009}.

\bibitem{AdaMil09}
D~Adamovi\'c and A~Milas.
\newblock The {$N=1$} triplet vertex operator superalgebras.
\newblock {\em Commun. Math. Phys.}, 288:225--270, 2009.
\newblock \textsf{arXiv:0712.0379 [math.QA]}.

\bibitem{RasCla11}
J~Rasmussen.
\newblock Classification of {Kac} representations in the logarithmic minimal
  models {$LM \left( 1,p \right)$}.
\newblock {\em Nucl. Phys.}, B853:404--435, 2011.
\newblock \textsf{arXiv:1012.5190 [\mbox{hep-th}]}.

\bibitem{BusKaz12}
P~Bushlanov, A~Gainutdinov, and I~Tipunin.
\newblock {Kazhdan}-{Lusztig} equivalence and fusion of {Kac} modules in
  {Virasoro} logarithmic models.
\newblock {\em Nucl. Phys.}, B862:232--269, 2012.
\newblock \textsf{arXiv:1102.0271 [\mbox{hep-th}]}.

\bibitem{MorKac15}
A~Morin-Duchesne, J~Rasmussen, and D~Ridout.
\newblock Boundary algebras and {Kac} modules for logarithmic minimal models.
\newblock \textsf{arXiv:1503.07584 [\mbox{hep-th}]}.

\bibitem{PeaLog06}
P~Pearce, J~Rasmussen, and J-B Zuber.
\newblock Logarithmic minimal models.
\newblock {\em J. Stat. Mech.}, 0611:P11017, 2006.
\newblock \textsf{arXiv:\mbox{hep-th}/0607232}.

\bibitem{PR07}
P~Pearce and J~Rasmussen.
\newblock Solvable critical dense polymers.
\newblock {\em J. Stat. Mech.}, 0702:P02015, 2007.
\newblock \textsf{arXiv:hep-th/0610273}.

\bibitem{ReaAss07}
N~Read and H~Saleur.
\newblock Associative-algebraic approach to logarithmic conformal field
  theories.
\newblock {\em Nucl. Phys.}, B777:316--351, 2007.
\newblock \textsf{arXiv:\mbox{hep-th}/0701117}.

\bibitem{RasFus07}
J~Rasmussen and P~Pearce.
\newblock Fusion algebra of critical percolation.
\newblock {\em J. Stat. Mech.}, 0709:P09002, 2007.
\newblock \textsf{arXiv:0706.2716 [\mbox{hep-th}]}.

\bibitem{RasFus07b}
J~Rasmussen and P~Pearce.
\newblock Fusion algebras of logarithmic minimal models.
\newblock {\em J. Phys.}, A40:13711--13734, 2007.
\newblock \textsf{arXiv:0707.3189 [\mbox{hep-th}]}.

\bibitem{GabInd96}
M~Gaberdiel and H~Kausch.
\newblock Indecomposable fusion products.
\newblock {\em Nucl. Phys.}, B477:293--318, 1996.
\newblock \textsf{arXiv:\mbox{hep-th}/9604026}.

\bibitem{EbeVir06}
H~Eberle and M~Flohr.
\newblock {Virasoro} representations and fusion for general augmented minimal
  models.
\newblock {\em J. Phys.}, A39:15245--15286, 2006.
\newblock \textsf{arXiv:\mbox{hep-th}/0604097}.

\bibitem{RidPer07}
P~Mathieu and D~Ridout.
\newblock From percolation to logarithmic conformal field theory.
\newblock {\em Phys. Lett.}, B657:120--129, 2007.
\newblock \textsf{arXiv:0708.0802 [\mbox{hep-th}]}.

\bibitem{RidLog07}
P~Mathieu and D~Ridout.
\newblock Logarithmic {$M \left( 2,p \right)$} minimal models, their
  logarithmic couplings, and duality.
\newblock {\em Nucl. Phys.}, B801:268--295, 2008.
\newblock \textsf{arXiv:0711.3541 [\mbox{hep-th}]}.

\bibitem{RidPer08}
D~Ridout.
\newblock On the percolation {BCFT} and the crossing probability of {Watts}.
\newblock {\em Nucl. Phys.}, B810:503--526, 2009.
\newblock \textsf{arXiv:0808.3530 [\mbox{hep-th}]}.

\bibitem{GabFus09}
M~Gaberdiel, I~Runkel, and S~Wood.
\newblock Fusion rules and boundary conditions in the $c=0$ triplet model.
\newblock {\em J. Phys.}, A42:325403, 2009.
\newblock \textsf{arXiv:0905.0916 [\mbox{hep-th}]}.

\bibitem{EicMin85}
H~Eichenherr.
\newblock Minimal operator algebras in superconformal quantum field theory.
\newblock {\em Phys. Lett.}, B151:26--30, 1985.

\bibitem{SotSta86}
G~Sotkov and M~Stanishkov.
\newblock {$N=1$} superconformal operator product expansions and superfield
  fusion rules.
\newblock {\em Phys. Lett.}, B177:361--367, 1986.

\bibitem{GabFus97}
M~Gaberdiel.
\newblock Fusion of twisted representations.
\newblock {\em Int. J. Mod. Phys.}, A12:5183--5208, 1997.
\newblock \textsf{arXiv:\mbox{hep-th}/9607036}.

\bibitem{IohFus01}
K~Iohara and Y~Koga.
\newblock Fusion algebras for {$N=1$} superconformal field theories through
  coinvariants, {II}: {$N=1$} super-{Virasoro}-symmetry.
\newblock {\em J. Lie Theory}, 11:305--337, 2001.

\bibitem{IohRepII03}
K~Iohara and Y~Koga.
\newblock Representation theory of {Neveu}-{Schwarz} and {Ramond} algebras
  {II}. {Fock} modules.
\newblock {\em Ann. Inst. Fourier (Grenoble)}, 53:1755--1818, 2003.

\bibitem{VerFus88}
E~{Verlinde}.
\newblock Fusion rules and modular transformations in {2D} conformal field
  theory.
\newblock {\em Nucl. Phys.}, B300:360--376, 1988.

\bibitem{HuaVer05}
Y-Z Huang.
\newblock Vertex operator algebras, the {Verlinde} conjecture, and modular
  tensor categories.
\newblock {\em Proc. Natl. Acad. Sci. USA}, 102:5352--5356, 2005.
\newblock \textsf{arXiv:math/0412261 [math.QA]}.

\bibitem{MooPol88}
G~Moore and N~Seiberg.
\newblock Polynomial equations for rational conformal field theories.
\newblock {\em Phys. Lett.}, B212:451--460, 1988.

\bibitem{CreRel11}
T~Creutzig and D~Ridout.
\newblock Relating the archetypes of logarithmic conformal field theory.
\newblock {\em Nucl. Phys.}, B872:348--391, 2013.
\newblock \textsf{arXiv:1107.2135 [\mbox{hep-th}]}.

\bibitem{CreMod12}
T~Creutzig and D~Ridout.
\newblock Modular data and {Verlinde} formulae for fractional level {WZW}
  models i.
\newblock {\em Nucl. Phys.}, B865:83--114, 2012.
\newblock \textsf{arXiv:1205.6513 [\mbox{hep-th}]}.

\bibitem{BabTak12}
A~Babichenko and D~Ridout.
\newblock {Takiff} superalgebras and conformal field theory.
\newblock {\em J. Phys.}, A46:125204, 2013.
\newblock \textsf{arXiv:1210.7094 [\mbox{math-ph}]}.

\bibitem{CreMod13}
T~Creutzig and D~Ridout.
\newblock Modular data and {Verlinde} formulae for fractional level {WZW}
  models {II}.
\newblock {\em Nucl. Phys.}, B875:423--458, 2013.
\newblock \textsf{arXiv:1306.4388 [\mbox{hep-th}]}.

\bibitem{RidMod14}
D~Ridout and S~Wood.
\newblock Modular transformations and {Verlinde} formulae for logarithmic
  $(p_+,p_-)$-models.
\newblock {\em Nucl. Phys.}, B880:175--202, 2014.
\newblock \textsf{arXiv:1310.6479 [\mbox{hep-th}]}.

\bibitem{RidBos14}
D~Ridout and S~Wood.
\newblock Bosonic ghosts at $c=2$ as a logarithmic {CFT}.
\newblock {\em Lett. Math. Phys.}, 105:279--307, 2015.
\newblock \textsf{arXiv:1408.4185 [\mbox{hep-th}]}.

\bibitem{CreLog13}
T~Creutzig and D~Ridout.
\newblock Logarithmic conformal field theory: Beyond an introduction.
\newblock {\em J. Phys.}, A46:494006, 2013.
\newblock \textsf{arXiv:1303.0847 [\mbox{hep-th}]}.

\bibitem{RidVer14}
D~Ridout and S~Wood.
\newblock The {Verlinde} formula in logarithmic {CFT}.
\newblock {\em J. Phys. Conf. Ser.}, 597:012065, 2015.
\newblock \textsf{arXiv:1409.0670 [\mbox{hep-th}]}.

\bibitem{NahQua94}
W~Nahm.
\newblock Quasirational fusion products.
\newblock {\em Int. J. Mod. Phys.}, B8:3693--3702, 1994.
\newblock \textsf{arXiv:\mbox{hep-th}/9402039}.

\bibitem{GurCTh99}
V~Gurarie.
\newblock $c$-theorem for disordered systems.
\newblock {\em Nucl. Phys.}, B546:765--778, 1999.
\newblock \textsf{arXiv:\mbox{cond-mat}/9808063}.

\bibitem{GabFus94b}
M~Gaberdiel.
\newblock Fusion rules of chiral algebras.
\newblock {\em Nucl. Phys.}, B417:130--150, 1994.
\newblock \textsf{arXiv:\mbox{hep-th}/9309105}.

\bibitem{RohRed96}
F~Rohsiepe.
\newblock On reducible but indecomposable representations of the {Virasoro}
  algebra.
\newblock \textsf{arXiv:\mbox{hep-th}/9611160}.

\bibitem{RidSta09}
K~Kyt\"{o}l\"{a} and D~Ridout.
\newblock On staggered indecomposable {Virasoro} modules.
\newblock {\em J. Math. Phys.}, 50:123503, 2009.
\newblock \textsf{arXiv:0905.0108 [\mbox{math-ph}]}.

\bibitem{BerSup85}
M~Bershadsky, V~Knizhnik, and M~Teitelman.
\newblock Superconformal symmetry in two dimensions.
\newblock {\em Phys. Lett.}, B151:31--36, 1985.

\bibitem{PRannecy}
P~Pearce and J~Rasmussen.
\newblock Polymers, percolation and fusion.
\newblock In {\em Proceedings of RAQIS'07}, pages 121--148, Annecy-le-Vieux,
  France, 2007.
\newblock
  \textsf{http://lapth.cnrs.fr/conferences/RAQIS/RAQIS07/proceedings07.pdf}.

\bibitem{FucNon04}
J~Fuchs, S~Hwang, A~Semikhatov, and I~Tipunin.
\newblock Nonsemisimple fusion algebras and the {Verlinde} formula.
\newblock {\em Commun. Math. Phys.}, 247:713--742, 2004.
\newblock \textsf{arXiv:\mbox{hep-th}/0306274}.

\bibitem{FloVer07}
M~Flohr and H~Knuth.
\newblock On {Verlinde}-like formulas in $c \left( p , 1 \right)$ logarithmic
  conformal field theories.
\newblock \textsf{arXiv:0705.0545 [\mbox{math-ph}]}.

\bibitem{GabFro08}
M~Gaberdiel and I~Runkel.
\newblock From boundary to bulk in logarithmic {CFT}.
\newblock {\em J. Phys.}, A41:075402, 2008.
\newblock \textsf{arXiv:0707.0388 [\mbox{hep-th}]}.

\bibitem{GaiRad09}
A~Gainutdinov and I~Tipunin.
\newblock {Radford}, {Drinfeld} and {Cardy} boundary states in $(1,p)$
  logarithmic conformal field models.
\newblock {\em J. Phys.}, A42:315207, 2009.
\newblock \textsf{arXiv:0711.3430 [\mbox{hep-th}]}.

\bibitem{PeaGro10}
P~Pearce, J~Rasmussen, and P~Ruelle.
\newblock {Grothendieck} ring and {Verlinde-like} formula for the {W}-extended
  logarithmic minimal model {$WLM \left( 1 , p \right)$}.
\newblock {\em J. Phys.}, A43:045211, 2010.
\newblock \textsf{arXiv:0907.0134 [\mbox{hep-th}]}.

\bibitem{RasVer10}
J~Rasmussen.
\newblock Fusion matrices, generalized {Verlinde} formulas, and partition
  functions in {$\mathcal{WLM}(1,p)$}.
\newblock {\em J. Phys.}, A43:105201, 2010.
\newblock \textsf{arXiv:0908.2014 [\mbox{hep-th}]}.

\bibitem{PeaCos11}
P~Pearce and J~Rasmussen.
\newblock Coset graphs in bulk and boundary logarithmic minimal models.
\newblock {\em Nucl. Phys.}, B846:616--649, 2011.
\newblock \textsf{arXiv:1010.5328 [\mbox{hep-th}]}.

\bibitem{KazTenIV94}
D~Kazhdan and G~Lusztig.
\newblock Tensor structures arising from affine {Lie} algebras. {IV}.
\newblock {\em J. Amer. Math. Soc.}, 7:383--453, 1994.

\bibitem{VasInd11}
R~Vasseur, J~Jacobsen, and H~Saleur.
\newblock Indecomposability parameters in chiral logarithmic conformal field
  theory.
\newblock {\em Nucl. Phys.}, B851:314--345, 2011.
\newblock \textsf{arXiv:1103.3134 [\mbox{hep-th}]}.

\bibitem{RPpoly08}
J~Rasmussen and P~Pearce.
\newblock Polynomial fusion rings of logarithmic minimal models.
\newblock {\em J. Phys.}, 41:175210, 2008.
\newblock \textsf{arXiv:0709.3337 [hep-th]}.

\bibitem{MilFus02}
A~Milas.
\newblock Fusion rings for degenerate minimal models.
\newblock {\em J. Alg.}, 254:300--335, 2002.
\newblock \textsf{arXiv:math/0003225}.

\bibitem{EhoFus94}
W~Eholzer and R~H\"{u}bel.
\newblock Fusion algebras of fermionic rational conformal field theories via a
  generalized {Verlinde} formula.
\newblock {\em Nucl. Phys.}, B414:348--378, 1994.
\newblock \textsf{arXiv:\mbox{hep-th}/9307031}.

\bibitem{RidSL208}
D~Ridout.
\newblock $\widehat{\mathfrak{sl}} \left( 2 \right)_{-1/2}$: A case study.
\newblock {\em Nucl. Phys.}, B814:485--521, 2009.
\newblock \textsf{arXiv:0810.3532 [\mbox{hep-th}]}.

\bibitem{RidRel15}
D~Ridout and S~Wood.
\newblock Relaxed singular vectors, {Jack} symmetric functions and fractional
  level $\widehat{\mathfrak{sl}} \left( 2 \right)$ models.
\newblock {\em Nucl. Phys.}, B894:621--664, 2015.
\newblock \textsf{arXiv:1501.07318 [\mbox{hep-th}]}.

\bibitem{KacQua03}
V~Kac, S~Roan, and M~Wakimoto.
\newblock Quantum reduction for affine superalgebras.
\newblock {\em Commun. Math. Phys.}, 241:307--342, 2003.
\newblock \textsf{arXiv:\mbox{math-ph}/0302015}.

\bibitem{RozQua92}
L~Rozansky and H~Saleur.
\newblock Quantum field theory for the multivariable {Alexander}-{Conway}
  polynomial.
\newblock {\em Nucl. Phys.}, B376:461--509, 1992.

\bibitem{SalGL106}
H~Saleur and V~Schomerus.
\newblock The {$GL \left( 1 \mid 1 \right)$} {WZW} model: From supergeometry to
  logarithmic {CFT}.
\newblock {\em Nucl. Phys.}, B734:221--245, 2006.
\newblock \textsf{arXiv:\mbox{hep-th}/0510032}.

\bibitem{SalSU207}
H~Saleur and V~Schomerus.
\newblock On the {$SU \left( 2 \mid 1 \right)$} {WZW} model and its statistical
  mechanics applications.
\newblock {\em Nucl. Phys.}, B775:312--340, 2007.
\newblock \textsf{arXiv:\mbox{hep-th}/0611147}.

\bibitem{GotWZN07}
G~G\"otz, T~Quella, and V~Schomerus.
\newblock The {WZNW} model on {$PSU \left( 1,1 \mid 2 \right)$}.
\newblock {\em JHEP}, 0703:003, 2007.
\newblock \textsf{arXiv:\mbox{hep-th}/0610070}.

\bibitem{CreWAl11}
T~Creutzig and D~Ridout.
\newblock W-algebras extending $\widehat{\mathfrak{gl}} \left( 1 \middle\vert 1
  \right)$.
\newblock {\em Springer Proceedings in Mathematics and Statistics},
  36:349--368, 2011.
\newblock \textsf{arXiv:1111.5049 [\mbox{hep-th}]}.

\bibitem{GKO85}
P~Goddard, A~Kent, and D~Olive.
\newblock Virasoro algebras and coset space models.
\newblock {\em Phys. Lett.}, B152:88--92, 1985.

\bibitem{GKO86}
P~Goddard, A~Kent, and D~Olive.
\newblock Unitary representations of the {Virasoro} and super-{Virasoro}
  algebras.
\newblock {\em Commun. Math. Phys.}, 103:105--119, 1986.

\bibitem{KacWak88}
V~Kac and M~Wakimoto.
\newblock Modular invariant representations of infinite-dimensional {Lie}
  algebras and superalgebras.
\newblock {\em Proc. Natl. Acad. Sci. USA}, 85:4956--4960, 1988.

\bibitem{KacWak89}
V~Kac and M~Wakimoto.
\newblock Classification of modular invariant representations of affine
  algebras.
\newblock {\em Adv. Ser. Math. Phys.}, 7:138--177, 1989.

\bibitem{PRcoset13}
P~Pearce and J~Rasmussen.
\newblock Coset construction of logarithmic minimal models: branching rules and
  branching functions.
\newblock {\em J. Phys.}, A46:355402, 2013.
\newblock \textsf{arXiv:1305.7304 [hep-th]}.

\bibitem{Date86}
E~Date, M~Jimbo, T~Miwa, and M~Okado.
\newblock Fusion of the eight vertex {SOS} model.
\newblock {\em Lett. Math. Phys.}, 12:209--215, 1986.

\bibitem{Date87a}
E~Date, M~Jimbo, T~Miwa, and M~Okado.
\newblock Automorphic properties of local height probabilities for integrable
  solid-on-solid models.
\newblock {\em Phys. Rev.}, B35:2105--2107, 1987.

\bibitem{Date87b}
E~Date, M~Jimbo, A~Kuniba, T~Miwa, and M~Okado.
\newblock Exactly solvable {SOS} models: local height probabilities and theta
  function identities.
\newblock {\em Nucl. Phys.}, B290:231--273, 1987.

\bibitem{Ahn91}
C~Ahn, S-W Chung, and S-H Tye.
\newblock New parafermion, {$SU(2)$} coset and {$N=2$} superconformal field
  theories.
\newblock {\em Nucl. Phys.}, B365:191--240, 1991.

\bibitem{Ber97}
A~Berkovich, B~McCoy, A~Schilling, and S~Warnaar.
\newblock Bailey flows and {Bose-Fermi} identities for the conformal coset
  models {$(A_1^{(1)})_N\times(A_1^{(1)})_{N'}/(A_1^{(1)})_{N+N'}$}.
\newblock {\em Nucl. Phys.}, B499:621--649, 1997.
\newblock \textsf{arXiv:\mbox{hep-th}/9702026}.

\bibitem{MDPR14}
A~Morin-Duchesne, P~Pearce, and J~Rasmussen.
\newblock Fusion hierarchies, {$T$}-systems, and {$Y$}-systems of logarithmic
  minimal models.
\newblock {\em J. Stat. Mech.}, 1405:P05012, 2014.
\newblock \textsf{arXiv:1401.7750 [\mbox{math-ph}]}.

\bibitem{GabRat96}
M~Gaberdiel and H~Kausch.
\newblock A rational logarithmic conformal field theory.
\newblock {\em Phys. Lett.}, B386:131--137, 1996.
\newblock \textsf{arXiv:\mbox{hep-th}/9606050}.

\bibitem{FeiLog06}
B~Feigin, A~Gainutdinov, A~Semikhatov, and I~Tipunin.
\newblock Logarithmic extensions of minimal models: Characters and modular
  transformations.
\newblock {\em Nucl. Phys.}, B757:303--343, 2006.
\newblock \textsf{arXiv:\mbox{hep-th}/0606196}.

\bibitem{SemNot07}
A~Semikhatov.
\newblock A note on the logarithmic $\left( p , p' \right)$ fusion.
\newblock \textsf{arXiv:0710.5157 [\mbox{hep-th}]}.

\bibitem{RasPol09}
J~Rasmussen.
\newblock Polynomial fusion rings of {$\mathcal{W}$}-extended logarithmic
  minimal models.
\newblock {\em J. Math. Phys.}, 50:043512, 2009.
\newblock \textsf{arXiv:0812.1070 [\mbox{hep-th}]}.

\bibitem{RasIrr10}
J~Rasmussen.
\newblock Fusion of irreducible modules in {$\mathcal{WLM}(p,p')$}.
\newblock {\em J. Phys.}, A43:045210, 2010.
\newblock \textsf{arXiv:0906.5414 [\mbox{hep-th}]}.

\bibitem{WooFus10}
S~Wood.
\newblock Fusion rules of the {$W \left( p,q \right)$} triplet models.
\newblock {\em J. Phys.}, A43:045212, 2010.
\newblock \textsf{arXiv:0907.4421 [\mbox{hep-th}]}.

\bibitem{RasWKac11}
J~Rasmussen.
\newblock {$\mathcal{W}$}-extended {Kac} representations and integrable
  boundary conditions in the logarithmic minimal models {$\mathcal{WLM}(1,p)$}.
\newblock {\em J. Phys.}, A44:395205, 2011.
\newblock \textsf{arXiv:1106.4893 [\mbox{hep-th}]}.

\bibitem{TsuTen12}
A~Tsuchiya and S~Wood.
\newblock The tensor structure on the representation category of the
  {$\mathcal{W}_p$} triplet algebra.
\newblock {\em J. Phys.}, A46:445203, 2013.
\newblock \textsf{arXiv:1201.0419 [\mbox{hep-th}]}.

\bibitem{PRR1p08}
P~Pearce, J~Rasmussen, and P~Ruelle.
\newblock Integrable boundary conditions and {$\mathcal{W}$}-extended fusion in
  the logarithmic minimal models {$\mathcal{LM}(1,p)$}.
\newblock {\em J. Phys.}, A41:295201, 2008.
\newblock \textsf{arXiv:0803.0785 [\mbox{hep-th}]}.

\bibitem{RPWperc08}
J~Rasmussen and P~Pearce.
\newblock {$\mathcal{W}$}-extended fusion algebra of critical percolation.
\newblock {\em J. Phys.}, A41:295208, 2008.
\newblock \textsf{arXiv:0804.4335 [\mbox{hep-th}]}.

\bibitem{RasWLM09}
J~Rasmussen.
\newblock {$\mathcal{W}$}-extended logarithmic minimal models.
\newblock {\em Nucl. Phys.}, B807:495--533, 2009.
\newblock \textsf{arXiv:0805.2991 [hep-th]}.

\bibitem{AdaMil08}
D~Adamovi\'c and A~Milas.
\newblock The {$N=1$} triplet vertex operator superalgebras: twisted sector.
\newblock {\em SIGMA}, 4:087, 2008.
\newblock \textsf{arXiv:0806.3560 [\mbox{math.QA}]}.

\bibitem{AdaMil09b}
D~Adamovi\'c and A~Milas.
\newblock Lattice construction of logarithmic modules for certain vertex
  algebras.
\newblock {\em Sel. Math. New Ser.}, 15:535--561, 2009.
\newblock \textsf{arXiv:0902.3417 [\mbox{math.QA}]}.

\bibitem{GabFus94}
M~Gaberdiel.
\newblock Fusion in conformal field theory as the tensor product of the
  symmetry algebra.
\newblock {\em Int. J. Mod. Phys.}, A9:4619--4636, 1994.
\newblock \textsf{arXiv:\mbox{hep-th}/9307183}.

\bibitem{HuaLog07}
Y-Z Huang, J~Lepowsky, and L~Zhang.
\newblock Logarithmic tensor product theory for generalized modules for a
  conformal vertex algebra.
\newblock \textsf{arXiv:0710.2687 [math.QA]}.

\bibitem{GaiLat13}
A~Gainutdinov and R~Vasseur.
\newblock Lattice fusion rules and logarithmic operator product expansions.
\newblock {\em Nucl. Phys.}, B868:223--270, 2013.
\newblock \textsf{arXiv:1203.6289 [\mbox{hep-th}]}.

\bibitem{GurCon04}
V~Gurarie and A~Ludwig.
\newblock Conformal field theory at central charge {$c=0$} and two-dimensional
  critical systems with quenched disorder.
\newblock In M~Shifman, editor, {\em From Fields to Strings: Circumnavigating
  Theoretical Physics. Ian Kogan Memorial Collection}, volume~2, pages
  1384--1440. World Scientific, Singapore, 2005.
\newblock \textsf{arXiv:\mbox{hep-th}/0409105}.

\bibitem{CarLog13}
J~Cardy.
\newblock Logarithmic conformal field theories as limits of ordinary {CFTs} and
  some physical applications.
\newblock {\em J. Phys.}, A46:494001, 2013.
\newblock \textsf{arXiv:1302.4279 [\mbox{cond-mat.stat-mech}]}.

\bibitem{GurLog13}
V~Gurarie.
\newblock Logarithmic operators and logarithmic conformal field theories.
\newblock {\em J. Phys.}, A46:494003, 2013.
\newblock \textsf{arXiv:1303.1113 [\mbox{cond-mat.stat-mech}]}.

\end{thebibliography}

\end{document}